\documentclass[11pt]{article}

\usepackage[active]{srcltx}
\usepackage{hyperref}
\usepackage{amsmath,amssymb}
\usepackage[all]{xy}

 1 

\newtheorem{teor}{Theorem}
\newtheorem{prop}{Proposition}
\newtheorem{corol}{Corollary}
\newtheorem{lem}{Lemma}
\newtheorem{definition}{Definition}

\def\qed{\ifvmode\Realemovelastskip\fi
{\unskip\nobreak\hfil\penalty50\hbox{}\nobreak\hfil \hbox{\vrule
height1.2ex width1.2ex}\parfillskip=0pt \finalhyphendemerits=0
\par\smallskip}}
\def\qedr{\ifvmode\Realemovelastskip\fi
{\unskip\nobreak\hfil\penalty50\hbox{}\nobreak\hfil \hbox{
$\diamond$}\parfillskip=0pt \finalhyphendemerits=0
\par\smallskip}}
\def\ds{\displaystyle}
\newenvironment{proof}{\noindent{\sl Proof:~~~}}{\quad \qed}

\def\beq{\begin{equation}}
\def\eeq{\end{equation}}
\def\bea{\begin{eqnarray}}
\def\eea{\end{eqnarray}}
\def\beann{\begin{eqnarray*}}
\def\eeann{\end{eqnarray*}}
\def\beasn{\begin{sneqnarray}}
\def\eeasn{\end{sneqnarray}}
\def\ben{\begin{enumerate}}
\def\een{\end{enumerate}}
\def\bit{\begin{itemize}}
\def\eit{\end{itemize}}

\def\derpar#1#2{\displaystyle\frac{\partial{#1}}{\partial{#2}}}
\def\derpars#1#2#3{\displaystyle\frac{\partial^2{#1}}{\partial{#2}\partial{#3}}}
\def\restric#1#2{\left.#1\right|_{#2}}

\def\W{{\cal W}}
\def\C{{\cal C}}
\def\P{{\cal P}}

\def\vf{{\mathfrak{X}}}
\def\df{{\mit\Omega}}
\def\Lag{{\cal L}}
\def\Leg{{\cal FL}}

\def\d{{\rm d}}

\def\Nat{\mathbb{N}}

\def\R{\mathbb{R}}

\def\pr{\operatorname{pr}}
\def\Tan{{\rm T}}
\def\Lie{\mathop{\rm L}\nolimits}
\def\inn{\mathop{i}\nolimits}
\def\Cinfty{{\rm C}^\infty}

\def\tabaddress#1{{\small\it\begin{tabular}[t]{c}#1
\\[1.2ex]\end{tabular}}}
\def\qed{\ifvmode\removelastskip\fi
{\unskip\nobreak\hfil\penalty50\hbox{}\nobreak\hfil \hbox{\vrule
height1.2ex width1.2ex}\parfillskip=0pt \finalhyphendemerits=0
\par\smallskip}}

\title{UNIFIED FORMALISM FOR HIGHER-ORDER NON-AUTONOMOUS DYNAMICAL SYSTEMS}

\author{
{\sc  Pedro Daniel Prieto-Mart\'\i nez\thanks{{\bf e}-{\it mail}:
   peredaniel@ma4.upc.edu} }\\
   {\sc Narciso Rom\'an-Roy\thanks{{\bf e}-{\it mail}:
   nrr@ma4.upc.edu}}  \\
   \tabaddress{Departamento de Matem\'atica Aplicada IV.
   Edificio C-3, Campus Norte UPC\\
   C/ Jordi Girona 1. 08034 Barcelona. Spain}}
   \date{November 29, 2011\\
   {\it Revised version\/}: October 22, 2012}

\begin{document}

\maketitle

\pagestyle{myheadings}

\thispagestyle{empty}

\begin{abstract}
This work is devoted to giving a geometric framework for describing
higher-order non-autonomous mechanical systems.
The starting point is to extend the Lagrangian-Hamiltonian unified formalism
of Skinner and Rusk for these kinds of systems,
generalizing previous developments for
higher-order autonomous mechanical systems and
first-order non-autonomous mechanical systems.
Then, we use this unified formulation to derive
the standard Lagrangian and Hamiltonian formalisms,
including the Legendre-Ostrogradsky map and
the Euler-Lagrange and the Hamilton equations,
both for regular and singular systems.
As applications of our model, two examples
of regular and singular physical systems are studied.
\end{abstract}


\vbox{\raggedleft AMS s.\,c.\,(2000): 70H50, 53C80, 37B55}\null
\markright{\rm P.D. Prieto-Mart\'\i nez, N. Rom\'an-Roy:   \sl Unified formalism for higher-order non-autonomous...}

\clearpage

\tableofcontents

\section{Introduction}
\label{section:intro}

Higher-order dynamical systems play a relevant role in certain
branches of theoretical physics, applied mathematics and numerical analysis.
In particular, they appear in theoretical physics,
in the mathematical description of
relativistic particles with spin, string theories,
Hilbert's Lagrangian for gravitation, Podolsky's generalization 
of electromagnetism and others
\cite{art:Banerjee_Mukherjee_Paul10,art:Barcelos_Natividade91_2,
art:Batlle_Gomis_Pons_Roman88,art:Belvedere_Amaral_Lemos_Carvalhaes00,
art:Govaerts92,art:Kuznetsov_Plyushchay94,art:Nesterenko89,art:Plyushchay91,
art:Popescu_Popescu07,art:Ramos_Roca95,art:Zloshchastiev00},
as well as in some problems of fluid mechanics and classical physics
(see, for instance, the example in Section
\ref{subsection:CylindricalBeam} taken from
\cite{book:Benson06,book:Elsgoltz83}), and 
in numerical models arising from the
discretization of first-order dynamical systems that preserve 
their inherent geometric structures
\cite{art:DeLeon_Martin_Santamaria04}.
In these kinds of systems, the dynamics have explicit
dependence on accelerations or higher-order derivatives 
of the generalized coordinates of position.

In recent years, much works has been devoted 
to the development of geometric formalisms for higher-order
mechanics and field theory (see, for instance, 
\cite{art:Aldaya_Azcarraga78_2,art:Aldaya_Azcarraga80,
proc:Cantrijn_Crampin_Sarlet86,art:Carinena_Lopez92,unpub:Crampin_Saunders11,
book:DeLeon_Rodrigues85,proc:Garcia_Munoz83,art:Gracia_Pons_Roman91,
art:Gracia_Pons_Roman92,art:Krupkova00,unpub:Mukherjee_Paul11,
art:Saunders_Crampin90}).
These formulations use higher-order tangent and jet bundles
as the main tool.
In particular, in a recent paper \cite{art:Prieto_Roman11}
a new geometric formulation has been proposed,
which is an extension to higher-order
autonomous mechanical systems of the formalism proposed by
R. Skinner and R. Rusk in his seminal paper \cite{art:Skinner_Rusk83}.
This formulation compresses the Lagrangian and Hamiltonian formalisms
into a single one, originally developed for
first-order autonomous mechanical systems
and later generalized to non-autonomous systems
\cite{art:Barbero_Echeverria_Martin_Munoz_Roman08,art:Cortes_Martinez_Cantrijn02},
control systems \cite{art:Barbero_Echeverria_Martin_Munoz_Roman07},
and first-order classical field theories
(see \cite{art:Roman09} and references therein).
Nevertheless, to our knowledge, there is neither
a complete geometrical description
of the Lagrangian and Hamiltonian formalisms
(partial studies on this subject can be found in
\cite{art:Crasmareanu00,art:deLeon_Martin94,proc:deLeon_Rodrigues87,
proc:DeLeon_Marrero92,art:Krupkova96}),
nor of the Skinner-Rusk unified formalism for
non-autonomous higher-order mechanical systems.

The aim of this work is to fill this gap.
In order to do this, we first develop the Lagrangian-Hamiltonian
unified formalism of Skinner-Rusk for
higher-order non-autonomous mechanical systems,
studying in particular how this formulation enables us to
obtain the generalized Legendre-Ostrogradsky map
connecting the Lagrangian and Hamiltonian formalisms,
as well as the Euler-Lagrange and the Hamilton equations
of motion.
Thus, starting from this unified framework,
we obtain a geometric description for
the Lagrangian and Hamiltonian formalisms for
higher-order non-autonomous mechanical systems.
This study is conducted both for regular
and singular dynamical systems.
Our analysis is performed by using higher-order jet bundles,
since we wish this work to serve as a model
to develop an unambiguous framework
for higher-order classical field theories that complete
previous approaches in this way
\cite{art:Campos_DeLeon_Martin_Vankerschaver09,art:Vitagliano10}.

The paper is organized as follows:
in Section \ref{section:structures}, we review the
geometric structures needed to develop the formalism,
such as the higher-order jet bundles,
the total derivatives and higher-order semisprays.
Section \ref{section:SkinnerRuskform} is devoted to
the geometric formulation of the Skinner-Rusk unified formalism
for higher-order non-autonomous mechanical systems,
including the description of the dynamical equations
using sections and vector fields.
In Sections \ref{section:lagform} and \ref{section:hamform},
we recover the standard Lagrangian and Hamiltonian formalisms,
presenting a complete description of both
for regular and singular systems.
Finally, in Section \ref{section:Examples}, two examples are analyzed;
the first is a regular system which models
the shape of a deformed elastic cylindrical
beam with fixed ends and has applications in Statics 
and other branches of classical physics \cite{book:Benson06,book:Elsgoltz83};
the second is a modification of a singular system
describing a relativistic particle
\cite{art:Plyushchay88,art:Pisarski86,art:Batlle_Gomis_Pons_Roman88,
art:Nesterenko89,art:Prieto_Roman11},
which in our case is subjected to a generic time-dependent potential.
The paper concludes in Section \ref{section:outlook}
with a summary of results and future research, and an appendix
in Section \ref{section:particular_situation}
where the particular situation of higher-order trivial bundles
is briefly analyzed.

All the manifolds are real, second countable and $\Cinfty$.
The maps and the structures are assumed to be $\Cinfty$. 
Sum over repeated indices is understood.


\section{Geometric structures of higher-order jet bundles over $\R$}
\label{section:structures}

\subsection{Higher-order jet bundles over $\R$}

(See \cite{book:DeLeon_Rodrigues85,book:Saunders89} for details).

Let $E \stackrel{\pi}{\longrightarrow} \R$ be a bundle ($\dim E = n+1$),
and let $\eta \in \df^1(\R)$ be the canonical volume form in $\R$.
If $k\in\Nat$, the \textsl{$k$th order jet bundle} of the projection
$\pi$, $J^{k}\pi$, is the $((k+1)n+1)$-dimensional
manifold of the $k$-jets of sections $\phi \in \Gamma(\pi)$.
A point in $J^{k}\pi$ is denoted by $j^{k}\phi$, where $\phi \in \Gamma(\pi)$
is any representative of the equivalence class.
We have the following natural projections: if $r \leqslant k$,
$$
\begin{array}{rcclcrccl}
\pi^k_r \colon & J^{k}\pi & \longrightarrow & J^r\pi & , &
\pi^k \colon  & J^{k}\pi & \longrightarrow & E \\
\ & j^k\phi & \longmapsto & j^r\phi & , &
\ & j^k\phi & \longmapsto & \phi \ ,
\end{array}
$$
Notice that $\pi^{k}_{0} = \pi^{k}$, where $J^0\pi$ is canonically identified with
$E$, and $\pi^{k}_{k} = {\rm Id}_{J^{k}\pi}$.
Furthermore, we denote $\bar{\pi}^k = \pi \circ \pi^k \colon J^{k}\pi \to \R$.

Local coordinates in $J^{k}\pi$ are constructed as follows:
let $t$ be the global coordinate in $\R$ such that $\eta = \d t$, and $(t,q^A)$,
($1 \leqslant A \leqslant n$),
local coordinates in $E$ adapted to the bundle structure.
Let $\phi \in \Gamma(\pi)$ such that $\phi = (t,\phi^A)$. 
Then, local coordinates in $J^{k}\pi$ are $(t,q^A,q_1^A,\ldots,q_k^A)$, with
$$
q^A = \phi^A \quad , \quad q_i^A = \frac{\d^i\phi^A}{\d t^i} \, .
$$
Usually we write $q_0^A$ instead of $q^A$, and so the local coordinates in $J^{k}\pi$
are written $(t,q_0^A,q_1^A,\ldots,q_k^A)$.

Using these coordinates, the local expression of the natural projections are
$$
\pi^{k}_{r}(t,q_0^A,q_1^A,\ldots,q_k^A) = (t,q_0^A,q_1^A,\ldots,q_r^A) \quad , \quad
\pi^k(t,q_0^A,q_1^A,\ldots,q_k^A) = (t,q_0^A) \ .
$$

If $\phi \in \Gamma(\pi)$ is a section of $\pi$, we denote by $j^{k}\phi$ the canonical
lifting of $\phi$ to $J^{k}\pi$, that is, the map $j^{k}\phi \colon \R \to J^{k}\pi$,
which is a section of the projection $\bar{\pi}^k$.

\textbf{Remark}: We use the same notation for points of $J^{k}\pi$ and
liftings of sections to $J^{k}\pi$, since giving a point
in $J^{k}\pi$ is equivalent to giving the lifting to $J^{k}\pi$ of a section of $\pi$ 
(see \cite{book:Saunders89} for details).


\subsection{Total time derivative}

(See \cite{book:Saunders89} for details).

\begin{definition}
Let $E \stackrel{\pi}{\longrightarrow} \R$ be a bundle, $t_o \in \R$, $\phi \in \Gamma(\pi)$
and $u \in \Tan_{t_o}\R$. The {\rm $k$th holonomic lift} of $u$ by $\phi$ is defined as
$((j^k\phi)_*(u),j_{t_o}^{k+1}\phi) \in (\pi^{k+1}_{k})^*\Tan J^k\pi \, ,$
where $j^{k+1}_{t_o}\phi \equiv (j^{k+1}\phi)(t_o)$.
\end{definition}

In local coordinates, if $u$ is given by
$u = \restric{u_o\derpar{}{t}}{{t_o}}$,
the $k$th holonomic lift of $u$ is given by
\begin{equation}
\label{eqn:LocalCoordHolonomicLiftingTangentVectors}
(j^k\phi)_*u = u_o\left( \restric{\derpar{}{t}}{j^{k}_{t_o}\phi}
+ \sum_{i=0}^{k} \restric{q_{i+1}^A(j^{k+1}_{t_o}\phi)\derpar{}{q_i^A}}{j^{k}_{t_o}\phi} \right) \, .
\end{equation}

The vector space $(\pi^{k+1}_k)^*(\Tan J^{k}\pi)_{j_{t_o}^{k+1}\pi}$
has a canonical splitting as a direct sum, as follows:
$$
(\pi^{k+1}_k)^*(\Tan J^{k}\pi)_{j_{t_o}^{k+1}\pi} = (\pi^{k+1}_{k})^*(V(\bar{\pi}^k))_{j^{k+1}_{t_o}\phi} \oplus
(j^{k}\phi)_*\Tan_{t_o}\R \, ,
$$
where $(j^{k}\phi)_*\Tan_{t_o}\R$ denotes the set of $k$th holonomic lifts
of tangent vectors in $\Tan_{t_o}\R$ by $\phi$.
As a consequence, the vector bundle
$\ds \xymatrix{
(\pi^{k+1}_{k})^*\Tan J^{k}\pi \ar[rr]^-{(\pi^{k+1}_{k})^*\tau_{J^{k}\pi}} & \ &
J^{k}\pi
}$
may be written as the direct sum of two subbundles:
$$
\xymatrix{
(\pi^{k+1}_{k})^*V(\bar{\pi}^k) \oplus H(\pi^{k+1}_{k}) \ar[rr]^-{(\pi^{k+1}_{k})^*\tau_{J^{k}\pi}} & \ &
J^k\pi \, ,
}
$$
where $H(\pi^{k+1}_{k})$ is the union of the fibres $(j^{k}\phi)_*\Tan_t\R$, for $t \in \R$.

Now, if $\vf(\pi^{k+1}_{k})$ denotes the module of vector fields along the projection
$\pi^{k+1}_{k}$, the submodule corresponding to sections of
$\restric{(\pi^{k+1}_{k})^*\tau_{J^{k}\pi}}{(\pi^{k+1}_{k})^*V(\bar{\pi}^k)}$
is denoted by $\vf^v(\pi^{k+1}_{k})$,
and the submodule corresponding to sections of
$\restric{(\pi^{k+1}_{k})^*\tau_{J^{k}\pi}}{H(\pi^{k+1}_{k})}$
is denoted by $\vf^h(\pi^{k+1}_{k})$.
The splitting for the bundles given above induces
the following canonical splitting for the module $\vf(\pi^{k+1}_{k})$:
$$
\vf(\pi^{k+1}_{k}) = \vf^v(\pi^{k+1}_{k}) \oplus \vf^h(\pi^{k+1}_{k}) \, .
$$
An element of the submodule $\vf^h(\pi^{k+1}_{k})$ is called a
\textsl{total derivative}.

\begin{definition}
Given a vector field $X \in \vf(\R)$, a section $\phi \in \Gamma(\pi)$ and a point $t_o \in \R$,
the {\rm $k$th holonomic lift} of $X$ by $\phi$,
$X^k \equiv j^{k}X \in \vf^h(\pi^{k+1}_{k})$, is defined as
$$
X^{k}_{j_{t_o}^{k+1}\phi} = (j^{k}\phi)_*X_{t_o} \, .
$$
\end{definition}

Hence, every vector field $X \in \vf(\R)$ defines a total derivative
given by its holonomic lift.

Alternatively, we have the following characterization of $X^k$
as a derivation: for every $f \in \Cinfty(J^k\pi)$ we have
$$
(d_{X^k}f)(j^{k+1}_{t_o}\phi) = d_X(f \circ j^{k}\phi)({t_o}) \, ,
$$
where $d_{X^k}$ is the derivation associated to $X^k$ and
$d_{X}$ is the derivation corresponding to $X$.

In local coordinates, if $X \in \vf(\R)$ is given by
$\ds X = X_o \derpar{}{t}$,
then, bearing in mind the local expression of the $k$th
holonomic lift for tangent vectors \eqref{eqn:LocalCoordHolonomicLiftingTangentVectors},
the $k$th holonomic lift of $X$ is
$$
X^k = X_o\left( \derpar{}{t} + \sum_{i=0}^{k} q_{i+1}^A\derpar{}{q_i^A}\right) \, .
$$

Finally, the \textsl{total time derivative} is the $k$th holonomic lift of
the coordinate vector field $\partial/\partial t \in \vf(\R)$,
which is denoted by $d_T \in \vf(\pi^{k+1}_{k})$, and
whose local expression is
\begin{equation}
\label{eqn:LocalCoordTotalTimeDerivative}
d_T = \derpar{}{t} + \sum_{i=0}^{k}q_{i+1}^A\derpar{}{q^A_{i}} \, .
\end{equation}

\textbf{Remark}: The usual notation for the total time
derivative is $\d/\d t$, as seen in \cite{book:Saunders89},
while the notation $d_T$ is usually reserved for the same operator
in the autonomous case. Nevertheless, in this paper we use
the same notation for both operators, and the one that is considered
will be understood from the context.


\subsection{Higher-order semisprays. Holonomic sections}

Now we generalize the concept of semispray introduced in
\cite{book:DeLeon_Rodrigues85} to the time-dependent case.

\begin{definition}
A section $\psi \in \Gamma(\bar{\pi}^k)$ is {\rm holonomic of type $r$}, $1 \leqslant r \leqslant k$,
if $j^{k-r+1}\phi = \pi^{k}_{k-r+1} \circ \psi$, 
where $\phi = \pi^{k} \circ \psi \in \Gamma(\pi)$;
that is, the section $\psi$ is the lifting of a section of $\pi$ up to $J^{k-r+1}\pi$.
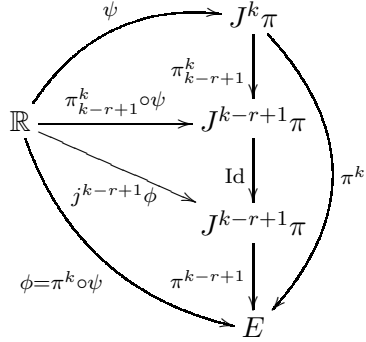
$$
\xymatrix{
\ & \ & J^k\pi \ar[d]_{\pi^k_{k-r+1}} \ar@/^2.5pc/[ddd]^{\pi^k} \\
\R \ar@/^1.5pc/[urr]^{\psi} \ar@/_1.5pc/[ddrr]_{\phi = \pi^k\circ\psi}
\ar[rr]^-{\pi^k_{k-r+1}\circ\psi} \ar[drr]_{j^{k-r+1}\phi} 
& \ & J^{k-r+1}\pi \ar[d]_{{\rm Id}} \\
\ & \ & J^{k-r+1}\pi \ar[d]_{\pi^{k-r+1}} \\
\ & \ & E
}
$$
In particular, a section $\psi$ is {\rm holonomic of type 1}
if, with $\phi = \pi^{k} \circ \psi$, then $j^{k}\phi = \psi$; that is,
$\psi$ is the canonical $k$-jet lifting of a section $\phi \in \Gamma(\pi)$.
Throughout this paper, sections that are
holonomic of type $1$ are simply called {\rm holonomic}.
\end{definition}

\begin{definition}
A vector field $X \in \vf(J^k\pi)$ is a {\rm semispray of type $r$}, $1 \leqslant r \leqslant k$,
if every integral section $\psi$ of $X$ is holonomic of type $r$.
\end{definition}

The local expression of a holonomic section
of type $r$, $\psi \in \Gamma(J^{k}\pi)$, is
$$
\psi(t) = (t,q_0^A,q_1^A,\ldots,q_{k-r+1}^A,\psi_{k-r+2}^A,\ldots,\psi_{k}^A) \, .
$$
Thus, the local expression of a semispray of type $r$ is
$$
X = f\derpar{}{t} + q_1^A\derpar{}{q_0^A} + q_2^A\derpar{}{q_1^A} + \ldots + q_{k-r+1}^A\derpar{}{q_{k-r}^A} + 
X_{k-r+1}^A\derpar{}{q_{k-r+1}^A} + \ldots + X_k^A\derpar{}{q_k^A} \, .
$$

From the local expression, it is clear that
every holonomic section of type $r$ is also holonomic of
type $s$, for $s \geqslant r$. The same remark is true for
semisprays.

We observe that, from the definition, semisprays of type $1$ in
$J^k\pi$ are the analogue to the holonomic vector fields in first-order mechanics; that is,
they are the vector fields whose integral sections (curves) are the canonical liftings
to $J^k\pi$ of sections (curves) on the basis.
Their local expressions are
$$
X = f\derpar{}{t} + q_1^A\derpar{}{q_0^A} + q_2^A\derpar{}{q_1^A} + \ldots + q_k^A\derpar{}{q_{k-1}^A} + X_k^A\derpar{}{q_k^A}\ .
$$

If $X\in\vf(J^k\pi)$ is a semispray of type $r$, a section $\phi \in \Gamma(\pi)$
is said to be a {\sl path} or {\sl solution} of $X$ if
$j^k\phi$ is an integral curve of $X$; that is,
$\widetilde{j^k\phi} = X \circ j^k\phi$,
where $\widetilde{j^k\phi}$ denotes the canonical lifting of $j^k\phi$
from $J^k\pi$ to $\Tan(J^k\pi)$.
Then, in coordinates, $\phi$ verifies the following system of differential equations of order $k+1$:
$$
\frac{\d^{k-r+2}\phi^A}{\d t^{k-r+2}} = X_{k-r+1}^A\left(\phi,\frac{\d\phi}{\d t},\ldots,\frac{\d^k\phi}{\d t^k}\right)
, \ \ldots \ ,
\frac{\d^{k+1}\phi^A}{\d t^{k+1}} = X_k^A\left(\phi,\frac{\d\phi}{\d t},\ldots,\frac{\d^k\phi}{\d t^k}\right) \ .
$$


\section{Skinner-Rusk unified formalism}
\label{section:SkinnerRuskform}


\subsection{Unified phase space. Geometric and dynamical structures}

Consider the configuration bundle $\pi \colon E \to \R$, where
$E$ is an $(n+1)$-dimensional smooth manifold.
Let $\Lag \in \df^{1}(J^{k}\pi)$ be a
\textsl{$k$th order Lagrangian density}, that is, a $\bar{\pi}^k$-semibasic
$1$-form. Thus, we can write $\Lag$ as
$\Lag = L\cdot(\bar{\pi}^k)^*\eta = L\d t \, ,$
where $L \in \Cinfty(J^{k}\pi)$ is the \textsl{$k$th-order Lagrangian function}.

According to \cite{art:Barbero_Echeverria_Martin_Munoz_Roman08,
art:Echeverria_Lopez_Marin_Munoz_Roman04,art:Prieto_Roman11}, we consider the following bundles:
$$
\W = J^{2k-1}\pi \times_{J^{k-1}\pi} \Tan^*(J^{k-1}\pi) \quad ; \quad \W_r = J^{2k-1}\pi \times_{J^{k-1}\pi} \, J^{k-1}\pi^* \, ,
$$
(the fiber product of the above bundles),
where $J^{k-1}\pi^* = \Tan^*(J^{k-1}\pi)/(\bar{\pi}^{k-1})^*\Tan^*\R$.
The bundles $\W$ and $\W_r$ are called the
\textsl{higher-order extended jet-momentum bundle} and the
\textsl{higher-order restricted jet-momentum bundle}, respectively.

\textbf{Comment}:
The reason for taking these bundles is in order
to recover the Lagrangian and Hamiltonian formalisms from
this unified framework, and as we see in Sections
\ref{section:lagform} and \ref{section:hamform},
those formalisms take place in the bundles $J^{2k-1}\pi$
and $J^{k-1}\pi^*$.

These bundles are endowed with the canonical projections
$$
\rho_1 \colon \W \to J^{2k-1}\pi \quad ; \quad
\rho_2 \colon \W \to \Tan^*(J^{k-1}\pi) \quad ; \quad
\rho_{J^{k-1}\pi} \colon \W \to J^{k-1}\pi \quad ; \quad
\rho_\R \colon \W \to \R
$$
$$
\rho_1^r \colon \W_r \to J^{2k-1}\pi \quad ; \quad
\rho_2^r \colon \W_r \to J^{k-1}\pi^* \quad ; \quad
\rho_{J^{k-1}\pi}^r \colon \W_r \to J^{k-1}\pi \quad ; \quad
\rho_\R^r \colon \W_r \to \R
$$

In addition, the natural quotient map $\mu \colon \Tan^*(J^{k-1}\pi) \to J^{k-1}\pi^*$
induces a natural projection (that is, a surjective submersion) $\mu_\W \colon \W \to \W_r$.
Thus, we have the following diagram
$$
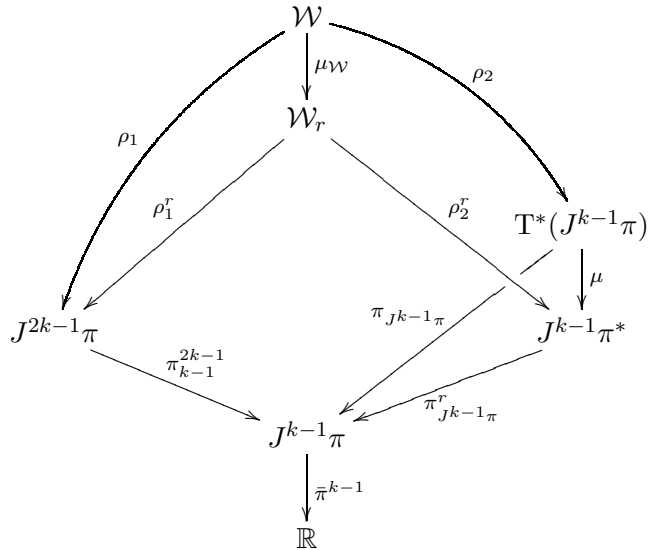
\xymatrix{
\ & \ & \W \ar@/_1.3pc/[llddd]_{\rho_1} \ar[d]^-{\mu_\W} \ar@/^1.3pc/[rrdd]^{\rho_2} & \ & \ \\
\ & \ & \W_r \ar[lldd]_{\rho_1^r} \ar[rrdd]^{\rho_2^r} & \ & \ \\
\ & \ & \ & \ & \Tan^*(J^{k-1}\pi) \ar[d]^-{\mu} \ar[lldd]_{\pi_{J^{k-1}\pi}}|(.25){\hole} \\
J^{2k-1}\pi \ar[rrd]^{\pi^{2k-1}_{k-1}} & \ & \ & \ & J^{k-1}\pi^* \ar[dll]^{\pi_{J^{k-1}\pi}^r} \\
\ & \ & J^{k-1}\pi \ar[d]^{\bar{\pi}^{k-1}} & \ & \ \\
\ & \ & \R & \ & \
}
$$
where $\pi_{J^{k-1}\pi} \colon \Tan^*(J^{k-1}\pi) \to J^{k-1}\pi$ is the canonical
submersion and $\pi_{J^{k-1}\pi}^r \colon J^{k-1}\pi^* \to J^{k-1}\pi$ is the map
satisfying $\pi_{J^{k-1}\pi} = \pi_{J^{k-1}\pi}^r \circ \mu$.

If $(U;t,q_0^A)$ is a local chart of coordinates in $E$, we denote
by $((\pi^{2k-1})^*(U);t,q_0^A,\ldots,q_{2k-1}^A)$ and
$((\pi_{J^{k-1}\pi} \circ \pi^{k-1})^*(U);t,q_0^A,\ldots,q_{k-1}^A,p,p_A^0,\ldots,p_A^{k-1})$
the induced local charts in $J^{2k-1}\pi$ and $\Tan^*(J^{k-1}\pi)$, respectively.
Thus $(t,q_0^A,\ldots,q_{k-1}^A,p_A^0,\ldots,p_A^{k-1})$ are
the natural coordinates in $J^{k-1}\pi^*$, and
the coordinates in $\W$ and $\W_r$ are
$(t,q_0^A,\ldots,q_{k-1}^A,q_k^A,\ldots,q_{2k-1}^A,p,p_A^0,\ldots,p_A^{k-1})$ 
and
$(t,q_0^A,\ldots,q_{k-1}^A,q_k^A,\ldots,q_{2k-1}^A,p_A^0,\ldots,p_A^{k-1})$,
respectively.
Note that $\dim\,\W = 3kn + 2$ and $\dim(\W_r) = 3kn + 1$.

The bundle $\W$ is endowed with some canonical geometric
structures. The first one is:

\begin{definition}
Let $\Theta_{k-1} \in \df^{1}(\Tan^*(J^{k-1}\pi))$ be the tautological
$1$-form, and $\Omega_{k-1} = -\d\Theta_{k-1} \in \df^{2}(\Tan^*(J^{k-1}\pi))$
the canonical symplectic $2$-form on $\Tan^*(J^{k-1}\pi)$.
We define the {\rm higher-order unified canonical forms} as
\begin{equation}\label{eqn:DefPresymplecticForm}
\Theta = \rho_2^*\Theta_{k-1} \in \df^{1}(\W) \quad ; \quad
\Omega = \rho_2^*\Omega_{k-1} \in \df^{2}(\W)\, .
\end{equation}
\end{definition}

Bearing in mind that the local expressions
for the canonical forms on $\Tan^*(J^{k-1}\pi)$ are
\begin{equation}
\label{eqn:LocalCoordCanonicalSymplecticForm}
\Theta_{k-1} = p_A^i \d q_i^A + p\,\d t \quad ; \quad
\Omega_{k-1} = \d q_i^A \wedge \d p_A^i - \d p \wedge \d t \, ,
\end{equation}
the above forms can be written locally as
\begin{equation}
\label{eqn:LocalCoordPresymplecticForm}
\Theta = \rho_2^*(p_A^i \d q_i^A + p\,\d t) = p_A^i \d q_i^A + p\,\d t
\quad ; \quad
\Omega = \rho_2^*(\d q_i^A \wedge \d p_A^i - \d p \wedge \d t)
= \d q_i^A \wedge \d p_A^i - \d p \wedge \d t
\end{equation}

Notice that from the local expressions \eqref{eqn:LocalCoordPresymplecticForm}
we have
\begin{equation}
\label{eqn:LocalBasisKerOmega}
\ker\,\Omega = \left\langle \derpar{}{q_k^A},\ldots,\derpar{}{q_{2k-1}^A} \right\rangle= \vf^{V(\rho_2)}(\W) \ .
\end{equation}
Thus, $\Omega$ is a presymplectic form in $\W$.

The second canonical structure in $\W$ is the following:

\begin{definition}
The {\rm higher-order coupling $1$-form} in $\W$ is the $\rho_\R$-semibasic
$1$-form $\hat{\C} \in \df^{1}(\W)$ defined as follows:
for every $w = (\bar{y},\alpha_q) \in \W$ (that is,
$\alpha_q \in \Tan^*_q (J^{k-1}\pi)$, where $q = \pi^{2k-1}_{k-1}(\bar{y})$ is the projection
of $\bar{y}$ to $J^{k-1}\pi$) and $u \in \Tan_w\W$, then
\begin{equation}
\label{eqn:DefCouplingForm}
\langle \hat{\C}(w) \mid u \rangle = \langle \alpha_q \mid (\Tan_w(j^{k-1}\phi \circ \rho_\R))(u) \rangle \, ,
\end{equation}
where $\phi \in \Gamma(\pi)$ is any representative of $\bar{y}$ (that is, $j^{2k-1}\phi = \bar{y}$).
\end{definition}

$\hat{\C}$ being a $\rho_\R$-semibasic form, there exists $\hat{C} \in \Cinfty(\W)$ such that
$\hat{\C} = \hat{C}\rho_\R^*\eta = \hat{C}\d t$. An easy computation in coordinates
gives the following local expression for the coupling $1$-form:
\begin{equation}
\label{eqn:LocalCoordCouplingForm}
\hat{\C} = (p + p_A^iq_{i+1}^A)\d t \, .
\end{equation}

We denote
$\hat{\Lag} = (\pi^{2k-1}_{k} \circ \rho_1)^*\Lag \in \df^{1}(\W)$.
As the Lagrangian density is a $\bar{\pi}^{k}$-semibasic form, we have
that $\hat{\Lag}$ is a $\rho_\R$-semibasic $1$-form, and thus we can write
$\hat{\Lag} = \hat{L}\rho_\R^*\eta = \hat{L}\d t$, where
$\hat{L} = (\pi^{2k-1}_{k} \circ \rho_1)^*L \in \Cinfty(\W)$ is the
pull-back of the Lagrangian function associated with $\Lag$. Then, we define
a \textsl{Hamiltonian submanifold}
\begin{equation*}
\label{eqn:DefHamiltonianSubmanifold}
\W_o = \left\{ w \in \W \colon \hat{\Lag}(w) = \hat{\C}(w) \right\} \stackrel{j_o}{\hookrightarrow} \W \, .
\end{equation*}
$\hat{\C}$ and $\hat{\Lag}$ being $\rho_\R$-semibasic $1$-forms,
the submanifold $\W_o$ is defined by the constraint $\hat{C} - \hat{L} = 0$.
In natural coordinates, bearing in mind the local expression
\eqref{eqn:LocalCoordCouplingForm} of $\hat{\C}$,
the constraint function is
\begin{equation*}
\hat{C} - \hat{L} = p + p_A^iq_{i+1}^A - \hat{L} = 0 \, .
\end{equation*}
We have the following natural projections in $\W_o$:
\begin{align*}
&\rho_\R^o \colon \W_o \to \R \quad ; \quad
\rho_1^o \colon \W_o \to J^{2k-1}\pi \quad ; \quad
\rho_2^o \colon \W_o \to \Tan^*(J^{k-1}\pi) \\
&\hat{\rho}_2^o = \mu \circ\rho_2^o \colon \W_o \to  J^{k-1}\pi^*
\quad ; \quad \rho_{J^{k-1}\pi}^o \colon \W_o \to J^{k-1}\pi \, .
\end{align*}
Local coordinates in $\W_o$ are
$(t,q_0^A,\ldots,q_{k-1}^A,q_{k}^A,\ldots,q_{2k-1}^A,p_A^0,\ldots,p_A^{k-1})$,
and the local expressions of the above maps are
\begin{align*}
&\rho_1^o(t,q_i^A,q_j^A,p_A^i) = (t,q_i^A,q_j^A) \quad ; \quad
\rho_2^o(t,q_i^A,q_j^A,p_A^i) = (t,q_i^A,\hat{L}-p_A^iq_{i+1}^A,p_A^i) \\
&\hat{\rho}_2^o(t,q_i^A,q_j^A,p_A^i) = (t,q_i^A,p_A^i) \quad ; \quad
j_o(t,q_i^A,q_j^A,p_A^i) = (t,q_i^A,q_j^A,\hat{L}-p_A^iq_{i+1}^A,p_A^i) \, .
\end{align*}

\begin{prop}\label{prop:WoDiffWr}
The submanifold $\W_o \hookrightarrow \W$ is $1$-codimensional, $\mu_\W$-transverse and
diffeomorphic to $\W_r$.
\end{prop}
\begin{proof}
$\W_o$ is obviously $1$-codimensional, since it is defined by one
constraint function.

To see that $\W_o$ is diffeomorphic to $\W_r$, we show that
the smooth map $\mu_\W \circ j_o \colon \W_o \to \W_r$ is one-to-one.
First, for every $(\bar{y},\alpha) \in \W_o$, we have
$L(\pi^{2k-1}_{k}(\bar{y})) = \hat{L}(\bar{y},\alpha) = \hat{C}(\bar{y},\alpha)$,
and
\begin{equation*}
(\mu_\W \circ j_o)(\bar{y},\alpha) = \mu_\W(\bar{y},\alpha) = (\bar{y},\mu(\alpha)) = (\bar{y},[\alpha]) \, .
\end{equation*}

First, $\mu_\W \circ j_o$ is injective; in fact,
let $(\bar{y}_1,\alpha_1), (\bar{y}_2,\alpha_2) \in \W_o$,
then we wish to prove that
$$
(\mu_\W \circ j_o)(\bar{y}_1,\alpha_1) = (\mu_\W \circ j_o)(\bar{y}_2,\alpha_2) \Leftrightarrow
(\bar{y}_1,\alpha_1) = (\bar{y}_2,\alpha_2) \Leftrightarrow
\bar{y}_1 = \bar{y}_2 \mbox{ and } \alpha_1 = \alpha_2 \, .
$$
Now, using the previous expression for $(\mu_\W \circ j_o)(\bar{y},\alpha)$, we have
$$
(\mu_\W \circ j_o)(\bar{y}_1,\alpha_1) = (\mu_\W \circ j_o)(\bar{y}_2,\alpha_2) \Leftrightarrow
(\bar{y}_1,[\alpha_1]) = (\bar{y}_2,[\alpha_2]) \Leftrightarrow
\bar{y}_1 = \bar{y}_2 \mbox{ and } [\alpha_1] = [\alpha_2] \, .
$$
Hence, by definition of $\W_o$, we have $L(\pi^{2k-1}_{k}(\bar{y}_1)) = L(\pi^{2k-1}_{k}(\bar{y}_2))
= \hat{C}(\bar{y}_1,\alpha_1) = \hat{C}(\bar{y}_2,\alpha_2)$. Locally, from
the third equality we obtain
$$
p(\alpha_1) + p_A^i(\alpha_1)q_{i+1}^A(\bar{y}_1) = p(\alpha_2) + p_A^i(\alpha_2)q_{i+1}^A(\bar{y}_2) \, ,
$$
but $[\alpha_1] = [\alpha_2]\, \Longrightarrow\,p^i_A(\alpha_1) = p_A^i([\alpha_1]) = p_A^i([\alpha_2]) = p^i_A(\alpha_2)$.
Then $p(\alpha_1) = p(\alpha_2)$, and $\alpha_1 = \alpha_2$.
Furthermore, $\mu_\W \circ j_o$ is surjective. In fact,
given $(\bar{y},[\alpha]) \in \W_r$, we wish to find
$(\bar{y},\beta) \in j_o(\W_o)$ such that $[\beta] = [\alpha]$. It suffices to take
$[\beta]$ such that, in local coordinates of $\W$,
$$
p_A^i(\beta) = p_A^i([\beta]) \quad , \quad
p(\beta) = L(\pi^{2k-1}_{k}(\bar{y})) - p_A^i([\alpha])q_{i+1}^A(\bar{y}) \, .
$$
This $\beta$ exists as a consequence of the definition of $\W_o$.
Now, since $\mu_\W \circ j_o$ is a one-to-one submersion, then, by equality
on the dimensions of $\W_o$ and $\W_r$, it is a one-to-one local diffeomorphism,
and thus a global diffeomorphism.

Finally, in order to prove that $\W_o$ is $\mu_\W$-transversal, it is necessary to check if
$\Lie(Y)(\xi) \equiv Y(\xi) \neq 0$,
for every $Y \in \ker{\mu_\W}_*$ and every constraint function $\xi$ defining $\W_o$.
Since $\W_o$ is defined by the constraint function $\hat{C} - \hat{L} = 0$ and
$\ker{\mu_\W}_* = \{\partial/\partial p \}$, we have
$$
\derpar{}{p}(\hat{C} - \hat{L}) = \derpar{}{p}(p + p_A^iq_{i+1}^A - \hat{L}) = 1 \, ,
$$
then $\W_o$ is $\mu_\W$-transversal.
\end{proof}

As a consequence of this last result, in the following we consider the diagram:
$$
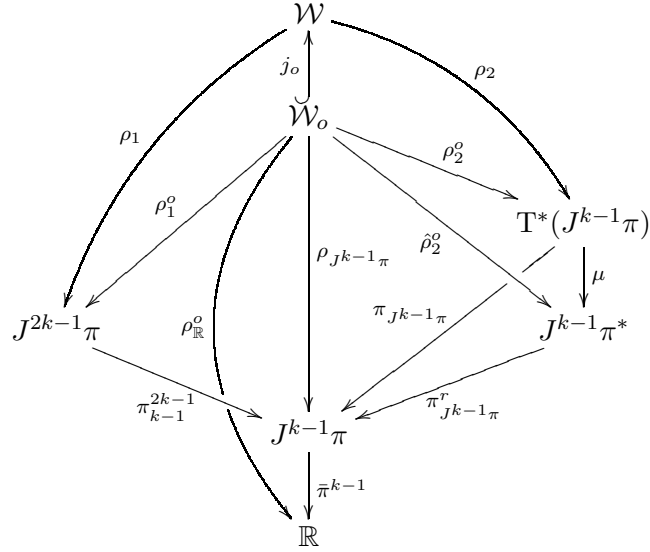
\xymatrix{
\ & \ & \W \ar@/_1.3pc/[llddd]_{\rho_1} \ar@/^1.3pc/[rrdd]^{\rho_2} & \ & \ \\
\ & \ & \W_o \ar@{^{(}->}[u]^{j_o} \ar[lldd]_{\rho_1^o} \ar[rrd]^{\rho_2^o} \ar[rrdd]_{\hat{\rho}_2^o} \ar[ddd]^<(0.4){\rho_{J^{k-1}\pi}} \ar@/_3pc/[dddd]_-{\rho_\R^o}|(.675){\hole} & \ & \ \\
\ & \ & \ & \ & \Tan^*(J^{k-1}\pi) \ar[d]^-{\mu} \ar[lldd]_{\pi_{J^{k-1}\pi}}|(.25){\hole} \\
J^{2k-1}\pi \ar[rrd]_{\pi^{2k-1}_{k-1}} & \ & \ & \ & J^{k-1}\pi^* \ar[dll]^{\pi_{J^{k-1}\pi}^r} \\
\ & \ & J^{k-1}\pi \ar[d]^{\bar{\pi}^{k-1}} & \ & \ \\
\ & \ & \R & \ & \
}
$$

As a consequence of Proposition \ref{prop:WoDiffWr}, the submanifold $\W_o$ induces a section
$\hat{h} \in \Gamma(\mu_\W)$, that is, a map $\hat{h}\colon\W_r \to \W$. This section is
specified by giving the local \textsl{Hamiltonian function}
\begin{equation}
\label{eqn:LocalHamiltonianFunctionUnified}
\hat{H} = -\hat{L} + p_A^iq_{i+1}^A \, ,
\end{equation}
that is, $\hat{h}(t,q_i^A,q_j^A,p_A^i) = (t,q_i^A,q_j^A,-\hat{H},p_A^i)$. The section $\hat{h}$ is
called a \textsl{Hamiltonian section} of $\mu_\W$, or a \textsl{Hamiltonian $\mu_\W$-section}.

Next, we can define the forms
\begin{equation*}
\label{eqn:DefFormsWo}
\Theta_o = j_o^*\Theta = (\rho_2^o)^*\Theta_{k-1} \in \df^{1}(\W_o) \quad ; \quad
\Omega_o = j_o^*\Omega = (\rho_2^o)^*\Omega_{k-1} \in \df^{2}(\W_o) \quad ,
\end{equation*}
with local expressions
\begin{equation}
\label{eqn:LocalCoordFormsWo}
\Theta_o = p_A^i\d q_i^A + (\hat{L} - p_A^iq_{i+1}^A)\d t \quad ; \quad
\Omega_o = \d q_i^A \wedge \d p_A^i + \d(p_A^iq_{i+1}^A - \hat{L}) \wedge \d t \quad ,
\end{equation}
and we have the presymplectic Hamiltonian systems $(\W_o,\Omega_o)$ and $(\W_r,\Omega_r)$, with
$\Omega_r = \hat{h}^*(\Omega)$.

Finally, it is necessary to introduce the following concepts:

\begin{definition}
A section $\psi_o \in \Gamma(\rho_\R^o)$ is {\rm holonomic of type $r$} in $\W_o$, $1 \leqslant r \leqslant 2k-1$,
if the section $\rho_1^o \circ \psi_o \in \Gamma(\bar{\pi}^{2k-1})$ is holonomic of type $r$ in $J^{2k-1}\pi$.
\end{definition}

\begin{definition}
A vector field $X_o\in\vf(\W_o)$ is said to be a {\rm semispray of type $r$} in $\W_o$ if
every integral section $\psi_o$ of $X_o$ is holonomic of type $r$ in $\W_o$.
\end{definition}

The local expression of a semispray of type $r$ in $\W_o$ is
$$
X_o = f\derpar{}{t} + \sum_{i=0}^{2k-1-r}q_{i+1}^A\derpar{}{q_i^A} + \sum_{i=2k-r}^{2k-1}X_i^A\derpar{}{q_i^A}
+\sum_{i=0}^{k-1}G^i_A\derpar{}{p^i_A} \ ,
$$
and, in particular, for a semispray of type $1$ in $\W_o$ we have
$$
X_o = f\derpar{}{t} + \sum_{i=0}^{2k-2}q_{i+1}^A\derpar{}{q_i^A} + X_{2k-1}^A\derpar{}{q_{2k-1}^A}
+\sum_{i=0}^{k-1}G^i_A\derpar{}{p^i_A} \ .
$$


\subsection{Dynamical equations}
\label{subsection:DynamicalEquationsUnified}

The dynamical equations for non-autonomous dynamical systems in general
can be geometrically written in several equivalent ways, using sections
(curves) which are the dynamical trajectories, or vector fields whose
integral curves are the dynamical trajectories. In this section
we explore both of these ways, and prove their equivalence.


\subsubsection{Dynamical equations for sections}
\label{subsubsection:DynEqSections}

The \textsl{Lagrangian-Hamiltonian problem for sections} associated with the system
$(\W_o,\Omega_o)$ consists in finding sections $\psi_o \in \Gamma(\rho_\R^o)$
(that is, curves $\psi_o \colon \R \to \W_o$) characterized by the condition
\begin{equation}
\label{eqn:DynEquationSections}
\psi_o^*\inn(Y)\Omega_o = 0, \quad \mbox{for every } Y \in \vf(\W_o) \, .
\end{equation}
In natural coordinates, let $Y \in \vf(\W_o)$ be a generic
vector field given by
\begin{equation}
\label{eqn:LocalCoordGenericVectorFieldWo}
Y = f \derpar{}{t} + \sum_{i=0}^{k-1} f_i^A\derpar{}{q_i^A} + \sum_{j=k}^{2k-1} F_j^A\derpar{}{q_j^A} + \sum_{i=0}^{k-1} G_A^i\derpar{}{p_A^i}\, ,
\end{equation}
bearing in mind the coordinate expression
\eqref{eqn:LocalCoordFormsWo} of $\Omega_o$,
the contraction $\inn(X)\Omega_o$ is
\begin{align*}
\inn(Y)\Omega_o
&= f\left( \derpar{\hat{L}}{q_r^A} \d q_r^A - q_{i+1}^A \d p_A^i - p_A^i \d q_{i+1}^A \right)
+ f_0^A\left( \d p_A^0 - \derpar{\hat{L}}{q_0^A} \d t \right)  \\
&\qquad + f_i^A\left( \d p_A^i - \derpar{\hat{L}}{q_i^A}\d t + p_A^{i-1} \d t \right)
+ F_k^A\left(p_A^{k-1} - \derpar{\hat{L}}{q_k^A}\right)\d t + G_A^i\left( q_{i+1}^A \d t - \d q_i^A \right) \, .
\end{align*}
Thus, taking the pull-back by the section
$\psi_o = (t,q_i^A(t),q_j^A(t),p_A^i(t))$,
we obtain
\begin{align*}
\psi_o^*\inn(Y)\Omega_o
&= f\left( \derpar{\hat{L}}{q_r^A} \dot{q}_r^A - q_{i+1}^A \dot{p}_A^i - p_A^i \dot{q}_{i+1}^A \right)\d t
+ f_0^A\left( \dot{p}_A^0 - \derpar{\hat{L}}{q_0^A} \right)\d t \\
&\qquad + f_i^A\left( \dot{p}_A^i - \derpar{\hat{L}}{q_i^A} + p_A^{i-1} \right)\d t
+ F_k^A\left(p_A^{k-1} - \derpar{\hat{L}}{q_k^A}\right)\d t + G_A^i\left( q_{i+1}^A - \dot{q}_i^A \right)\d t \, .
\end{align*}
Finally, requiring this last expression to vanish and bearing in
mind that the equation must hold for every vector field $Y \in \vf(\W_o)$
(that is, it must hold for every function $f,f_i^A,F_j^A,G_A^i \in \Cinfty(\W_o)$)
we obtain the following system of equations
\begin{align}
\label{eqn:LocalCoordRedundantEqSections}
&\derpar{\hat{L}}{q_r^A} \dot{q}_r^A - q_{i+1}^A \dot{p}_A^i - p_A^i \dot{q}_{i+1}^A = 0 \\
\label{eqn:LocalCoordMomentumDiffEq1}
&\dot{p}_A^0 = \derpar{\hat{L}}{q_0^A} \\
\label{eqn:LocalCoordMomentumDiffEq2}
&\dot{p}_A^i = \derpar{\hat{L}}{q_i^A} - p_A^{i-1}\\
\label{eqn:LocalCoordLastMomentumCoord}
&p_A^{k-1} = \derpar{\hat{L}}{q_k^A} \\
\label{eqn:LocalCoordHolonomyK}
&\dot{q}_i^A = q_{i+1}^A \ .
\end{align}
It is easy to check that equation \eqref{eqn:LocalCoordRedundantEqSections}
is redundant, since it is a consequence of the others. Equations
\eqref{eqn:LocalCoordMomentumDiffEq1}, \eqref{eqn:LocalCoordMomentumDiffEq2}
and \eqref{eqn:LocalCoordHolonomyK} are differential equations whose solutions
are the functions defining the section $\psi_o$.
In fact, equations \eqref{eqn:LocalCoordMomentumDiffEq1}, \eqref{eqn:LocalCoordMomentumDiffEq2}
give the higher-order Euler-Lagrange equations, as we see at the end of this Section
and in Section \ref{section:lagform}.
In addition, observe that equations
\eqref{eqn:LocalCoordLastMomentumCoord} do not involve any derivative
of $\psi_o$: they are pointwise algebraic conditions.
These equations arise from the
$\hat{\rho}_2^o$-vertical part of the vector fields $Y$.
Moreover, we have the following result:

\begin{lem}
If $Y \in \vf^{V(\hat{\rho}_2^o)}(\W_o)$, then $\inn(Y)\Omega_o$ is $\rho_\R^o$-semibasic.
\end{lem}
\begin{proof}
A direct calculation in coordinates leads to this result.
Bearing in mind that a local basis for the $\hat{\rho}_2^o$-vertical
vector fields is given by \eqref{eqn:LocalBasisKerOmega} and the local
expression \eqref{eqn:LocalCoordFormsWo} of $\Omega_o$, we have
\begin{equation*}
\inn\left( \derpar{}{q_j^A} \right) \Omega_o =
\begin{cases}
\left(p_A^{k-1} - \derpar{\hat{L}}{q_k^A}\right)\d t, & \mbox{for } j = k; \\
0 = 0\cdot\d t, & \mbox{for } j = k+1,\ldots,2k-1.
\end{cases}
\end{equation*}
Thus, in both cases we obtain a $\rho_\R^o$-semibasic form.
\end{proof}

As a consequence of this result, we can define the submanifold
\begin{equation*}
\label{eqn:DefSubmanifoldW1}
\W_c = \left\{ w \in \W_o \colon (\inn(Y)\Omega_o)(w)=0 \mbox{ for every } Y \in \vf^{V(\hat{\rho}_2^o)}(\W_o)\right\} \stackrel{j_1}{\hookrightarrow} \W_o \, ,
\end{equation*}
where every section $\psi_o$ solution of equation
\eqref{eqn:DynEquationSections} must take values.
It is called the \textsl{first constraint submanifold} of the
Hamiltonian presymplectic system $(\W_o,\Omega_o)$.

Locally, $\W_c$ is defined in $\W_o$ by the constraints $p_A^{k-1} - \partial \hat{L} / \partial q_k^A = 0$,
as we have seen in \eqref{eqn:LocalCoordLastMomentumCoord} and in the proof of the previous Lemma.
In combination with equations \eqref{eqn:LocalCoordMomentumDiffEq2}, we have:

\begin{prop}
\label{prop:W1GraphFLSections}
$\W_c$ contains a submanifold $\W_1 \hookrightarrow \W_c$ which can be identified
as the graph of a map $\Leg \colon J^{2k-1}\pi \to J^{k-1}\pi^*$ defined locally by
$$
\Leg^*t = t \quad , \quad
\Leg^*q_r^A = q_r^A \quad , \quad
\Leg^*p_A^{r-1} = \sum_{i=0}^{k-r}(-1)^i d_T^i\left( \derpar{\hat{L}}{q_{r+i}^A} \right) \, .
$$
\end{prop}
\begin{proof}
As $\W_c$ is defined locally by the constraints $p_A^{k-1} - \partial \hat{L} / \partial q_k^A = 0$,
it suffices to prove that these constraints give rise to the functions defining
the map given above, and thus to the submanifold $\W_1$. We do this in coordinates.

Taking into account that $d_T(p_A^i) = \dot{p}_A^i$ along sections, the constraint
function defining $\W_c$, in combination with equations \eqref{eqn:LocalCoordMomentumDiffEq2}
give rise to the following constraint functions
\begin{align*}
&p_A^{k-1} - \derpar{\hat{L}}{q_k^A} = 0 \\
&p^{k-2}_A - \left( \derpar{\hat{L}}{q_{k-1}^A} - d_T(p_A^{k-1}) \right) = p_A^{k-2} - \sum_{i=0}^{1}(-1)^i d_T^i\left(\derpar{\hat{L}}{q_{k-1+i}^A}\right) = 0 \\
&\qquad \qquad \qquad \qquad \vdots \\
&p^{1}_A - \left( \derpar{\hat{L}}{q_{2}^A} - d_T(p_A^{2}) \right) = p_A^1 - \sum_{i=0}^{k-2}(-1)^i d_T^i\left(\derpar{\hat{L}}{q_{2+i}^A}\right) = 0 \\
&p^{0}_A - \left( \derpar{\hat{L}}{q_{1}^A} - d_T(p_A^{1}) \right) = p_A^0 - \sum_{i=0}^{k-1}(-1)^i d_T^i\left(\derpar{\hat{L}}{q_{1+i}^A}\right) = 0 \ ,
\end{align*}
Therefore, these constraints define a submanifold $\W_1 \hookrightarrow \W_c$ and
we may consider that this $\W_1$ is the graph of a map $\Leg \colon J^{2k-1}\pi \to J^{k-1}\pi^*$ given by
$$
\Leg^*t = t \quad , \quad
\Leg^*q_r^A = q_r^A \quad , \quad
\Leg^*p_A^{r-1} = \sum_{i=0}^{k-r}(-1)^i d_T^i\left( \derpar{\hat{L}}{q_{r+i}^A} \right) \, .
$$
\end{proof}

Bearing in mind that the submanifold $\W_o \hookrightarrow \W$ is defined locally by
the constraint function $p + p_A^iq_{i+1}^A - \hat{L} = 0$, and that $\W_1$ is a
submanifold of $\W_c$, and thus a sumbanifold of $\W_o$,
from the above Proposition we can state the following result,
which is a straightforward consequence of the previous result:

\begin{corol}
\label{corol:W1GraphExtendedFLSections}
$\W_1$ is the graph of a map $\widetilde{\Leg} \colon J^{2k-1}\pi \to \Tan^*(J^{k-1}\pi)$ defined locally by
$$
\widetilde{\Leg}^*t = t \quad , \quad
\widetilde{\Leg}^*q_r^A = q_r^A \, ,
$$
$$
\widetilde{\Leg}^*p = \hat{L} - \sum_{r=1}^kq_r^A\sum_{i=0}^{k-r}(-1)^id_T^i\left( \derpar{\hat{L}}{q_{r+i}^A} \right) \quad , \quad
\widetilde{\Leg}^*p_A^{r-1} = \sum_{i=0}^{k-r}(-1)^i d_T^i\left( \derpar{\hat{L}}{q_{r+i}^A} \right) \, .
$$
\end{corol}

\textbf{Remark}: The submanifold $\W_1$ can be obtained from $\W_c$
using a constraint algorithm. Hence, $\W_1$ acts as the initial
phase space of the system.

The maps $\widetilde{\Leg}$ and $\Leg$ are called the
\textsl{extended Legendre-Ostrogradsky map} and the
\textsl{restricted Legendre-Ostrogradsky map} associated
to the Lagrangian density $\Lag$, respectively.
A justification of this terminology is given in Section
\ref{section:hamform}. Now we can give the following definition:

\begin{definition}
A Lagrangian density $\Lag \in \df^{1}(J^{k}\pi)$ is {\rm regular}
if the restricted Legendre-Ostrogradsky map $\Leg$ is a local diffeomorphism.
If the map $\Leg$ is a global diffeomorphism, then $\Lag$ is said to be
{\rm hyperregular}.
\end{definition}

Computing in natural coordinates the local expression of the
tangent map to $\Leg$, the regularity condition for $\Lag$
is equivalent to
$$
\det\left( \derpars{L}{q_k^B}{q_k^A} \right)(\bar{y}) \neq 0, \quad \mbox{for every } \bar{y} \in J^{k}\pi \, .
$$

Equivalently, if we denote $\hat{p}^{r-1}_A = \Leg^*p_A^{r-1}$, then
the Lagrangian density $\Lag$ is regular if, and only if, the set
$(t,q_i^A,\hat{p}_A^i)$, $0 \leqslant i \leqslant k-1$,
is a set of local coordinates in $J^{2k-1}\pi$. The local functions
$\hat{p}_A^i$ are called the \textsl{Jacobi-Ostrogradsky momentum coordinates},
and they satisfy that
\begin{equation}
\label{eqn:RelationsMomenta}
\hat{p}_A^{r-1} = \derpar{L}{q_r^A} - d_T(\hat{p}_A^r) \quad ,
\end{equation}
which are exactly the relations given by \eqref{eqn:LocalCoordMomentumDiffEq2},
taking into account that $d_T = \d/\d t$ along sections.

Notice that
equations \eqref{eqn:LocalCoordMomentumDiffEq1}, \eqref{eqn:LocalCoordMomentumDiffEq2},
and \eqref{eqn:LocalCoordHolonomyK} do not allow us to determinate the functions
$q_j^A$, $k \leqslant j \leqslant 2k-1$, of the section $\psi_o$.
Thus, in the general case, we need an additional
condition when stating the problem, which is the holonomy condition for the section $\psi_o$.
Therefore, the Lagrangian-Hamiltonian problem must be reformulated as follows:

\noindent\textsl{The Lagrangian-Hamiltonian problem
consists in finding holonomic sections $\psi_o \in \Gamma(\rho_\R^o)$ characterized by the
equation \eqref{eqn:DynEquationSections}}.

\textbf{Remarks}:
\bit
\item
In fact, the functions $q_j^A$, $k \leqslant j \leqslant 2k-1$,
are determined by the equations
\eqref{eqn:LocalCoordMomentumDiffEq1} and \eqref{eqn:LocalCoordMomentumDiffEq2},
bearing in mind that the section $\psi_o$ must lie in the submanifold
$\W_1 = {\rm graph}(\widetilde{\Leg})$.
It is easy to see that, by replacing the local expression of the
extended Legendre-Ostrogradsky map in the equations
\eqref{eqn:LocalCoordMomentumDiffEq1} and \eqref{eqn:LocalCoordMomentumDiffEq2},
these equations
lead to the Euler-Lagrange equations and to the remaining $(k-1)n$ equations
that give the full holonomy condition:
\begin{align*}
&\left(\dot{q}_j^B - q_{j+1}^B \right) \restric{\derpars{\hat{L}}{q_k^B}{q_k^A}}{\psi_o} - \sum_{i=k}^{j-1}\left(\dot{q}_i^B-q_{i+1}^B\right)(\cdots\cdots) = 0 \qquad (k \leqslant j \leqslant 2k-2) \\
&\restric{\derpar{\hat{L}}{q_0^A}}{\psi_o} - \restric{\frac{\d}{\d t}\derpar{\hat{L}}{q_1^A}}{\psi_o}
+ \restric{\frac{\d^2}{\d t^2}\derpar{\hat{L}}{q_2^A}}{\psi_o} + \ldots +
(-1)^k \restric{\frac{\d^k}{\d t^k}\derpar{\hat{L}}{q_k^A}}{\psi_o} = 0 \, ,
\end{align*}
where the terms in brackets $(\cdots)$ contain terms involving
partial derivatives of the Lagrangian function and iterated total
time derivatives, and the first sum (for $j=k$) is empty.
However, observe that these equations may or may not be compatible,
and a sufficient condition to ensure compatibility is the regularity of the Lagrangian
density. Thus, for singular Lagrangian densities, the holonomy condition for
the section $\psi_o$ is required.

\item
The requirement of the section $\psi_o$ to be holonomic is a relevant difference
from the first-order case, where the holonomy condition
is deduced straightforwardly from the dynamical equations when
written in local coordinates.
Nevertheless, in the higher-order case, the equations allow us to recover only the
holonomy of type $k$, as seen in \eqref{eqn:LocalCoordHolonomyK}, and the
highest-order holonomy condition can only be recovered from the equations
if the Lagrangian density is regular. Hence, this condition is required ``ad hoc''.
\item
The regularity of the Lagrangian density has no relevant role
at first sight. However, as we have seen in the first remark,
equations \eqref{eqn:LocalCoordMomentumDiffEq1} and \eqref{eqn:LocalCoordMomentumDiffEq2}
give the higher-order Euler-Lagrange equations, which have a unique solution
if the Lagrangian density is regular. For singular Lagrangians, these equations
may give rise to new constraints, and a constraint algorithm should be used 
for finding a submanifold where the equations can be solved.
\eit


\subsubsection{Dynamical equations for vector fields}
\label{subsubsection:DynEqVectFields}

The \textsl{Lagrangian-Hamiltonian problem for vector fields} associated with
the system $(\W_o,\Omega_o)$ consists in finding vector fields $X_o \in \vf(\W_o)$
such that
\begin{equation}
\label{eqn:DynEquationVectorFields}
\inn(X_o)\Omega_o = 0 \quad ; \quad \inn(X_o)(\rho_\R^o)^*\eta = 1 \, .
\end{equation}

According to \cite{art:Chinea_deLeon_Marrero94} and \cite{art:deLeon_Marin_Marrero_Munoz_Roman02},
we have:

\begin{prop}
\label{prop:ExistSolDynEq}
A solution $X_o \in \vf(\W_o)$ to equation \eqref{eqn:DynEquationVectorFields} exists only
on the points of the submanifold $\mathcal{S}_c$ defined by
\begin{equation}
\label{eqn:DefW1}
\mathcal{S}_c = \left\{ w \in \W_o \colon (\inn(Z)\d\hat{H})(w) = 0 \mbox{ for every } Z \in \ker\Omega \right\}
\stackrel{j_1}{\hookrightarrow} \W_o \, .
\end{equation}
\end{prop}

We have the following result:

\begin{prop}
\label{prop:W1GraphFLVectorFields}
The submanifold $\mathcal{S}_c \hookrightarrow \W_o$ contains a submanifold
$\mathcal{S}_1 \hookrightarrow \mathcal{S}_c$ which is the graph of the extended
Legendre-Ostrogradsky map; that is, $\mathcal{S}_1 = {\rm graph}\,\widetilde{\Leg}$;
and hence $\mathcal{S}_1 = \W_1$.
\end{prop}
\begin{proof}
As $\mathcal{S}_c$ is defined by (\ref{eqn:DefW1}), it suffices to prove that
the constraints defining $\mathcal{S}_c$ give rise to the constraint functions
defining the graph of the extended Legendre-Ostrogradsky map associated to $\Lag$.
We do this calculation in coordinates.
Taking the local expression \eqref{eqn:LocalHamiltonianFunctionUnified}
of the local Hamiltonian function $\hat{H} \in \Cinfty(\W_o)$, we have 
$$
\d\hat{H} = \sum_{i=0}^{k-1}(q_{i+1}^A\d p^i_A + p^i_A\d q_{i+1}^A) - \sum_{i=0}^{k} \derpar{\hat{L}}{q_i^A}\d q_i^A \ ,
$$
and using the local basis of $\ker\,\Omega$ given in \eqref{eqn:LocalBasisKerOmega},
we obtain that the equations defining the submanifold $\mathcal{S}_c$ are
$$
\inn(Z)\d\hat{H} = 0 \Longleftrightarrow p^{k-1}_A - \derpar{\hat{L}}{q_k^A} = 0 \, , \ \mbox{for every } 1 \leqslant A \leqslant n \ .
$$
Note that these expressions relate the momentum coordinates
$p^{k-1}_A$ with the Jacobi-Ostrogradsky functions
$\ds \hat p^{k-1}_A=\partial\hat{L}/\partial q_k^A$, and so we obtain
the last group of equations of the restricted Legendre-Ostrogradsky map.
Now, using the same argument as in the proof of Proposition \ref{prop:W1GraphFLSections}
and the relations \eqref{eqn:RelationsMomenta} for the momenta,
we can consider that $\mathcal{S}_c$ contains a submanifold $\mathcal{S}_1$
which is the graph of a map
$$
\begin{array}{rcl} F \colon J^{2k-1}\pi & \longrightarrow & J^{k-1}\pi^* \\ 
(t,q_i^A,q_j^A) & \longmapsto & (t,q_i^A,p_A^i) 
\end{array}
$$
which we identify with the
restricted Legendre-Ostrogradsky map
by making the identification $p^{r-1}_A=\hat p^{r-1}_A$.

Finally, taking into account that $\mathcal{S}_1$ is also a submanifold of $\W_o$,
which is defined by the constraint $p + p_A^iq_{i+1}^A - \hat{L} = 0$,
we have as a direct consequence that $\mathcal{S}_1$ is the graph of the extended
Legendre-Ostrogradsky map $\widetilde{\Leg}$, and hence $\mathcal{S}_1 = \W_1$.
\end{proof}

We denote by $\vf_{\W_1}(\W)$ the set of vector fields in $\W_o$ at support on $\W_1$.
Hence, we look for vector fields $X_o \in \vf_{\W_1}(\W_o)$ which are
solutions to equations \eqref{eqn:DynEquationVectorFields} at support
on $\W_1$; that is
\begin{equation}\label{eqn:DynEqVectorFieldsSuppW1}
\restric{\inn(X_o)\Omega_o}{\W_1} = 0 \quad ; \quad
\restric{\inn(X_o)(\rho_\R^o)^*\eta}{\W_1} = 1 \ .
\end{equation}

In natural coordinates, let $X_o \in \vf(\W_o)$ be a generic vector field
given locally by \eqref{eqn:LocalCoordGenericVectorFieldWo}. Thus,
from \eqref{eqn:DynEquationVectorFields} we obtain the following system
of $(2k+1)n+2$ equations
\begin{align}
\label{eqn:LocalCoordRedundantEqVectorFields}
&-f_0^A\derpar{\hat{L}}{q_0^A} + f_i^A\left( p_A^{i-1} - \derpar{\hat{L}}{q_i^A} \right)
+ F_k^A\left(p_A^{k-1} - \derpar{\hat{L}}{q_k^A}\right) +G_A^iq_{i+1}^A = 0 \ , \\
\label{eqn:LocalCoordSemisprayK}
&f_i^A = fq_{i+1}^A \ , \\
\label{eqn:LocalCoordDynEquationVectorFields}
&G_A^0 = f\derpar{\hat{L}}{q_0^A} \quad , \quad G_A^i = f\left( \derpar{\hat{L}}{q_i^A} - p_A^{i-1} \right) = fd_T(p_A^i) \ , \\
\label{eqn:LocalCoordFixedGauge}
&f = 1 \ , \\
\label{eqn:LocalCoordLastMomentumCoordVectorField}
&f\left(p_A^{k-1} - \derpar{\hat{L}}{q_k^A}\right) = 0 \ ,
\end{align}
where $0 \leqslant i \leqslant k-1$ in \eqref{eqn:LocalCoordSemisprayK},
and $1 \leqslant i \leqslant k-1$ in \eqref{eqn:LocalCoordDynEquationVectorFields}.
By a simple calculation one can see that equation \eqref{eqn:LocalCoordRedundantEqVectorFields}
is redundant, since it is a combination of the others. Therefore
\begin{equation}
\label{eqn:LocalCoordSolutionVectorField}
X_o = \derpar{}{t} + q_{i+1}^A\derpar{}{q_i^A} + F_j^A\derpar{}{q_j^A} + \derpar{\hat{L}}{q_0^A}\derpar{}{p_A^0} + d_T(p_A^i) \derpar{}{p_A^i}\, .
\end{equation}

\textbf{Remark}: In a more general situation,
the second equation in \eqref{eqn:DynEquationVectorFields} is written
$\inn(X_o)\rho_\R^*\eta \neq 0$,
that is, a $\rho_\R^o$-transversal condition for the
vector field $X_o$.
In local coordinates, this replaces equation
\eqref{eqn:LocalCoordFixedGauge} by $f \neq 0$, thus giving the vector field
\begin{equation*}
X_o = f \left( \derpar{}{t} + q_{i+1}^A\derpar{}{q_i^A} + F_j^A\derpar{}{q_j^A} + \derpar{\hat{L}}{q_0^A}\derpar{}{p_A^0} + d_T(p_A^i) \derpar{}{p_A^i} \right) \, ,
\end{equation*}
where $f \in \Cinfty(\W_o)$ is any non-vanishing function. This gives a whole family of
vector field solutions to the dynamical equations, and taking a particular constant
value for $f$ just fixes a specific vector field in this family. From a physical
viewpoint, taking a particular value for $f$ is just fixing the gauge.

Observe that equations \eqref{eqn:LocalCoordLastMomentumCoordVectorField} are just
a compatibility condition for the vector field $X_o$, which, together with the relations
\eqref{eqn:RelationsMomenta} for the momenta, state that vector field $X_o$ solutions
to equations \eqref{eqn:DynEquationVectorFields} exist only at support on the submanifold
defined by the graph of the extended Legendre-Ostrogradsky map. Thus, we recover, in
coordinates, the result stated in Propositions \ref{prop:ExistSolDynEq} and \ref{prop:W1GraphFLVectorFields}.
Furthermore, equations \eqref{eqn:LocalCoordSemisprayK} show that $X_o$ is a semispray of type $k$ in $\W_o$.

The component functions $F_j^A$, $k \leqslant j \leqslant 2k-1$, are undetermined.
Nevertheless, recall that $X_o$ is a vector field that must be tangent to the submanifold $\W_1$.
Thus, it is necessary to impose that $\restric{\Lie(X_o)\xi}{\W_1} = 0$ for every constraint function
$\xi$ defining $\W_1$. Locally, this is equivalent to imposing $\restric{X_o(\xi)}{\W_1} = 0$.
Hence, taking into account Prop. \ref{prop:W1GraphFLVectorFields},
these conditions lead to
\begin{align*}
&\left(\derpar{}{t} + q_{i+1}^A\derpar{}{q_i^A} + F_j^A\derpar{}{q_j^A} + \derpar{\hat{L}}{q_0^A}\derpar{}{p_A^0}
+ d_T(p_A^i) \derpar{}{p_A^i}\right)
\left( p^{k-1}_A - \derpar{\hat{L}}{q_k^A} \right) = 0 \\
&\left(\derpar{}{t} + q_{i+1}^A\derpar{}{q_i^A} + F_j^A\derpar{}{q_j^A} + \derpar{\hat{L}}{q_0^A}\derpar{}{p_A^0}
+ d_T(p_A^i) \derpar{}{p_A^i}\right)
\left( p^{k-2}_A - \sum_{i=0}^{1}(-1)^i d_T^i\left(\derpar{\hat{L}}{q_{k-1+i}^A}\right)  \right) = 0 \\
&\qquad \qquad \qquad \qquad \vdots \\
&\left(\derpar{}{t} + q_{i+1}^A\derpar{}{q_i^A} + F_j^A\derpar{}{q_j^A} + \derpar{\hat{L}}{q_0^A}\derpar{}{p_A^0}
+ d_T(p_A^i) \derpar{}{p_A^i}\right)
\left( p^{1}_A - \sum_{i=0}^{k-2}(-1)^i d_T^i\left(\derpar{\hat{L}}{q_{2+i}^A}\right)  \right) = 0 \\
&\left(\derpar{}{t} + q_{i+1}^A\derpar{}{q_i^A} + F_j^A\derpar{}{q_j^A} + \derpar{\hat{L}}{q_0^A}\derpar{}{p_A^0}
+ d_T(p_A^i) \derpar{}{p_A^i}\right)
\left( p^{0}_A - \sum_{i=0}^{k-1}(-1)^i d_T^i\left(\derpar{\hat{L}}{q_{1+i}^A}\right)  \right) = 0 \ ,
\end{align*}
(observe that we do not need to check $\Lie(X_o)(p - \widetilde{\Leg}^*p) = 0$, since
this is the constraint defining the submanifold $\W_o \hookrightarrow \W$, and
$X_o$ is a vector field already defined in $\W_o$)
and, from here, we obtain the following $kn$ equations
\begin{equation}
\label{eqn:TangencyVectorFieldXo}
\begin{array}{l}
\displaystyle \left(F_k^B-q_{k+1}^B\right)\derpars{\hat{L}}{q_k^B}{q_k^A} = 0 \\[10pt]
\displaystyle \left(F_{k+1}^B - q_{k+2}^B\right)\derpars{\hat{L}}{q_k^B}{q_k^A} - 
\left(F_k^B-q_{k+1}^B \right) d_T\left(\derpars{\hat{L}}{q_k^B}{q_k^A}\right) = 0 \\
\qquad \qquad \qquad \qquad \vdots \\
\displaystyle \left(F_{2k-2}^B - q_{2k-1}^B\right)\derpars{\hat{L}}{q_k^B}{q_k^A} -
 \sum_{i=0}^{k-3} \left(F_{k+i}^B-q_{k+i+1}^B \right) (\cdots\cdots) = 0 \\
\displaystyle (-1)^k\left(F_{2k-1}^B - d_T\left(q_{2k-1}^B\right)\right) \derpars{\hat{L}}{q_k^B}{q_k^A} + 
\sum_{i=0}^{k} (-1)^id_T^i\left( \derpar{\hat{L}}{q_i^A} \right) - \sum_{i=0}^{k-2} \left(F_{k+i}^B-q_{k+i+1}^B \right)
 (\cdots\cdots) = 0 \ ,
\end{array}
\end{equation}
where the terms in brackets $(\cdots\cdots)$ contain relations involving partial
derivatives of the Lagrangian function $\hat{L}$ and applications of the total
derivative $d_T$, which are not written for simplicity.
These equations may or may not be compatible, and a sufficient condition for
compatibility is the regularity of the Lagrangian density $\Lag$.
In particular, we have:

\begin{prop}
\label{prop:RegularLagrangianTangencyCondition}
If $\Lag \in \df^1(J^k\pi)$ is a regular Lagrangian density, then
there exists a unique vector field $X_o \in \vf_{\W_1}(\W_o)$ which is a solution
to equation \eqref{eqn:DynEqVectorFieldsSuppW1}; it  is  tangent to $\W_1$,
and is a semispray of type $1$ in $\W_o$.
\end{prop}
\begin{proof}
As the Lagrangian density $\Lag$ is regular, the Hessian matrix
$\ds \left(\derpars{\hat{L}}{q_k^B}{q_k^A}\right)$ is regular at every point,
and this enables us to solve the above $k$ systems of $n$ equations (\ref{eqn:TangencyVectorFieldXo})
determining all the functions $F_i^A$ uniquely, as follows
\bea
\label{eqn:TangencyConditionRegularLagrangian}
 F_i^A = q_{i+1}^A \quad , \quad  (k \leqslant i \leqslant 2k-2) \\
 (-1)^k\left(F_{2k-1}^B - d_T\left(q_{2k-1}^B\right)\right) \derpars{\hat{L}}{q_k^B}{q_k^A} +
 \sum_{i=0}^{k} (-1)^id_T^i\left( \derpar{\hat{L}}{q_i^A} \right) = 0 \ .
 \nonumber
\eea
In this way, the tangency condition holds for $X_o$ at every point on $\W_1$.
Furthermore, the equalities \eqref{eqn:TangencyConditionRegularLagrangian}
show that $X_o$ is a semispray of type $1$ in $\W_o$ with local expression
\begin{equation}
\label{eqn:LocalCoordSolutionVectorFieldRegularLagrangian}
X_o = \derpar{}{t} + q_{i+1}^A\derpar{}{q_i^A} + F_{2k-1}^A \derpar{}{q_{2k-1}^A} + \derpar{\hat{L}}{q_0^A}\derpar{}{p_A^0} + d_T(p_A^i) \derpar{}{p_A^i}\, .
\end{equation}
\end{proof}

However, if $\Lag$ is  not regular, the equations \eqref{eqn:TangencyVectorFieldXo}
may or may not be compatible, and the compatibility condition may give rise to new constraints.
In the most favourable cases,
there is a submanifold $\W_f \hookrightarrow \W_1$ (it could be $\W_f = \W_1$)
such that there exist vector fields $X_o\in\vf_{\W_1}(\W_o)$, tangent to $\W_f$,
which are solutions to the equations
\begin{equation}
\label{eqn:DynEqVectorFieldsSuppWf}
\restric{\inn(X_o)\Omega_o}{\W_f} = 0 \quad , \quad
\restric{\inn(X_o)(\rho_\R^o)^*\eta}{\W_f} = 1 \ .
\end{equation}

Finally, the relation among the results obtained in the
two last sections is as follows:

\begin{teor}
\label{thm:EquivalenceTheoremUnified}
The following assertions on a holonomic section $\psi_o \in \Gamma(\rho_\R^o)$ are equivalent:
\begin{enumerate}
\item $\psi_o$ is a solution to equation \eqref{eqn:DynEquationSections}, that is,
$$\psi_o^*\inn(Y)\Omega_o = 0, \quad \mbox{for every }Y \in \vf(\W_o) \, .$$
\item If $\psi_o$ is given locally by $\psi_o(t) = (t,q_i^A(t),q_j^A(t),p_A^i(t))$,
$0 \leqslant i \leqslant k-1$, $k \leqslant j \leqslant 2k-1$, then the components
of $\psi_o$ satisfy equations \eqref{eqn:LocalCoordMomentumDiffEq1} and
\eqref{eqn:LocalCoordMomentumDiffEq2}, that is, the following system of $kn$ differential equations
\begin{equation}
\label{eqn:LocalCoordProofEquivalenceTheoremDiffEq}
\dot{p}_A^0 = \derpar{\hat{L}}{q_0^A} \quad ; \quad \dot{p}_A^i = \derpar{\hat{L}}{q_i^A} - p_A^{i-1} \, .
\end{equation}
\item $\psi_o$ is a solution to the equation
\begin{equation}
\label{eqn:DynEquationIntegralCurve}
\inn(\psi^\prime_o)(\Omega_o \circ \psi_o) = 0 \, ,
\end{equation}
where $\psi^\prime_o \colon \R \to \Tan\W_o$
is the canonical lifting of $\psi_o$ to the tangent bundle.
\item $\psi_o$ is an integral curve of a vector field
contained in a class of $\rho_\R^o$-transverse semisprays of type $1$,
$\left\{ X_o \right\} \subset \vf(\W_o)$,
satisfying the first equation in \eqref{eqn:DynEquationVectorFields}, that is,
\begin{equation*}
\inn(X_o)\Omega_o = 0 \, .
\end{equation*}
\end{enumerate}
\end{teor}
\begin{proof}

($1 \Leftrightarrow2$) \quad
As we have seen in Section \ref{subsubsection:DynEqSections},
equation \eqref{eqn:DynEquationSections} gives, in natural coordinates,
the equations \eqref{eqn:LocalCoordRedundantEqSections},
\eqref{eqn:LocalCoordMomentumDiffEq1}, \eqref{eqn:LocalCoordMomentumDiffEq2},
\eqref{eqn:LocalCoordLastMomentumCoord} and \eqref{eqn:LocalCoordHolonomyK}.
As stated there, equation \eqref{eqn:LocalCoordRedundantEqSections}
is redundant, since it is a combination of the others, and from equations
\eqref{eqn:LocalCoordLastMomentumCoord} we deduce that
the section $\psi_o \in \Gamma(\rho_\R^o)$ lies in the submanifold $\W_1$.
Hence, equation \eqref{eqn:DynEquationSections} is locally equivalent
to equations \eqref{eqn:LocalCoordMomentumDiffEq1}, \eqref{eqn:LocalCoordMomentumDiffEq2},
and \eqref{eqn:LocalCoordHolonomyK}.
However, as we assume
that $\psi_o$ is holonomic, equations
\eqref{eqn:LocalCoordHolonomyK} hold identically, and thus
equation \eqref{eqn:DynEquationSections}
is locally equivalent to equations \eqref{eqn:LocalCoordMomentumDiffEq1}
and \eqref{eqn:LocalCoordMomentumDiffEq2}, that is, to equations
\eqref{eqn:LocalCoordProofEquivalenceTheoremDiffEq}.

($2\Leftrightarrow3$) \quad
If $\psi_o(t) = (t,q_i^A(t),q_j^A(t),p_A^i(t))$
is the local expression of $\psi_o$ in natural coordinates,
then $\psi_o^\prime(t) = (1,\dot{q}_i^A(t),\dot{q}_j^A(t),\dot{p}_A^i(t))$,
and the inner product $\inn(\psi_o^\prime)(\Omega_o \circ \psi_o)$ gives, in coordinates,
\begin{align*}
\inn(\psi_o^\prime)(\Omega_o \circ \psi)_o &=
\left( p_A^i\dot{q}_{i+1}^A - \dot{q}_r^A\derpar{\hat{L}}{q_r^A} + \dot{p}_A^iq_{i+1}^A \right)\d t
+ \left( \derpar{\hat{L}}{q_0^A} - \dot{p}_A^0 \right) \d q_0^A \\
&\qquad + \left( \derpar{\hat{L}}{q_i^A} - p_A^{i-1} - \dot{p}_A^i \right)\d q_i^A
+ \left( p_A^{k-1} - \derpar{\hat{L}}{q_k^A} \right)\d q_k^A
+ \left( \dot{q}_i^A - q_{i+1}^A \right)\d p_A^i \, .
\end{align*}
Now, requiring this last expression to vanish, we obtain
the system of $(2k+1)n+1$ equations
\begin{align*}
p_A^i\dot{q}_{i+1}^A - \dot{q}_r^A\derpar{\hat{L}}{q_r^A} + \dot{p}_A^iq_{i+1}^A = 0 \quad &; \quad
\dot{p}_A^0 = \derpar{\hat{L}}{q_0^A} \quad ; \quad
\dot{p}_A^i = \derpar{\hat{L}}{q_i^A} - p_A^{i-1} \\
p_A^{k-1} = \derpar{\hat{L}}{q_k^A} &\quad ; \quad
\dot{q}_i^A = q_{i+1}^A
\end{align*}

Observe that this system of equations is the same
given by \eqref{eqn:LocalCoordRedundantEqSections},
\eqref{eqn:LocalCoordMomentumDiffEq1}, \eqref{eqn:LocalCoordMomentumDiffEq2},
\eqref{eqn:LocalCoordLastMomentumCoord} and \eqref{eqn:LocalCoordHolonomyK}.
The same remarks given in the proof of ($1\Leftrightarrow2$) apply in this case.
In particular,
the fifth group of $kn$ equations $\dot{q}_i^A = q_{i+1}^A$ is identically
satisfied by the section $\psi_o$, since we assume it to be holonomic.
Thus, bearing in mind the above item, we have proved that equation \eqref{eqn:DynEquationIntegralCurve} is locally
equivalent to the $kn$ differential equations \eqref{eqn:LocalCoordProofEquivalenceTheoremDiffEq}.

($2\Leftrightarrow4$) \quad
As we have seen in this Section,
if a generic vector field $X_o \in \vf(\W_o)$ is given locally
by \eqref{eqn:LocalCoordGenericVectorFieldWo}, then
the first equation
in \eqref{eqn:DynEquationVectorFields} is locally equivalent
to equations \eqref{eqn:LocalCoordRedundantEqVectorFields},
\eqref{eqn:LocalCoordSemisprayK},
\eqref{eqn:LocalCoordDynEquationVectorFields}
and \eqref{eqn:LocalCoordLastMomentumCoordVectorField}.
As already stated, equation \eqref{eqn:LocalCoordRedundantEqVectorFields}
is redundant, since it is a combination of the others; and the $n$ equations
\eqref{eqn:LocalCoordLastMomentumCoordVectorField}
state, in coordinates, the result given in Proposition \ref{prop:ExistSolDynEq}.
In addition, since the vector fields $X_o$ in the class are semisprays of
type $1$, the $kn$ equations \eqref{eqn:LocalCoordSemisprayK}
are identically satisfied. Thus, the first equation in \eqref{eqn:DynEquationVectorFields}
is locally equivalent to the $kn$ equations \eqref{eqn:LocalCoordDynEquationVectorFields}.
Finally, the $\rho_\R^o$-transverse condition for the class $\left\{X_o\right\}$
is locally equivalent to $f \neq 0$.

Now, let $\sigma \in \Gamma(\rho_\R^o)$ be an integral
curve of $X_o$, that is, $\sigma^\prime = X_o \circ \sigma$.
If $\sigma$ is
given locally by $\sigma(t) = (t,q_i^A(t),q_j^A(t),p_A^i(t))$,
then $\sigma^\prime(t) = (1,\dot{q}_i^A(t),\dot{q}_j^A(t),\dot{p}_A^i(t))$,
and, taking $f = 1$ as a representative of the class $\left\{X_o\right\}$,
the condition of $\sigma$ to be an integral curve is locally equivalent to the equations
\begin{equation*}
\dot{q}_i^A = f_i^A \circ \sigma \quad ; \quad
\dot{q}_j^A = F_j^A \circ \sigma \quad ; \quad
\dot{p}_A^i = G_A^i \circ \sigma \, .
\end{equation*}
Replacing these equations in \eqref{eqn:LocalCoordSemisprayK} and
\eqref{eqn:LocalCoordDynEquationVectorFields},
we obtain the following $2kn$ differential equations
$$
\dot{q}_i^A = q_{i+1}^A \quad ; \quad
\dot{p}_A^0 = \derpar{\hat{L}}{q_0^A} \quad ; \quad
\dot{p}_A^i = \derpar{\hat{L}}{q_i^A} - p_A^{i-1} \, .
$$
Observe that, as every vector field in the class is a semispray of type $1$,
the first $kn$ equations are identically satisfied. Thus,
the condition of $\sigma$ to be an integral curve of a $\rho_\R^o$-transverse
semispray of type $1$, $X_o \in \vf(\W_o)$, satisfying the
first equation in \eqref{eqn:DynEquationVectorFields} is
locally equivalent to equations \eqref{eqn:LocalCoordProofEquivalenceTheoremDiffEq}.
\end{proof}


\section{Lagrangian formalism}
\label{section:lagform}


\subsection{General setting}

Now we recover the Lagrangian dynamics from the unified formalism.
We do not distinguish between the regular and singular cases,
since the results remain the same in either case, but a few comments
on the singular case will be given.
First, we have:

\begin{prop}
\label{prop:W1DiffeoJ2k-1Pi}
The map $\rho_1^1 = \rho_1^o \circ j_1 \colon \W_1 \to J^{2k-1}\pi$ is a diffeomorphism.
\end{prop}
\begin{proof}
As $\W_1 = {\rm graph}\,\Leg$, we have that $J^{2k-1}\pi \simeq \W_1$.
Furthermore, $\rho_1^1$ is a surjective submersion and, by the equality between dimensions,
it is also an injective immersion and hence it is a diffeomorphism.
\end{proof}

Now, we must define the Poincar\'{e}-Cartan forms
in order to establish the dynamical equations for the
Lagrangian formalism. First, we have the following result:

\begin{lem}
\label{lemma:RelationFormsLagrangianUnified}
Let $\Theta_{k-1} \in \df^{1}(\Tan^*(J^{k-1}\pi))$,
$\Omega_{k-1} = -\d\Theta_{k-1} \in \df^{2}(\Tan^*(J^{k-1}\pi))$
be the canonical forms in $\Tan^*(J^{k-1}\pi)$. We define the
{\rm Poincar\'{e}-Cartan forms} as
$\Theta_{\Lag} = \widetilde{\Leg}^*\Theta_{k-1} \in \df^{1}(J^{2k-1}\pi)$,
$\Omega_{\Lag} = - \d\Theta_{\Lag} = \widetilde{\Leg}^*\Omega_{k-1} \in \df^{2}(J^{2k-1}\pi)$.
Then
$\Theta_o = (\rho_1^o)^*\Theta_{\Lag}$ and
$\Omega_o = (\rho_1^o)^*\Omega_{\Lag}$.
\end{lem}
\begin{proof}
We have for $\Theta_{\Lag}$:
$$
(\rho_1^o)^*\Theta_\Lag = (\rho_1^o)^*(\widetilde{\Leg}^*\Theta_{k-1})
= (\widetilde{\Leg} \circ \rho_1^o)^*\Theta_{k-1} = (\rho_2^o)^*\Theta_{k-1}
= \Theta_o \, ,
$$
and for $\Omega_{\Lag}$:
$$
(\rho_1^o)^*\Omega_\Lag = (\rho_1^o)^*(-\d\Theta_\Lag) = -\d(\rho_1^o)^*\Theta_\Lag
= -\d\Theta_o = \Omega_o \, .
$$
\end{proof}

Alternatively, according to \cite{art:Saunders87} and \cite{book:Saunders89}
(see also \cite{art:Aldaya_Azcarraga78_2}, \cite{proc:Garcia_Munoz83}),
we can define the
\textsl{Poincar\'{e}-Cartan $1$-form} using the canonical
structures of the higher-order jet bundles; in particular,
$$
\Theta_\Lag = S_\eta^{(k)}(\d L) + \Lag \in \df^{1}(J^{2k-1}\pi) \, ,
$$
where $S_\eta^{(k)}$ is the generalization to
higher-order jet bundles of the operator used in the classical
Hamilton-Cartan formalism for problems in the calculus of variations
which involve time explicitly
(see \cite{art:Saunders87} and \cite{book:Saunders89} for details).

Using natural coordinates, the local expression of the Poincar\'e-Cartan
$1$-form is
\begin{equation}
\label{eqn:LocalCoordPoincareCartan1Form}
\Theta_\Lag = \sum_{r=1}^{k} \sum_{i=0}^{k-r}(-1)^i d_T^i \left( \derpar{L}{q^A_{r+i}} \right)(\d q^A_{r-1} - q_r^A\d t) + L\d t \, .
\end{equation}

\textbf{Remark}: $\Theta_\Lag$ is a $\pi^{2k-1}_{k-1}$-semibasic $1$-form.

From the Poincar\'{e}-Cartan $1$-form, the concept of regularity for a higher-order Lagrangian density
is a straightforward generalization of the well-known definition for first-order non-autonomous
dynamical systems. In fact, first we define the \textsl{Poincar\'e-Cartan $2$-form} as
$\Omega_\Lag = -\d\Theta_\Lag \in \df^{2}(J^{2k-1}\pi)$. Then

\begin{definition}
A Lagrangian density $\Lag \in \df^{1}(J^{k}\pi)$ is {\rm regular} if the
Poincar\'e-Cartan $2$-form $\Omega_\Lag$ has maximal rank. Elsewhere $\Lag$ is singular.
\end{definition}

In natural coordinates, the local expression of the $2$-form $\Omega_\Lag$ is
\begin{align}
\Omega_\Lag = &\sum_{r=1}^{k} \sum_{i=0}^{k-r}(-1)^{i+1} \left( d_T^i \left( \derpars{L}{t}{q_{r+i}^A}\d t + \derpars{L}{q_j^B}{q_{r+i}^A}\d q_j^B \right) \wedge (\d q_{r-1}^A - q_r^A\d t) \right. \nonumber \\
\label{eqn:LocalCoordPoincareCartan2Form}
&\quad \left. - d_T^i\left(\derpar{L}{q^A_{r+i}} \right) \d q_r^A \wedge \d t \right)  - \derpar{L}{q_j^B}\d q_j^B \wedge \d t \, .
\end{align}
From this expression in local coordinates, we can see that the regularity
condition for $\Lag$ is equivalent to
$$
\det\left( \derpars{L}{q_k^B}{q_k^A} \right)(\bar{y}) \neq 0 \, ,
$$
for every $\bar{y} \in J^{2k-1}\pi$. Thus, this notion of regularity is
equivalent to the one given before. Geometrically, $\Lag$ is regular if, and only if,
$(\Omega_\Lag,(\bar{\pi}^{2k-1})^*\eta)$ is a cosymplectic structure on $J^{2k-1}\pi$,
that is, $\Omega_\Lag$ and $(\bar{\pi}^{2k-1})^*\eta$ are both closed and
$\Omega_\Lag^{kn} \wedge (\bar{\pi}^{2k-1})^*\eta$ is a volume form.


\subsection{Dynamical equations for sections}
\label{subsection:LagrangianSections}

Using the previous results, we can recover the Lagrangian sections
in $J^{2k-1}\pi$ from the sections in the unified formalism.

\begin{prop}
\label{prop:LagrangianSections}
Let $\psi_o \in \Gamma(\rho_\R^o)$ be a holonomic section solution to equation
\eqref{eqn:DynEquationSections}. Then the section
$\psi_\Lag = \rho_1^o \circ \psi_o \in \Gamma(\bar{\pi}^{2k-1})$
is holonomic, and is a solution
to the equation
\begin{equation}
\label{eqn:LagrangianDynEqSections}
\psi_\Lag^*\inn(Y)\Omega_\Lag = 0, \quad \mbox{for every }Y \in \vf(J^{2k-1}\pi)
\end{equation}
\end{prop}
\begin{proof}
Since, by definition, $\psi_o \in \Gamma(\rho_\R^o)$ is holonomic if
$\rho_1^o \circ \psi_o\in \Gamma(\bar{\pi}^{2k-1})$ is holonomic, it is
obvious that $\psi_\Lag = \rho_1^o \circ \psi_o$ is a holonomic section.

Now, recall that, since $\rho_1^o$ is a submersion, for every
$Y \in \vf(J^{2k-1}\pi)$ there exist some $Z \in \vf(\W_o)$ such that
${\rho_1^o}_*Z = Y$, that is, $Y$ and $Z$ are $\rho_1^o$-related. Note that
this vector field is not unique, since $Z + Z_o$, with $Z_o \in \ker{\rho_1^o}_*$
is also $\rho_1^o$-related with $Y$. Thus, using this particular choice of
$\rho_1^o$-related vector fields, we have
\begin{equation*}
\psi_\Lag^*\inn(Y)\Omega_\Lag = (\rho_1^o \circ \psi_o)^*\inn(Y)\Omega_\Lag =
\psi_o^*((\rho_1^o)^*\inn(Y)\Omega_\Lag) = \psi_o^*(\inn(Z)(\rho_1^o)^*\Omega_\Lag) =
\psi_o^*i(Z)\Omega_o \, .
\end{equation*}
Since the equality $\psi_o^*\inn(Z)\Omega_o = 0$ holds
for every $Z \in \vf(\W_o)$, in particular it holds for every $Z \in \vf(\W_o)$
which is $\rho_1^o$-related with $Y \in \vf(J^{2k-1}\pi)$. Hence, we obtain
$$
\psi_\Lag^*\inn(Y)\Omega_\Lag = \psi_o^*\inn(Z)\Omega_o = 0 \, .
$$
\end{proof}

The diagram for this situation is the following:
$$
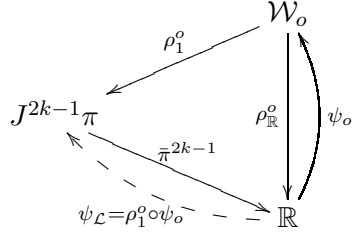
\xymatrix{
\ & \ & \W_o \ar[dd]_-{\rho_\R^o} \ar[dll]_-{\rho_1^o}  \\
J^{2k-1}\pi \ar[drr]^{\bar{\pi}^{2k-1}} & \ & \ \\
\ & \ & \R \ar@/_1pc/[uu]_{\psi_o} \ar@/^1pc/@{-->}[ull]^{\psi_\Lag = \rho_1^o \circ \psi_o} \\
}
$$

\textbf{Remark}:
Observe that, from this result, we have no equivalence between
section $\psi_o \in \Gamma(\rho_\R^o)$ solutions to equation
\eqref{eqn:DynEquationSections} and section
$\psi_\Lag \in \Gamma(\bar{\pi}^{2k-1})$ solutions to equation
\eqref{eqn:LagrangianDynEqSections}, but only that every holonomic section
$\psi_o$ solution to the dynamical equations in the unified formalism
can be projected to a holonomic section $\psi_\Lag$ solution to
the Lagrangian equations. Nevertheless, recall that section $\psi_o$
solutions to equation \eqref{eqn:DynEquationSections} take their values
in the submanifold $\W_1$, which is diffeomorphic to $J^{2k-1}\pi$,
and thus it is possible to establish an equivalence using the
diffeomorphism $\rho_1^1$.

Assume $\psi_o \in \Gamma(\rho_\R^o)$ is given locally
by $\psi_o(t) = (t,q_i^A(t),q_j^A(t),p_A^i(t))$,
$0 \leqslant i \leqslant k-1$, $k \leqslant j \leqslant 2k-1$.
Since $\psi_o$ is assumed to be a holonomic section solution
to equation \eqref{eqn:DynEquationSections}, it must satisfy equations
\eqref{eqn:LocalCoordMomentumDiffEq1}, \eqref{eqn:LocalCoordMomentumDiffEq2}
and \eqref{eqn:LocalCoordHolonomyK}. The last group of equations is
automatically satisfied because of the holonomy condition.
Now, bearing in mind that the section $\psi_o$ takes values in the
submanifold $\W_1$, and the characterization of $\W_1$ given
in Proposition \ref{prop:W1GraphFLSections}, equations
\eqref{eqn:LocalCoordMomentumDiffEq1} and \eqref{eqn:LocalCoordMomentumDiffEq2}
can be $\rho_1^o$-projected to $J^{2k-1}\pi$, thus giving the following equations
for the section $\psi_\Lag = \rho_1^o \circ \psi_o$:
$$
\restric{\derpar{L}{q_0^A}}{\psi_\Lag} - \restric{\frac{\d}{\d t}\derpar{L}{q_1^A}}{\psi_\Lag}
+ \restric{\frac{\d^2}{\d t^2}\derpar{L}{q_2^A}}{\psi_\Lag} + \ldots +
(-1)^k \restric{\frac{\d^k}{\d t^k}\derpar{L}{q_k^A}}{\psi_\Lag} = 0 \, .
$$
Finally, bearing in mind that $\psi_\Lag$ is holonomic in $J^{2k-1}\pi$, there
exists a section $\phi \in \Gamma(\pi)$, whose local expression is $\phi(t) = (t,q_0^A(t))$,
such that $j^{2k-1}\phi = \psi_\Lag$, and
thus the above equations can be rewritten in the following form
\begin{equation}
\label{eqn:Euler-LagrangeEquations}
\restric{\derpar{L}{q_0^A}}{j^{2k-1}\phi} - \restric{\frac{\d}{\d t}\derpar{L}{q_1^A}}{j^{2k-1}\phi}
+ \restric{\frac{\d^2}{\d t^2}\derpar{L}{q_2^A}}{j^{2k-1}\phi} + \ldots +
(-1)^k \restric{\frac{\d^k}{\d t^k}\derpar{L}{q_k^A}}{j^{2k-1}\phi} = 0 \, .
\end{equation}
Therefore, we obtain the Euler-Lagrange equations for a $k$th order non-autonomous system.
As stated before, equation \eqref{eqn:Euler-LagrangeEquations} may or may not be compatible,
and in this last case a constraint algorithm must be used in order to obtain
a submanifold $S_f \hookrightarrow J^{2k-1}\pi$ (if such submanifold exists)
where the equations can be solved.


\subsection{Dynamical equations for vector fields}

Now, using the results stated at the beginning of the Section,
we can recover a vector field solution to the Lagrangian equations
starting from a vector field solution to the equation in the unified
formalism. First we have:

\begin{lem}
\label{lemma:CorrepondenceXoXL}
Let $X_o \in \vf(\W_o)$ be a vector field tangent to $\W_1$.
Then there exists a unique vector field $X_\Lag \in \vf(J^{2k-1}\pi)$
such that $X_\Lag \circ \rho_1^o \circ j_1 = \Tan\rho_1^o \circ X_o \circ j_1$.
\end{lem}
\begin{proof}
Since $X_o$ is tangent to $\W_1$, there exists a vector field $X_1 \in \vf(\W_1)$
such that $\Tan j_1 \circ X_1 = X_o \circ j_1$. Furthermore, as $\rho_1^1$ is a
diffeomorphism, there is a unique vector field $X_\Lag \in \vf(J^{2k-1}\pi)$ which
is $\rho_1^1$-related with $X_1$; that is, $X_\Lag \circ \rho_1^1 = \Tan\rho_1^1 \circ X_1$.
Then
$$
X_\Lag \circ \rho_1^o \circ j_1 = X_\Lag \circ \rho_1^1 =
\Tan\rho_1^1 \circ X_1 = \Tan\rho_1^o \circ \Tan j_1 \circ X_1 =
\Tan \rho_1^o \circ X_o \circ j_1 \, .
$$
\end{proof}

The above result states that
for every $X_o \in \vf_{\W_1}(\W_o)$
there exists a vector field $X_\Lag \in \vf(J^{2k-1}\pi)$
such that the following diagram commutes
$$
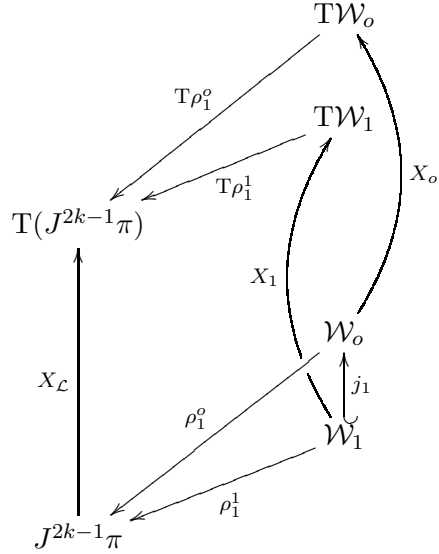
\xymatrix{
\ & \ & \Tan\W_o \ar[ddll]_-{\Tan\rho_1^o} \\
\ & \ & \Tan\W_1 \ar[dll]^-{\Tan\rho_1^1} \\
\Tan(J^{2k-1}\pi) & \ & \ \\
\ & \ & \W_o \ar[ddll]_{\rho_1^o} \ar@/_1.8pc/[uuu]_-{X_o} \\
\ & \ & \W_1 \ar[dll]^-{\rho_1^1} \ar@/^1.8pc/[uuu]^-{X_1}|(.225){\hole} \ar@{_{(}->}[u]_{j_1} \\
J^{2k-1}\pi \ar[uuu]^-{X_\Lag} & \ & \
}
$$

As a consequence we obtain:

\begin{teor}
\label{thm:Unified-LagrangianSolutions}
Let $X_o \in \vf_{\W_1}(\W_o)$ be a vector field solution to equations
\eqref{eqn:DynEqVectorFieldsSuppW1} and tangent to $\W_1$
(at least on the points of a submanifold $\W_f \hookrightarrow \W_1$).
Then there exists a unique semispray of type $k$, $X_\Lag \in \vf(J^{2k-1}\pi)$,
which is a solution to the equations
\begin{equation}
\label{eqn:LagrangianDynEqVectorFields}
\inn(X_\Lag)\Omega_{\Lag} = 0 \quad , \quad
\inn(X_\Lag)(\bar{\pi}^{2k-1})^*\eta = 1
\end{equation}
(at least on the points of $S_f = \rho_1^o(\W_f)$).
In addition, if $\Lag \in \df^{1}(J^{k}\pi)$ is a regular
Lagrangian density, then $X_\Lag$ is a semispray of type $1$.

Conversely, if $X_\Lag \in \vf(J^{2k-1}\pi)$ is a semispray of type $k$
(resp., of type $1$), which is a solution to equations
\eqref{eqn:LagrangianDynEqVectorFields}
(at least on the points of a submanifold $S_f \hookrightarrow J^{2k-1}\pi$),
then there exists a unique vector field $X_o \in \vf_{\W_1}(\W_o)$ which
is a solution to equations \eqref{eqn:DynEqVectorFieldsSuppW1}
(at least on the points of
$\W_f = (\rho_1^1)^{-1}(S_f) \hookrightarrow \W_1 \hookrightarrow \W_o$),
and it is a semispray of type $k$ in $\W_o$ (resp., of type $1$).
\end{teor}
\begin{proof}
Applying Lemmas \ref{lemma:RelationFormsLagrangianUnified}
and \ref{lemma:CorrepondenceXoXL}, we have:
$$
0 = \restric{\inn(X_o)\Omega_o}{\W_1} = \restric{\inn(X_o)(\rho_1^o)^*\Omega_\Lag}{\W_1} =
\restric{(\rho_1^o)^*\inn(X_\Lag)\Omega_\Lag}{\W_1} \, ,
$$
$$
1 = \restric{\inn(X_o)(\rho_\R^o)^*\eta}{\W_1} = \restric{\inn(X_o)(\bar{\pi}^{2k-1}\circ \rho_1^o)^*\eta}{\W_1} =
\restric{(\rho_1^o)^*\inn(X_\Lag)(\bar{\pi}^{2k-1})^*\eta}{\W_1} \, .
$$
However, as $\rho_1^o$ is a surjective submersion, this is equivalent to
$$
0 = \restric{\inn(X_\Lag)\Omega_\Lag}{\rho_1^o(\W_1)} =
\restric{\inn(X_\Lag)\Omega_\Lag}{J^{2k-1}\pi} \, ,
$$
$$
1 = \restric{\inn(X_\Lag)(\bar{\pi}^{2k-1})^*\eta}{\rho_1^o(\W_1)} =
\restric{\inn(X_\Lag)(\bar{\pi}^{2k-1})^*\eta}{J^{2k-1}\pi} \, ,
$$
since $\rho_1^o(\W_1) = J^{2k-1}\pi$ (or the submanifold $S_f \hookrightarrow J^{2k-1}\pi$).
The converse is immediate, reversing this reasoning.

In order to prove that $X_\Lag$ is a semispray of type $k$, we
compute its local expression in coordinates. From the local expression
\eqref{eqn:LocalCoordSolutionVectorField} for the vector field $X_o$
(where the functions $F_j^A$ are the solutions of equations
\eqref{eqn:TangencyVectorFieldXo}), and using Lemma
\ref{lemma:CorrepondenceXoXL}, we obtain that the local expression
of the vector field $X_\Lag \in \vf(J^{2k-1}\pi)$ is
$$
X_\Lag = \derpar{}{t} + \sum_{i=0}^{k}q_{i+1}^A\derpar{}{q_i^A} + \sum_{j=k}^{2k-1}F_j^A\derpar{}{q_j^A} \, ,
$$
which is the local expression for a semispray of type $k$ in $J^{2k-1}\pi$.

Finally, if $\Lag \in \df^{1}(J^{k}\pi)$ is a regular Lagrangian density,
equations \eqref{eqn:TangencyVectorFieldXo} become
\eqref{eqn:TangencyConditionRegularLagrangian}, and hence the local expression
of $X_\Lag$ is
$$
X_\Lag = \derpar{}{t} + \sum_{i=0}^{2k-2}q_{i+1}^A\derpar{}{q_i^A} + F_{2k-1}^A\derpar{}{q_{2k-1}^A} \, ,
$$
which is the local expression for a semispray of type $1$ in $J^{2k-1}\pi$.
\end{proof}

\textbf{Remarks}:

\begin{itemize}
\item
It is important to point out that, if $\Lag$ is not a regular
Lagrangian density, then $X_o$ is a semispray of type $k$ in $\W_o$, but
not necessarily a semispray of type $1$. This means that $X_\Lag$ may be a
solution to the Lagrangian equations for vector fields, but the trajectories
given by its integral sections are not solutions to the dynamical system
(the sections solution to the dynamical problem must be holonomic,
but the integral sections of $X_\Lag$ are only holonomic of type $k$).
Thus, for singular Lagrangians, this must be imposed as an additional condition.
This constitutes a relevant difference from the case of first-order dynamical systems,
where this condition ($X_\Lag$ is a semispray of type $1$) is obtained
straightforwardly in the unified formalism.

For singular Lagrangians, only in the most interesting cases can we assure
the existence of a submanifold $\W_f \hookrightarrow \W_1$ and vector fields
$X_o \in \vf_{\W_1}(\W_o)$ tangent to $\W_f$ which are solutions to equations
\eqref{eqn:DynEqVectorFieldsSuppWf}. Then, considering the submanifold
$S_f = \rho_1^1(\W_f) \hookrightarrow J^{2k-1}\pi$, in the best cases we have
that those semisprays of type $1$ $X_\Lag$ exist, perhaps on another submanifold
$M_f \hookrightarrow S_f$ where they are tangent, and are solutions to equations
\begin{equation}
 \label{eqn:LagrangianDynEqVectorFieldsSupportMf}
\restric{\inn(X_\Lag)\Omega_\Lag}{M_f} = 0 \quad , \quad
\restric{\inn(X_\Lag)(\bar{\pi}^{2k-1})^*\eta}{M_f} = 1 \, .
\end{equation}

\item
Notice that Theorem \ref{thm:Unified-LagrangianSolutions} states that
there is a one-to-one correspondence between vector field $X_o \in \vf_{\W_1}(\W_o)$
solutions to equations \eqref{eqn:DynEqVectorFieldsSuppW1} and vector field
$X_\Lag \in \vf(J^{2k-1}\pi)$ solutions to \eqref{eqn:LagrangianDynEqVectorFields},
but not uniqueness. In fact, we cannot assure uniqueness of the vector field $X_\Lag$
unless the Lagrangian density is regular, as we can see in the following result:
\end{itemize}

\begin{corol}
If the Lagrangian density $\Lag\in\df^{1}(J^{k}\pi)$ is regular,
then there is a unique semispray of type $1$, $X_\Lag \in \vf(J^{2k-1}\pi)$,
which is a solution to equations \eqref{eqn:LagrangianDynEqVectorFields}.
\end{corol}
\begin{proof}
If the Lagrangian density $\Lag \in \df^{1}(J^{k}\pi)$ is regular, using
Proposition \ref{prop:RegularLagrangianTangencyCondition}, there exists a unique
semispray of type $1$, $X_o \in \vf(\W_o)$, solution to equations
\eqref{eqn:DynEqVectorFieldsSuppW1} and tangent to $\W_1$.
Then, using Theorem \ref{thm:Unified-LagrangianSolutions},
there is a unique vector field $X_\Lag \in \vf(J^{2k-1}\pi)$,
which is a semispray of type $1$ in $J^{2k-1}\pi$ and is a
solution to equations \eqref{eqn:LagrangianDynEqVectorFields}.
\end{proof}

In other words, uniqueness of the vector field $X_\Lag$ is a consequence of uniqueness of $X_o$.

Finally, as a consequence of Theorem \ref{thm:EquivalenceTheoremUnified}
and the results stated in this Section, we obtain:

\begin{teor}
\label{thm:EquivalenceTheoremLagrangian}
The following assertions on a section $\phi \in \Gamma(\pi)$ are equivalent:
\begin{enumerate}
\item $j^{2k-1}\phi$ is a solution to equation \eqref{eqn:LagrangianDynEqSections}, that is,
$$(j^{2k-1}\phi)^*\inn(Y)\Omega_\Lag = 0, \quad \mbox{for every }Y \in \vf(J^{2k-1}\pi) \, .$$
\item In natural coordinates, if $\phi = (t,q_0^A(t))$, then
$j^{2k-1}\phi = (t,q_0^A(t),q_1^A(t),\ldots,q_{2k-1}^A(t))$
is a solution to the $k$th order Euler-Lagrange equations given by
\eqref{eqn:Euler-LagrangeEquations}, that is,
$$
\restric{\derpar{L}{q_0^A}}{j^{2k-1}\phi} - \restric{\frac{\d}{\d t}\derpar{L}{q_1^A}}{j^{2k-1}\phi}
+ \restric{\frac{\d^2}{\d t^2}\derpar{L}{q_2^A}}{j^{2k-1}\phi} + \ldots +
(-1)^k \restric{\frac{\d^k}{\d t^k}\derpar{L}{q_k^A}}{j^{2k-1}\phi} = 0 \, .
$$
\item Denoting $\psi_\Lag = j^{2k-1}\phi$, then $\psi_\Lag$ is a solution to the equation
$$\inn(\psi^\prime_\Lag)(\Omega_\Lag \circ \psi_\Lag) = 0 \, ,$$
where $\psi^\prime_\Lag \colon \R \to \Tan(J^{2k-1}\pi)$ is the
canonical lifting of $\psi_\Lag$ to the tangent bundle.
\item $j^{2k-1}\phi$ is an integral curve of a vector field
contained in a class of $\bar{\pi}^{2k-1}$-transverse semisprays of type $1$,
$\left\{ X_\Lag \right\} \subset \vf(J^{2k-1}\pi)$,
satisfying the first equation in \eqref{eqn:LagrangianDynEqVectorFields}, that is,
$$\inn(X_\Lag)\Omega_\Lag = 0 \, .$$
\end{enumerate}
\end{teor}


\section{Hamiltonian formalism}
\label{section:hamform}


\subsection{General setting}

In order to describe the Hamiltonian formalism on the basis of the unified one,
we must distinguish between the regular and non-regular cases.
In fact, the only ``non-regular'' case we consider is the almost-regular one,
so we need to define the concept of \textsl{almost-regular Lagrangian density}.

Before doing so, we must define the generalization of the Legendre map from the
first-order time-dependent case. Since $\Theta_\Lag \in \df^{1}(J^{2k-1}\pi)$
is a $\pi^{2k-1}_{k-1}$-semibasic $1$-form, we can give the following definition:

\begin{definition}
The {\rm extended Legendre-Ostrogradsky map} associated with the Lagrangian
density $\Lag$ is the map
$\widetilde{\Leg} \colon J^{2k-1}\pi \to \Tan^*(J^{k-1}\pi)$ defined as follows:
for every $u \in \Tan(J^{2k-1}\pi)$,
$$
\Theta_\Lag(u) = \left\langle \Tan\pi^{2k-1}_{k-1}(u) \mid \widetilde{\Leg}(\tau_{J^{2k-1}\pi}(u)) \right\rangle \, ,
$$
where $\tau_{J^{2k-1}\pi} \colon \Tan(J^{2k-1}\pi) \to J^{2k-1}\pi$ is the canonical
submersion.
\end{definition}

This map verifies that $\pi_{J^{k-1}\pi} \circ \widetilde{\Leg} = \pi^{2k-1}_{k-1}$,
where $\pi_{J^{k-1}\pi} \colon \Tan^*(J^{k-1}\pi) \to J^{k-1}\pi$ is the natural projection.
Furthermore, if $\Theta_{k-1} \in \df^{1}(\Tan^*(J^{k-1}\pi))$ and
$\Omega_{k-1} = -\d\Theta_{k-1} \in \df^{2}(\Tan^*(J^{k-1}\pi))$ are the
canonical $1$ and $2$ forms of the cotangent bundle $\Tan^*(J^{k-1}\pi)$,
we have that
$$
\widetilde{\Leg}^*\Theta_{k-1} = \Theta_\Lag \quad , \quad
\widetilde{\Leg}^*\Omega_{k-1} = \Omega_\Lag \, .
$$

Bearing in mind the local expression
\eqref{eqn:LocalCoordCanonicalSymplecticForm}
of the tautological $1$-form on $\Tan^*(J^{k-1}\pi)$
and the local expression
\eqref{eqn:LocalCoordPoincareCartan1Form} of $\Theta_\Lag$,
we have that the local expression of the map
$\widetilde{\Leg}$ is:
$$
\widetilde{\Leg}^*t = t \quad , \quad
\widetilde{\Leg}^*q_r^A = q_r^A \, ,
$$
$$
\widetilde{\Leg}^*p = L - \sum_{r=1}^kq_r^A\sum_{i=0}^{k-r}(-1)^id_T^i\left( \derpar{L}{q_{r+i}^A} \right) \quad , \quad
\widetilde{\Leg}^*p_A^{r-1} = \sum_{i=0}^{k-r}(-1)^i d_T^i\left( \derpar{L}{q_{r+i}^A} \right) \, ,
$$
that is, this map coincides with the extended Legendre-Ostrogradsky map
defined locally in Section \ref{subsubsection:DynEqSections}, thus
justifying the notation and terminology introduced therein.

Notice that $\dim\Tan^*(J^{k-1}\pi) = 2kn+2 > 2kn+1 = \dim J^{2k-1}\pi$. Thus,
$\Tan^*(J^{k-1}\pi)$ is not a suitable dual bundle to $J^{2k-1}\pi$
for giving a Hamiltonian description of the dynamical system.
Therefore, according to, for instance, \cite{art:Roman09}
and the references therein, we consider the bundle
$J^{k-1}\pi^* = \Tan^*(J^{k-1}\pi)/(\bar{\pi}^{k-1})^*\Tan^*\R$,
with the natural projections
$$
\mu \colon \Tan^*(J^{k-1}\pi) \to J^{k-1}\pi^* \quad , \quad
\pi_{J^{k-1}\pi}^r \colon J^{k-1}\pi^* \to J^{k-1}\pi \quad , \quad
\bar{\tau} = \pi_{J^{k-1}\pi}^r \circ \bar{\pi}^{k-1} \colon J^{k-1}\pi^* \to \R
\, ,
$$
where $\pi_{J^{k-1}\pi}^r$ is the map satisfying $\pi_{J^{k-1}\pi} = \pi_{J^{k-1}\pi}^r \circ \mu$.
Notice that $\dim J^{k-1}\pi^* = 2kn+1$.

Thus, we define the {\sl restricted Legendre-Ostrogradsky map} as
$\Leg = \mu \circ \widetilde{\Leg} \colon J^{2k-1}\pi \to J^{k-1}\pi^*$.
This map satisfies $\pi_{J^{k-1}\pi}^r \circ \Leg = \pi^{2k-1}_{k-1}$, and has
the following local expression
$$
\Leg^*t = t \quad , \quad
\Leg^*q_r^A = q_r^A \quad , \quad
\Leg^*p_A^{r-1} = \sum_{i=0}^{k-r}(-1)^i d_T^i\left( \derpar{L}{q_{r+i}^A} \right) \, .
$$
In other words, this map coincides with the restricted Legendre-Ostrogradsky map
defined locally in Section \ref{subsubsection:DynEqSections}. This
justifies the notation and terminology introduced in that Section.

\begin{prop}
\label{prop:RankExtendedRestrictedLegTrans}
For every $\bar{y} \in J^{2k-1}\pi$ we have that
 ${\rm rank}(\widetilde{\Leg}(\bar{y})) = {\rm rank}(\Leg(\bar{y}))$.
\end{prop}

We do not prove this result.
Following the patterns in \cite{art:DeLeon_Marin_Marrero96},
the idea is to compute in natural coordinates
the local expressions of the Jacobian matrices of $\Leg$ and $\widetilde{\Leg}$.
Then, observe that the ranks of both maps depend on the rank of the Hessian matrix
of $L$ with respect to $q_k^A$ at the point $\bar{y}$, and that the additional row in the Jacobian matrix
of $\widetilde{\Leg}$ is a linear combination of the others. See \cite{art:DeLeon_Marin_Marrero96}
for details in the first-order case.

As a consequence of Proposition \ref{prop:RankExtendedRestrictedLegTrans}, and taking into
account the different definitions given for
the regularity of the Lagrangian density, we arrive at the following result:

\begin{prop}
\label{prop:CharacterizationRegularLagrangian}
Given a Lagrangian $\Lag \in \df^{1}(J^{k}\pi)$, the following
statements are equivalent:
\ben
\item $\Omega_\Lag$ has maximal rank on $J^{2k-1}\pi$.
\item The pair $(\Omega_\Lag,(\bar{\pi}^{2k-1})^*\eta)$ is a cosymplectic structure on $J^{2k-1}\pi$.
\item $\Leg \colon J^{2k-1}\pi \to J^{k-1}\pi^*$ is a local diffeomorphism.
\item $\widetilde{\Leg} \colon J^{2k-1}\pi \to \Tan^*(J^{k-1}\pi)$ is an immersion.
\een
\end{prop}
\begin{proof}
It is easy to check that all the statements are locally equivalent to
$$
\det\left( \derpars{L}{q_k^B}{q_k^A} \right)(\bar{y}) \neq 0, \quad \mbox{for every } \bar{y} \in J^{k}\pi \, .
$$
\end{proof}

Now, we denote by
$\widetilde{\P} = {\rm Im}(\widetilde{\Leg}) = \widetilde{\Leg}(J^{2k-1}\pi)
\stackrel{\tilde{\jmath}}{\hookrightarrow} \Tan^*(J^{k-1}\pi)$
the image of the extended Legendre-Ostrogradsky map;
and by $\P = {\rm Im}(\Leg) = \Leg(J^{2k-1}\pi)
\stackrel{\jmath}{\hookrightarrow} J^{k-1}\pi^*$ the image of the
restricted Legendre-Ostrogradsky map.
Let $\bar{\tau}_o = \bar{\tau} \circ \jmath \colon \P \to \R$ be the natural projection.
We can now give the following definition:

\begin{definition}
A Lagrangian $\Lag \in \df^{1}(J^{k}\pi)$ is called an
{\rm almost-regular Lagrangian density} if:
\begin{enumerate}
\item $\P$ is a closed submanifold of $J^{k-1}\pi^*$.
\item $\Leg$ is a submersion onto its image.
\item For every $\bar{y} \in J^{2k-1}\pi$, the fibers $\Leg^{-1}(\Leg(\bar{y}))$ are connected
submanifolds of $J^{2k-1}\pi$.
\end{enumerate}
\end{definition}

As a consequence of Prop. \ref{prop:RankExtendedRestrictedLegTrans},
we have that $\widetilde{\P}$ is diffeomorphic to $\P$.
This diffeomorphism is just $\mu$ restricted to the image set $\widetilde{\P}$,
and we denote it by $\widetilde{\mu}$.
This enables us to state:

\begin{lem}
\label{lemma:ProjectedHamiltonianSection}
If the Lagrangian density $\Lag \in \df^{1}(J^{k}\pi)$ is,
at least, almost-regular,
the Hamiltonian section $\hat{h} \in \Gamma(\mu_\W)$
induces a Hamiltonian section $h \in \Gamma(\mu)$ defined by
\begin{equation}
h([\alpha]) = (\rho_2 \circ \hat{h})([(\rho_2^r)^{-1}(\jmath([\alpha]))]),\quad \mbox{for every } [\alpha] \in \P.
\end{equation}
\end{lem}
\begin{proof}
It is clear that, given $[\alpha] \in J^{k-1}\pi^*$, the section
$\hat{h}$ maps every point $(\bar{y},[\alpha]) \in (\rho^r_2)^{-1}([\alpha])$
into $\rho_2^{-1}[\rho_2(\hat{h}(\bar{y},[\alpha]))]$.
So we have the diagram
$$
\xymatrix{
\widetilde{\mathcal{P}} \ar[rr]^{\tilde{\jmath}} \ar[d]^-{\tilde{\mu}} & \ & \Tan^*(J^{k-1}\pi) \ar[d]^-{\mu} & \ & \W \ar[d]_-{\mu_\W} \ar[ll]_-{\rho_2} \\
\mathcal{P} \ar[rr]^-{\jmath} \ar[urr]^-{h}& \ & J^{k-1}\pi^* & \ & \W_r \ar@/_0.7pc/[u]_{\hat{h}} \ar[ll]_-{\rho_2^r}
}
$$
Thus, the crucial point is
the $\rho_2$-projectability of the local function $\hat{H}$. However, since a
local base for $\ker{\rho_2}_*$ is given by
$$
\ker{\rho_2}_* = \left\langle \derpar{}{q_k^A},\ldots,\derpar{}{q_{2k-1}^A} \right\rangle \, ,
$$
we have that $\hat{H}$ is $\rho_2$-projectable if and only if
$$
p_{A}^{k-1} = \derpar{L}{q_k^A} \, .
$$
This condition is fulfilled when $[\alpha] \in \mathcal{P}$, which
implies that $\rho_2[\hat{h}((\rho_2^r)^{-1}([\alpha]))] \in \widetilde{\mathcal{P}}$.
\end{proof}

\textbf{Remark}: In the hyperregular case, we have $\P = J^{k-1}\pi^*$.

Locally, this Hamiltonian $\mu$-section is specified by the local Hamiltonian function
$H \in \Cinfty(J^{k-1}\pi^*)$, that is,
$$
h(t,q_i^A,p^i_A) = (t,q_i^A,-H,p_A^i) \, .
$$


\subsection{Hyperregular and regular systems.
Dynamical equations for sections and vector fields}
\label{subsection:HamiltonianRegularCase}

Now we analyze the case when $\Lag$ is a regular Lagrangian density,
although by simplicity we focus on the hyperregular case (the regular
case is recovered from this by restriction on the corresponding open sets
where $\Leg$ is a local diffeomorphism). This means that
the phase space of the system is $J^{k-1}\pi^*$ (or the corresponding
open sets).

In this case, we can give the explicit expression for the
local Hamiltonian function, which is
\begin{equation}
\label{eqn:LocalCoordLocalHamiltonianFunctionHamForm}
H = \sum_{i=0}^{k-2}p_A^iq_{i+1}^A + p_A^{k-1}(\Leg^{-1})^*q_k^A - (\pi_{k}^{2k-1} \circ \Leg^{-1})^*L \, .
\end{equation}

The Hamiltonian section $h$ is used to construct the
\textsl{Hamilton-Cartan forms} in $J^{k-1}\pi^*$ by making
\begin{equation*}
\Theta_h = h^*\Theta_{k-1} \in \df^{1}(J^{k-1}\pi^*) \quad , \quad
\Omega_h = h^*\Omega_{k-1} \in \df^{2}(J^{k-1}\pi^*) \, ,
\end{equation*}
where $\Theta_{k-1}$ and $\Omega_{k-1}$ are the canonical $1$ and $2$ forms
of the cotangent bundle $\Tan^*(J^{k-1}\pi)$.
Bearing in mind the local expression \eqref{eqn:LocalCoordCanonicalSymplecticForm}
of $\Theta_{k-1}$ and $\Omega_{k-1}$, the local expression of the forms $\Theta_{h}$
and $\Omega_{h}$ is
\begin{equation*}
\Theta_h = p_A^i\d q_i^A - H\d t \quad , \quad
\Omega_h = \d q_i^A \wedge \d p_A^i + \d H \wedge \d t \, ,
\end{equation*}

Notice that $\Leg^*\Theta_h = \Theta_\Lag$ and $\Leg^*\Omega_h = \Omega_\Lag$.

\begin{prop}
\label{prop:DiffeomorphismHamiltonianHyperregularCase}
If $\Lag \in \df^{1}(J^{k}\pi)$ is a hyperregular Lagrangian,
then $\hat{\rho}_2^1 = \hat{\rho}_2^o \circ j_1 \colon \W_1 \to J^{k-1}\pi^*$
is a diffeomorphism.
\end{prop}
\begin{proof}
The following diagram is commutative
$$
\xymatrix{
\ & \ & \W_o \ar@/_1.3pc/[ddll]_-{\rho_1^o} \ar@/^1.3pc/[ddrr]^-{\hat{\rho}_2^o} \ & \ \\
\ & \ & \W_1 \ar[dll]_-{\rho_1^1} \ar[drr]^-{\hat{\rho}_2^1} \ar@{^{(}->}[u]^{j_1} & \ & \ \\
J^{2k-1}\pi \ar[rrrr]^-{\Leg} & \ & \ & \ & J^{k-1}\pi^*
}
$$
that is, we have $\hat{\rho}_2^1 = \hat{\rho}_2^o \circ j_1 = \Leg \circ \rho_1^1$.
Now, by Proposition \ref{prop:W1DiffeoJ2k-1Pi}, the map
$\rho_1^1$ is a diffeomorphism. In addition, as $\Lag$ is hyperregular,
the map $\Leg$ is also a diffeomorphism, and thus $\hat{\rho}_2^1$
is a composition of diffeomorphisms, and hence a diffeomorphism itself.
\end{proof}

This last result allows us to recover the Hamiltonian formalism
in the same way we recovered the Lagrangian one (see Section \ref{section:lagform}),
just using the diffeomorphism to define a correspondence between the solutions
of both equations.


Using the previous results, we can recover the Hamiltonian sections
in $J^{k-1}\pi^*$ from the sections solution to the equations in the
unified formalism.

\begin{prop}
\label{prop:HamiltonianSectionsRegular}
Let $\Lag \in \df^{1}(J^{k}\pi)$ be a hyperregular Lagrangian.
Let $\psi_o \in \Gamma(\rho_\R^o)$ be a section solution to equation
\eqref{eqn:DynEquationSections}. Then the section
$\psi_h = \hat{\rho}_2^o \circ \psi_o \in \Gamma(\bar{\tau})$
is a solution to the equation
\begin{equation}
\label{eqn:HamiltonianDynEqSections}
\psi_h^*\inn(Y)\Omega_h = 0, \quad \mbox{for every }Y \in \vf(J^{k-1}\pi^*)
\end{equation}
\end{prop}
\begin{proof}
The proof of this result is analogous to the proof
given for Proposition \ref{prop:LagrangianSections}.
\end{proof}

The diagram for this situation is the following:
$$
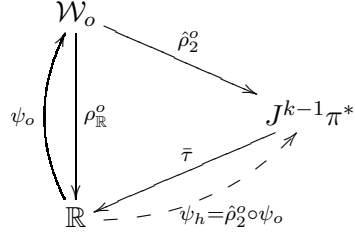
\xymatrix{
\W_o \ar[dd]^-{\rho_\R^o} \ar[drr]^-{\hat{\rho}_2^o} & \ & \  \\
\ & \ & J^{k-1}\pi^* \ar[dll]_{\bar{\tau}} \\
\R \ar@/^1pc/[uu]^{\psi_o} \ar@/_1pc/@{-->}[urr]_{\psi_h = \hat{\rho}_2^o \circ \psi_o} & \ & \ \\
}
$$

\textbf{Remarks}:
\begin{itemize}
\item
Observe that, for the Hamiltonian sections, the
condition of holonomy on the section $\psi_o$ is not required. This
is because we only need $\psi_o$ to be a holonomic section of type $k$,
and this condition is always fulfilled.
\item
As for the Lagrangian sections given by Proposition \ref{prop:LagrangianSections},
this last result does not give an equivalence between sections
$\psi_o \in \Gamma(\rho_\R^o)$, which are solutions to equation
\eqref{eqn:DynEquationSections}, and sections $\psi_h \in \Gamma(\bar{\tau})$,
which are solutions to equation \eqref{eqn:HamiltonianDynEqSections}. However, recall
that sections $\psi_o$, which are solutions to the dynamical equations in the
unified formalism, take values in $\W_1$, and hence we are able to establish the
equivalence using the diffeomorphism $\hat{\rho}_2^1$.
\end{itemize}

Let $\psi_o(t) = (t,q_i^A(t),q_j^A(t),p_A^i(t)) \in \Gamma(\rho_\R^o)$,
$0 \leqslant i \leqslant k-1$, $k \leqslant j \leqslant 2k-1$, be a solution
to equation \eqref{eqn:DynEquationSections}. Hence, $\psi_o$ must satisfy equations
\eqref{eqn:LocalCoordMomentumDiffEq1}, \eqref{eqn:LocalCoordMomentumDiffEq2} and
\eqref{eqn:LocalCoordHolonomyK}. Now, bearing in mind the local expression for
the local Hamiltonian function $H$ given in
\eqref{eqn:LocalCoordLocalHamiltonianFunctionHamForm}, we obtain
the following $2kn$ equations for the section
$\psi_h = \hat{\rho}_2^o \circ \psi_o = (t,q_i^A(t),p_A^i(t))$:
\begin{equation}
\label{eqn:HamiltonEquations}
\dot{q}_i^A = \restric{\derpar{H}{p_A^i}}{\psi_h} \quad ; \quad
\dot{p}_A^i = - \restric{\derpar{H}{q_i^A}}{\psi_h} \, .
\end{equation}
So we obtain the Hamilton equations for a $k$th-order
non-autonomous system.


Next, we recover the Hamiltonian vector field
from the vector field solution to the dynamical equations
\eqref{eqn:DynEquationVectorFields}
in the hyperregular case. As $\hat{\rho}_2^1$ is a
diffeomorphism by Proposition
\ref{prop:DiffeomorphismHamiltonianHyperregularCase},
the reasoning we follow is the same as that for the
Lagrangian formalism.

\begin{lem}
\label{lemma:CorrepondenceXoXh}
Let $\Lag \in \df^{1}(J^{k}\pi)$ be a hyperregular Lagrangian.
Let $X_o \in \vf(\W_o)$ be a vector field tangent to $\W_1$.
Then there exists a unique vector field $X_h \in \vf(J^{k-1}\pi^*)$
such that $X_h \circ \hat{\rho}_2^o \circ j_1 = \Tan\hat{\rho}_2^o \circ X_o \circ j_1$.
\end{lem}
\begin{proof}
The proof of this result is similar to the
proof given for Lemma \ref{lemma:CorrepondenceXoXL}.
\end{proof}

This result states that,
for every $X_o \in \vf_{\W_1}(\W_o)$,
we have a vector field $X_h \in \vf(J^{k-1}\pi^*)$
such that the following diagram commutes
$$
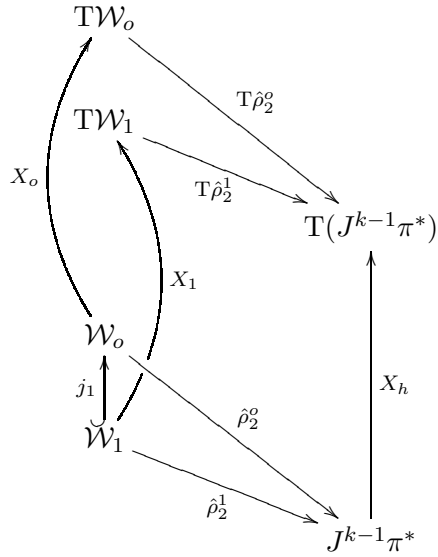
\xymatrix{
\Tan\W_o \ar[ddrr]^-{\Tan\hat{\rho}_2^o} & \ & \ \\
\Tan\W_1 \ar[drr]_-{\Tan\hat{\rho}_2^1} & \ & \ \\
\ & \ & \Tan(J^{k-1}\pi^*) \\
\W_o \ar[ddrr]^-{\hat{\rho}_2^o} \ar@/^1.8pc/[uuu]^-{X_o} & \ & \ \\
\W_1 \ar[drr]_-{\hat{\rho}_2^1} \ar@{^{(}->}[u]^-{j_1} \ar@/_1.8pc/[uuu]_-{X_1}|(.23){\hole} & \ & \ \\
\ & \ & J^{k-1}\pi^* \ar[uuu]_-{X_h} \\
}
$$


\begin{teor}
\label{thm:Unified-HamiltonianSolutions}
Let $\Lag \in \df^{1}(J^{k}\pi)$ be a hyperregular Lagrangian,
and $X_o \in \vf_{\W_1}(\W_o)$ the vector field solution to equations
\eqref{eqn:DynEqVectorFieldsSuppW1} and tangent to $\W_1$. Then, there
exists a unique vector field $X_h \in \vf(J^{k-1}\pi^*)$, which is a solution
to the equations
\begin{equation}
\label{eqn:HamiltonianDynEqVectorFields}
\inn(X_h)\Omega_h = 0 \quad , \quad \inn(X_h)\bar{\tau}^*\eta = 1 \, .
\end{equation}

Conversely, if $X_h \in \vf(J^{k-1}\pi^*)$ is a solution to equations
\eqref{eqn:HamiltonianDynEqVectorFields}, then there exists a unique
vector field $X_o \in \vf_{\W_1}(\W_o)$, tangent to $\W_1$, which is a
solution to equations \eqref{eqn:DynEqVectorFieldsSuppW1}.
\end{teor}
\begin{proof}
The proof of this result is analogous to the first part
of the proof given for Theorem \ref{thm:Unified-LagrangianSolutions},
Lemma \ref{lemma:CorrepondenceXoXh} now being used to obtain the
vector field $X_h \in \vf(J^{k-1}\pi^*)$.
\end{proof}

In local coordinates, if the vector field $X_o \in \vf_{\W_1}(\W_o)$
solution to equations \eqref{eqn:DynEqVectorFieldsSuppW1}
is given by \eqref{eqn:LocalCoordSolutionVectorFieldRegularLagrangian},
by using Lemma \ref{lemma:CorrepondenceXoXh} we obtain the local
expression for the vector field $X_h$, which is
$$
X_h = \derpar{}{t} + q_{i+1}^A\derpar{}{q_i^A} + \derpar{L}{q_0^A}\derpar{}{p_A^0}
+ d_T(p_A^i)\derpar{}{p_A^i} \, .
$$

Finally, to close the hyperregular case,
as a consequence of Theorem \ref{thm:EquivalenceTheoremUnified}
and the results stated in this Section, we obtain the following result:

\begin{teor}
\label{thm:EquivalenceTheoremHamiltonianRegular}
The following assertions on a section $\psi_h \in \Gamma(\bar{\tau})$ are equivalent:
\begin{enumerate}
\item $\psi_h$ is a solution to equation \eqref{eqn:HamiltonianDynEqSections}, that is,
$$\psi_h^*\inn(Y)\Omega_h = 0, \quad \mbox{for every }Y \in \vf(J^{k-1}\pi^*) \, .$$
\item In natural coordinates, if $\psi_h$ is given by $\psi_h(t) = (t,q_i^A(t),p_A^i(t))$,
$0 \leqslant i \leqslant k-1$, then the components
of $\psi_h$ satisfy the $k$th order Hamilton equations given by \eqref{eqn:HamiltonEquations}, that is,
$$\dot{q}_i^A = \restric{\derpar{H}{p_A^i}}{\psi_h} \quad ; \quad
\dot{p}_A^i = - \restric{\derpar{H}{q_i^A}}{\psi_h} \, .$$
\item $\psi_h$ is a solution to the equation
$$ \inn(\psi^\prime_h)(\Omega_h \circ \psi_h) = 0 \, ,$$
where $\psi^\prime_h \colon \R \to \Tan(J^{k-1}\pi^*)$
is the canonical lifting of $\psi_h$ to the tangent bundle.
\item $\psi_h$ is an integral curve of a vector field contained
in a class of $\bar{\tau}$-transverse vector fields,
$\left\{ X_h \right\} \subset \vf(J^{k-1}\pi^*)$,
satisfying the first equation in \eqref{eqn:HamiltonianDynEqVectorFields}, that is,
$$ \inn(X_h)\Omega_h = 0 \, . $$
\end{enumerate}
\end{teor}


\subsection{Singular (almost-regular) Lagrangians.
Dynamical equations for sections and vector fields}
\label{subsection:HamiltonianSingularCase}

Recall that, for almost-regular Lagrangians, only in the most
favourable cases can we assure the existence of some submanifold
$\W_f \hookrightarrow \W_1$ where the dynamical equations
can be solved.
In this case, the solutions to the Hamiltonian formalism cannot be obtained
straightforwardly from the solutions in the unified formalism, but rather by passing
through the Lagrangian formalism and using the Legendre-Ostrogradsky map.

In this case, the phase space of the system is
$\P = {\rm Im}(\Leg) \hookrightarrow J^{k-1}\pi^*$.
We denote by $\Leg_o \colon J^{2k-1}\pi \to \P$ the map defined
by $\Leg = \jmath \circ \Leg_o$.
As in the hyperregular
case, the Hamiltonian section $h$ is used to construct the Hamilton-Cartan forms
on $\P$ as follows:
$$
\Theta_h^o = h^*\Theta_{k-1} \in \df^{1}(\P) \quad , \quad
\Omega_h^o = h^*\Omega_{k-1} \in \df^{2}(\P) \, .
$$
They verify that $\Leg_o^*\Theta_h^o = \Theta_\Lag$ and $\Leg_o^*\Omega_h^o = \Omega_\Lag$.


\begin{prop}
\label{prop:HamiltonianSectionsSingular}
Let $\Lag \in \df^{1}(J^{k}\pi)$ be an almost-regular Lagrangian.
Let $\psi_o \in \Gamma(\rho_\R^o)$ be a section solution to equation
\eqref{eqn:DynEquationSections}. Then, the section
$\psi_h^o = \Leg_o \circ \psi_\Lag \in \Gamma(\bar{\tau}_o)$ is a solution
to the equation
\begin{equation}
\label{eqn:HamiltonianSingularDynEqSections}
(\psi_h^o)^*\inn(Y)\Omega_h^o = 0,\quad \mbox{for every } Y \in \vf(\P) \, .
\end{equation}
\end{prop}
\begin{proof}
Since the Lagrangian density is almost-regular, the
map $\Leg_o$ is a submersion onto its image, $\mathcal{P}$.
Hence, for every $Y \in \vf(\mathcal{P})$ there exist some $Z \in \vf(J^{2k-1}\pi)$
such that $Z$ is $\Leg_o$-related with $Y$, that is, ${\Leg_o}_*Z = Y$.
Using this, we have
$$
(\psi_h^o)^*\inn(Y)\Omega_h^o = (\Leg_o \circ \psi_\Lag)^*\inn(Y)\Omega_h^o =
\psi_\Lag^*(\Leg^*\inn(Y)\Omega_h^o) = \psi_\Lag^*\inn(Z)\Leg_o^*\Omega_h^o =
\psi_\Lag^*\inn(Z)\Omega_\Lag \, .
$$
Then, using Proposition \ref{prop:LagrangianSections}, we have proved
$$
(\psi_h^o)^*\inn(Y)\Omega_h^o = \psi_\Lag^*\inn(Z)\Omega_\Lag = 0 \, .
$$
\end{proof}

The diagram for this situation is the following:
$$
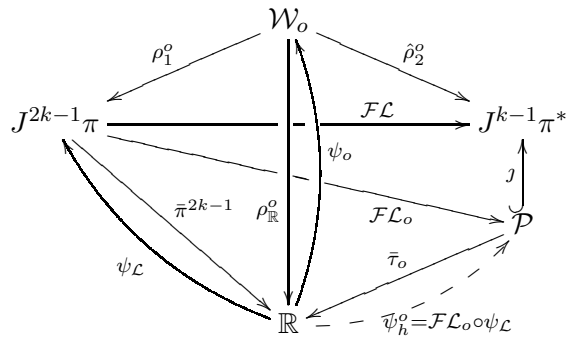
\xymatrix{
\ & \ & \W_o \ar[ddd]_<(0.6){\rho_\R^o} \ar[drr]^{\hat{\rho}_2^o} \ar[dll]_{\rho_1^o} & \ & \ \\
J^{2k-1}\pi \ar[ddrr]^{\bar{\pi}^{2k-1}} \ar[rrrr]^<(0.65){\Leg}|(.492){\hole}|(.55){\hole} \ar[drrrr]_<(0.7){\Leg_o}|(.492){\hole}|(.56){\hole} & \ & \ & \ & J^{k-1}\pi^* \\
\ & \ & \ & \ & \mathcal{P} \ar@{^{(}->}[u]^{\jmath} \ar[dll]_{\bar{\tau}_o} \\
\ & \ & \R \ar@/_1pc/[uuu]_<(0.535){\psi_o} \ar@/^1pc/[uull]^{\psi_\Lag} \ar@/_1pc/@{-->}[urr]_-{\psi^o_h = \Leg_o \circ \psi_\Lag} & \ & \ \\
}
$$


Now, assume that there exists a submanifold $\W_f \hookrightarrow \W_1$
and vector fields $X_o \in \vf_{\W_1}(\W_o)$ tangent to $\W_f$
which are solutions to equations \eqref{eqn:DynEqVectorFieldsSuppWf}.
Now consider the submanifolds
$S_f = \rho_1^1(\W_f) \hookrightarrow J^{2k-1}\pi$ and
$P_f = \hat{\rho}_2^1(\W_f) = \Leg(S_f) \hookrightarrow \P \hookrightarrow J^{k-1}\pi^*$.
Using Theorem \ref{thm:Unified-LagrangianSolutions},
from the vector fields $X_o \in \vf_{\W_1}(\W_o)$ we obtain the
corresponding vector fields $X_\Lag \in \vf(J^{2k-1}\pi)$,
and from these, the semisprays of type $1$ (if they exist), which
are perhaps defined on a submanifold $M_f \hookrightarrow S_f$,
are tangent to $M_f$ and are solutions to equations
\eqref{eqn:LagrangianDynEqVectorFieldsSupportMf}. So we have the diagram

$$
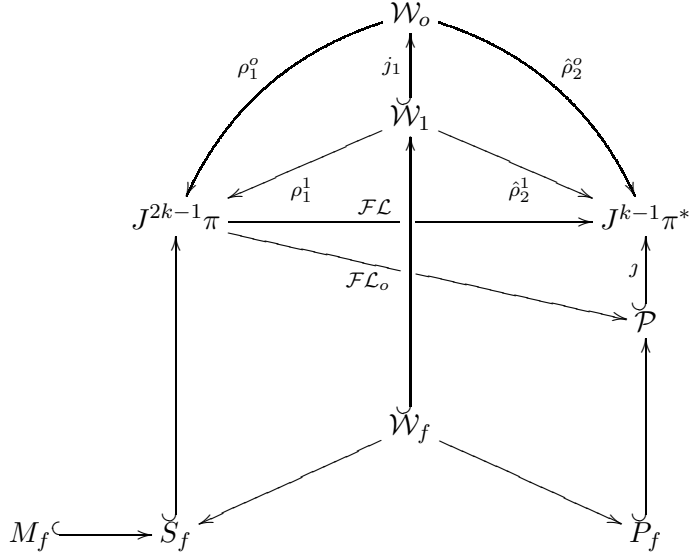
\xymatrix{
\ & \ & \ & \W_o \ar@/_1.3pc/[ddll]_{\rho_1^o} \ar@/^1.3pc/[ddrr]^{\hat{\rho}_2^o} & \ & \ \\
\ & \ & \ & \W_1 \ar[dll]^{\rho_1^1} \ar[drr]_{\hat{\rho}_2^1} \ar@{^{(}->}[u]^{j_1} & \ & \ \\
\ & J^{2k-1}\pi \ar[rrrr]^<(0.35){\Leg}|(.493){\hole} \ar[drrrr]_<(0.35){\Leg_o}|(.495){\hole} & \ & \ & \ & J^{k-1}\pi^* \\
\ & \ & \ & \ & \ & \mathcal{P} \ar@{^{(}->}[u]^-{\jmath} \\
\ & \ & \ & \W_f \ar@{^{(}->}[uuu] \ar[dll] \ar[drr] & \ & \ \\
M_f \ar@{^{(}->}[r] & S_f \ar@{^{(}->}[uuu] & \ & \ & \ & P_f \ar@{^{(}->}[uu] \\
}
$$

Now, following analogous procedures for autonomous
and non-autonomous systems 
\cite{art:deLeon_Marin_Marrero_Munoz_Roman02,art:Gracia_Pons_Roman92},
one can prove that there
are semisprays of type $1$ in $M_f$ (perhaps only on the points of another submanifold
$\bar{M}_f \hookrightarrow M_f$), which are $\Leg$-projectable on $P_f$. These vector
fields $X_h^o = \Leg_*X_\Lag \in \vf(\P)$ are tangent to $P_f$
and are solutions to equations
\begin{equation}
\label{eqn:HamiltonianDynEqVectorFieldsSupportPf}
\restric{\inn(X_h^o)\Omega_h^o}{P_f} = 0 \quad , \quad
\restric{\inn(X_h^o)\bar{\tau}_o^*\eta}{P_f} = 1
\end{equation}

Conversely, as $\Leg_o$ is a submersion, for every vector field
$X_h^o \in \vf(\P)$ solution to equations
\eqref{eqn:HamiltonianDynEqVectorFieldsSupportPf}, there is a
semispray of type $1$, $X_\Lag \in \vf(J^{2k-1}\pi)$, such that
${\Leg_o}_*X_\Lag = X_{h}^o$, and we can recover solutions to equations
\eqref{eqn:DynEqVectorFieldsSuppWf} using Theorem \ref{thm:Unified-LagrangianSolutions}.

Of course, for the almost-regular case, we have a similar result
to Theorem \ref{thm:EquivalenceTheoremHamiltonianRegular},
on the points of the final constraint submanifold $P_f$.


\section{Examples}
\label{section:Examples}


\subsection{The shape of a deformed elastic cylindrical beam with fixed ends}
\label{subsection:CylindricalBeam}

As a first example we consider a deformed elastic cylindrical beam with both
ends fixed. The problem is to determinate its shape;
that is, the width of every section transversal to the axis.
This system has been studied on many occasions,
such as \cite{book:Benson06} (Chapter 3, \S 3.9)
and \cite{book:Elsgoltz83} (Chapter IV, \S 4).
Strictly speaking, it is not
a time-dependent mechanical system, but it can be modeled
using a configuration bundle over a compact subset of $\R$,
where the base coordinate
represents every transversal section of the beam, thus allowing us
to show an application of our formalism. For simplicity,
instead of a compact subset, we take
the whole real line as the base manifold.


The configuration bundle for this system is $\pi \colon E \to \R$,
where $E$ is a $2$-dimensional smooth manifold.
Let $x$ be the global coordinate in $\R$,
and $\eta \in \df^{1}(\R)$ the volume form in $\R$ with local expression
$\eta = \d x$.
Natural coordinates in $E$ adapted to the bundle structure are $(x,q_0)$.
Now, taking natural coordinates in the higher-order jet bundle of $\pi$,
the second-order Lagrangian density for this system,
$\Lag \in \df^{1}(J^{2}\pi)$, is locally given by
$$
\Lag(x,q_0,q_1,q_2) = L\cdot(\bar{\pi}^{2})^*\eta = 
\left( \frac{1}{2}\mu(x) q_2^2 + \rho(x) q_0 \right) \d x \, ,
$$
where $\mu, \rho \in \Cinfty(J^{2}\pi)$ are functions that only depend on
the coordinate $x$ and represent physical parameters of the beam:
$\rho$ is the linear density and $\mu$ is a non-vanishing function involving
Young's modulus of the material, the radius of curvature
and the sectional moment of the cross-section considered
(see \cite{book:Benson06} for a detailed description).
This is 
a regular Lagrangian density, since the Hessian matrix of the 
Lagrangian function $L \in \Cinfty(J^{2}\pi)$ associated with $\Lag$
with respect to $q_2$ is
$$
\left( \derpars{L}{q_2}{q_2} \right) = \mu(x) \, ,
$$
and this $1 \times 1$ matrix has maximum rank, since $\mu$ is 
a non-vanishing function.

\textbf{Remark}:
If the beam is homogeneous,
$\mu$ and $\rho$ are constants (with $\mu\not=0$),
and thus the Lagrangian density is ``autonomous'', that is, it does
not depend explicitly on the coordinate of the base manifold.
This case is analyzed in \cite{book:Elsgoltz83}.

As this is a second-order system, we consider the bundles
$\W = J^{3}\pi \times_{J^{1}\pi} \Tan^*(J^{1}\pi)$
and $\W_r = J^{3}\pi \times_{J^{1}\pi} J^{1}\pi^*$,
with natural coordinates $(x,q_0,q_1,q_2,q_3,p,p^0,p^1)$
and $(x,q_0,q_1,q_2,q_3,p^0,p^1)$, respectively.
Now, using the notation and terminology introduced throughout this article,
if $\Theta_{1} \in \df^{1}(\Tan^*(J^{1}\pi))$ and
$\Omega_{1} \in \df^{2}(\Tan^*(J^{1}\pi))$ are
the canonical forms of $\Tan^*(J^{1}\pi)$,
we define the forms $\Theta = \rho_2^*\Theta_{1} \in \df^{1}(\W)$
and $\Omega = \rho_2^*\Omega_{1} \in \df^{2}(\W)$,
whose local expressions are
$$
\Theta = p^0\d q_0 + p^1\d q_1 + p\d x \quad ; \quad
\Omega = \d q_0 \wedge \d p^0 + \d q_1 \wedge \d p^1 - \d p \wedge \d x \, .
$$
The coupling $1$-form $\hat{\C} \in \df^{1}(\W)$ has the local expression
$$
\hat{\C} = \hat{C}\cdot\rho_\R^*\eta = (p + p^0q_1 + p^1q_2)\d x \ ,
$$
and then we can introduce the Hamiltonian submanifold
$$
\W_o = \left\{ w \in \W \colon \hat{\Lag}(w) = \hat{\C}(w) \right\}
\stackrel{j_o}{\hookrightarrow} \W \, ,
$$
which is locally defined by the constraint function
$\hat{C} - \hat{L} = 0$, whose coordinate expression is
$$
\hat{C}-\hat{L}=p+p^0q_1+p^1q_2-\frac{1}{2}\mu(x)q_2^2-\rho(x)q_0 = 0 \, .
$$
Finally, we construct the
Hamiltonian $\mu_\W$-section $\hat{h} \in \Gamma(\mu_\W)$, which is specified
by giving the local Hamiltonian function $\hat{H}$, whose local expression is
$$
\hat{H}(x,q_0,q_1,q_2,q_3,p^0,p^1) = 
p^0q_1 + p^1q_2 - \frac{1}{2}\mu(x) q_2^2 - \rho(x) q_0 \, ;
$$
that is, we have 
$\hat{h}(x,q_0,q_1,q_2,q_3,p^0,p^1) = (x,q_0,q_1,q_2,q_3,-\hat{H},p^0,p^1)$.
Using this Hamiltonian section, we define the forms
$\Theta_o = j_o^*\Theta \in \df^{1}(\W_o)$ and
$\Omega_o = j_o^*\Omega \in \df^{2}(\W_o)$, with local expressions
\begin{align*}
& \Theta_o = p^0\d q_0 + p^1 \d q_1 + \left(\frac{1}{2}\mu(x) q_2^2 + \rho(x) q_0 - p^0q_1 - p^1q_2 \right)\d x \, , \\
&\Omega_o = \d q_0 \wedge \d p^0 + \d q_1 \wedge \d p^1 + \left( -\rho(x)\d q_0 + p^0 \d q_1 + (p^1 - \mu(x)q_2)\d q_2 + q_1 \d p^0 + q_2 \d p^1\right) \wedge \d x \, .
\end{align*}


In order to state the Lagrangian-Hamiltonian problem for sections
in this system, let $Y \in \vf(\W_o)$ be a generic vector field
locally given by
$$
Y = f\derpar{}{x} + f_0\derpar{}{q_0} + f_1\derpar{}{q_1} + 
F_2\derpar{}{q_2} + F_3\derpar{}{q_3} + G^0\derpar{}{p^0} + 
G^1\derpar{}{p^1} \, .
$$
Now, if $\psi_o(x) = (x,q_0(x),q_1(x),q_2(x),q_3(x),p^0(x),p^1(x))$ is a
holonomic section of the projection $\rho_\R^o$,
equation \eqref{eqn:DynEquationSections} leads to 
the following $5$ equations
(the redundant equation \eqref{eqn:LocalCoordRedundantEqSections} is omitted):
\begin{align}
\label{eqn:Example1_Holonomy}
&\dot{q}_0 = q_1 \quad ; \quad \dot{q}_1 = q_2 \\
\label{eqn:Example1_DiffEquations}
&\dot{p}^0 = \rho(x) \quad ; \quad \dot{p}^1 = -p^0 \\
\label{eqn:Example1_LastMomentumCoord1}
&p^1 = q_2\mu(x)
\end{align}
Equations \eqref{eqn:Example1_Holonomy}
give us the condition of holonomy of type $2$ for the section, which are also
redundant since we assume that $\psi_o$ is holonomic.
Equation \eqref{eqn:Example1_LastMomentumCoord1} is a pointwise algebraic
condition, from which we know that the section $\psi_o$ must lie in a
submanifold $\W_1$ that can be identified with the graph of the
extended Legendre-Ostrogradsky map, $\widetilde{\Leg}$.

Now we compute the local expression of the map
$\widetilde{\Leg} \colon J^{3}\pi \to \Tan^*(J^{1}\pi)$.
From Corollary \ref{corol:W1GraphExtendedFLSections} we know the general
expression for this map, and we obtain:
\begin{equation}
\label{eqn:Example1_ExtendedLegendreMap}
\widetilde{\Leg}^*p^0 = -q_2 \derpar{\mu}{x} - q_3\mu \quad ; \quad \widetilde{\Leg}^*p^1 = q_2\mu \quad ; \quad
\widetilde{\Leg}^*p = -\frac{1}{2}\mu q_2^2 + q_1q_2 \derpar{\mu}{x} + q_1q_3\mu + q_0\rho \, .
\end{equation}

Therefore, the section $\psi_o \in \Gamma(\rho_\R^o)$ is a holonomic section of the
projection $\rho_\R^o$, which lies in the submanifold $\W_1 \hookrightarrow \W_o$
defined by the above constraint functions, and whose last components 
satisfy the differential equations
$$
\dot{p}^0 = \rho(x) \quad ; \quad \dot{p}^1 = -p^0 \, .
$$



Now we state the Lagrangian-Hamiltonian problem for vector fields:
we wish to find $X_o \in \vf(\W_o)$ solution to
\eqref{eqn:DynEquationVectorFields}. If $X_o$
is locally given by
$$
X_o = f\derpar{}{x} + f_0\derpar{}{q_0} + f_1\derpar{}{q_1} + 
F_2\derpar{}{q_2} + F_3\derpar{}{q_3} + G^0\derpar{}{p^0} + 
G^1\derpar{}{p^1} \, ,
$$
then equations \eqref{eqn:DynEquationVectorFields} lead to the following
(again, the redundant equation \eqref{eqn:LocalCoordRedundantEqVectorFields}
is omitted):
\begin{align}
\label{eqn:Example1_Semispray2}
&f_0 = f\cdot q_1 \quad ; \quad f_1 = f\cdot q_2 \\
\label{eqn:Example1_DynamicalEquations}
&G^0 = f\cdot\rho(x) \quad ; \quad G^1 = -f\cdot p^0 \\
\label{eqn:Example1_FixingGauge}
&f = 1 \\
\label{eqn:Example1_LastMomentumCoord2}
&f\cdot \left(p^1 - q_2\mu(x)\right) = 0
\end{align}
Equations \eqref{eqn:Example1_Semispray2}
give us the condition of semispray of type $2$ in $\W_o$ for $X_o$.
In addition, equation \eqref{eqn:Example1_LastMomentumCoord2} is an algebraic relation
from which we obtain, in coordinates, the result stated in Propositions \ref{prop:ExistSolDynEq}
and \ref{prop:W1GraphFLVectorFields}, that is, the vector field $X_o$ is defined along a
submanifold $\W_1$ which we identify with the graph of the extended 
Legendre-Ostrogradsky map and is defined by
$$
\W_1 = \left\{ w \in \W_o \colon \xi_0(w) = \xi_1(w) = 0 \right\} \, ,
$$
where $\xi_r = p^r - \widetilde{\Leg}^*p^r$, $r = 1,2$.
Thus, using \eqref{eqn:Example1_Semispray2},
\eqref{eqn:Example1_DynamicalEquations}
and \eqref{eqn:Example1_FixingGauge}, $X_o$ is given locally by
\begin{equation}
\label{eqn:Example1_VectorFieldBeforeTangency}
X_o = \derpar{}{x} + q_1\derpar{}{q_0} + q_2\derpar{}{q_1} + 
F_2\derpar{}{q_2} + F_3\derpar{}{q_3} + \rho\derpar{}{p^0}
- p^0\derpar{}{p^1} \, .
\end{equation}
Notice that the functions $F_2$ and $F_3$ in \eqref{eqn:Example1_VectorFieldBeforeTangency}
are not determined until the tangency of $X_o$ on $\W_1$ is required.
This condition is locally equivalent to checking if the following identities hold
$$
\restric{\Lie(X_o)\xi_0}{\W_1} = 0 \quad , \quad
\restric{\Lie(X_o)\xi_1}{\W_1} = 0 \, .
$$
As we have seen in Section \ref{subsubsection:DynEqVectFields}, these
equations lead to the Lagrangian equations for the vector field $X_o$;
that is, on the points of $\W_o$ we obtain
\begin{align}
\label{eqn:Example1_EulerLagrangeVectFieldEq}
&\Lie(X_o)\xi_0 = \rho + q_2 \derpars{\mu}{x}{x} + q_3\derpar{\mu}{x} + F_2 \derpar{\mu}{x} + F_3\mu = 0\\
\label{eqn:Example1_Semispray1}
&\Lie(X_o)\xi_1 = (q_3 - F_2)\mu = 0 \ .
\end{align}
Equation \eqref{eqn:Example1_Semispray1} gives us the condition of semispray
of type $1$ for the vector field $X_o$ (recall that $\mu$ is non-vanishing),
and equation \eqref{eqn:Example1_EulerLagrangeVectFieldEq} is the
Euler-Lagrange equation for $X_o$. Observe that, since $\mu$ is a
non-vanishing function, these equations have a unique solution for
$F_2$ and $F_3$. Hence, there is a unique vector field $X_o \in \vf(\W_o)$
solution to the equations $\restric{\inn(X_o)\Omega_o}{\W_1} = 0$
and $\restric{\inn(X_o)(\rho_\R^o)^*\eta}{\W_1} = 1$, which is tangent
to the submanifold $\W_1 \hookrightarrow \W_o$, and is given locally by
$$
X_o = \derpar{}{x} + q_1\derpar{}{q_0} + q_2\derpar{}{q_1} + q_3\derpar{}{q_2}
- \frac{1}{\mu}\left( \rho + q_2 \derpars{\mu}{x}{x} + 2q_3\derpar{\mu}{x} \right) \derpar{}{q_3}
+ \rho\derpar{}{p^0} - p^0\derpar{}{p^1} \, .
$$



Finally, we recover the Lagrangian and Hamiltonian solutions
for sections and vector fields. For the Lagrangian solutions,
by Proposition \ref{prop:LagrangianSections},
from the holonomic section
$\psi_o \in \Gamma(\rho_\R^o)$ we can recover a holonomic section
$\psi_\Lag = \rho_1^o \circ \psi_o \in \Gamma(\bar{\pi}^3)$ solution to equation
\eqref{eqn:LagrangianDynEqSections}. In particular, if
$\psi_o(x) = (x,q_0(x),q_1(x),q_2(x),q_3(x),p^0(x),p^1(x))$, then
$\psi_\Lag(x) = (x,q_0(x),q_1(x),q_2(x),q_3(x))$ is a holonomic section
solution to equations \eqref{eqn:Example1_DiffEquations}, which,
bearing in mind the local expression \eqref{eqn:Example1_ExtendedLegendreMap}
of the extended Legendre-Ostrogradsky map, can be written locally as
\begin{align}
\label{eqn:Example1_EulerLagrangeSections}
&\rho + \dot{q}_2\derpar{\mu}{x} + q_2\derpars{\mu}{x}{x} + \dot{q}_3\mu + q_3\derpar{\mu}{x} = 0 \\
\label{eqn:Example1_FullHolonomySections}
&(\dot{q}_2 - q_3)\mu = 0
\end{align}
Equation \eqref{eqn:Example1_FullHolonomySections} gives the condition
for the section $\psi_\Lag$ to be holonomic, and it is redundant since
we required this condition to be fulfilled at the beginning. Now,
if $\phi(x) = (x,y(x))$ is a section of $\pi$
such that $j^{3}\phi = \psi_\Lag$, then the Euler-Lagrange equation
can be written locally
$$
\frac{\d^2}{\d x^2}(\mu\ddot{y}) + \rho = 0 \ .
$$
In the case of an homogeneous beam, the Euler-Lagrange equation
reduces to $\mu y^{(iv)} + \rho = 0$.

For the Lagrangian vector field,
from Lemma \ref{lemma:CorrepondenceXoXL} and Theorem
\ref{thm:Unified-LagrangianSolutions}, we can recover,
from the semispray of type $1$ $X_o \in \vf(\W_o)$
a semispray of type $1$, $X_\Lag \in \vf(J^{3}\pi)$,
which is a solution to equations \eqref{eqn:LagrangianDynEqVectorFields},
and is locally given by
$$
X_\Lag = \derpar{}{x} + q_1\derpar{}{q_0} + q_2\derpar{}{q_1} + q_3\derpar{}{q_2}
- \frac{1}{\mu}\left( \rho + q_2 \derpars{\mu}{x}{x} + 2q_3\derpar{\mu}{x} \right) \derpar{}{q_3} \, .
$$

Now, as $\Lag$ is a regular Lagrangian
density, for the Hamiltonian solutions
we can use the results stated in Section
\ref{subsection:HamiltonianRegularCase} and recover the Hamiltonian
solutions directly from the unified formalism. For the Hamiltonian
sections, using Proposition \ref{prop:HamiltonianSectionsRegular},
from a section $\psi_o \in \Gamma(\rho_\R^o)$ fulfilling equation
\eqref{eqn:DynEquationSections} we can recover a section
$\psi_h = \rho_2^o \circ \psi_o \in \Gamma(\bar{\tau})$ solution
to equation \eqref{eqn:HamiltonianDynEqSections}.
In particular, if
$\psi_o(x) = (x,q_0(x),q_1(x),q_2(x),q_3(x),p^0(x),p^1(x))$,
then $\psi_h(x) = (x,q_0(x),q_1(x),p^0(x),p^1(x))$ is
a section solution to equations \eqref{eqn:Example1_Holonomy}
and \eqref{eqn:Example1_DiffEquations}, which can be written
locally as
$$
\dot{q}_0 = \restric{\derpar{H}{p^0}}{\psi_h} \quad ; \quad
\dot{q}_1 = \restric{\derpar{H}{p^1}}{\psi_h} \quad ; \quad
\dot{p}^0 = -\restric{\derpar{H}{q_0}}{\psi_h} \quad ; \quad
\dot{p}^1 = -\restric{\derpar{H}{q_1}}{\psi_h} \, .
$$
where $H \in \Cinfty(J^{1}\pi^*)$ is the local Hamiltonian
function with local expression
$$
H = p^0q_1 + \frac{(p^1)^2}{2\mu} - \rho q_0 \, .
$$

For the Hamiltonian vector field, from Lemma \ref{lemma:CorrepondenceXoXh}
and Theorem \ref{thm:Unified-HamiltonianSolutions}, the vector field
$X_o \in \vf(\W_o)$ gives a vector field $X_h \in \vf(J^1\pi^*)$
solution to equations \eqref{eqn:HamiltonianDynEqVectorFields},
which is locally given by
$$
X_h = \derpar{}{x} + q_1\derpar{}{q_0} + \frac{p^1}{\mu}\derpar{}{q_1}
+ \rho\derpar{}{p^0} - p^0\derpar{}{p^1} \ .
$$


\subsection{The second-order relativistic particle subjected to a potential}

Consider a relativistic particle whose action is proportional to its extrinsic curvature
\cite{art:Plyushchay88,art:Pisarski86,art:Batlle_Gomis_Pons_Roman88,art:Nesterenko89,art:Prieto_Roman11}.
Now assume this system is subjected to the action of a generic potential depending on
the time and the position of the particle, thus obtaining a time-dependent
dynamical system.

The configuration bundle for this system is $E \stackrel{\pi}{\to} \R$,
where $E$ is a $(n+1)$-dimensional smooth manifold. Let $t$ be the
global coordinate in $\R$, and $\eta \in \df^{1}(\R)$ the volume form
in $\R$ with local expression $\eta = \d t$. Natural coordinates in $E$
adapted to the bundle structure are denoted by $(t,q_0^i)$, $1 \leqslant i \leqslant n$.
Now, bearing in mind the natural coordinates in the higher-order jet bundle of $\pi$,
the second-order Lagrangian density for this system, $\Lag \in \df^{1}(J^{2}\pi)$,
is locally given by
$$
\Lag(t,q_0^i,q_1^i,q_2^i) = \left( \frac{\alpha}{(q_1^i)^2} \left[ (q_1^i)^2(q_2^i)^2 - (q_1^iq_2^i)^2 \right]^{1/2} + V(t,q_0^i) \right)\d t \equiv \left(\frac{\alpha}{(q_1^i)^2} \sqrt{g} + V(t,q_0^i)\right) \d t  \ ,
$$
where $\alpha$ is some nonzero constant and $V \in \Cinfty(J^{2}\pi)$
is a function depending only on $t$ and $q_0^i$. This is a singular Lagrangian density, as we
can see by computing the Hessian matrix of the Lagrangian function
$L \in \Cinfty(J^{2}\pi)$ associated with $\Lag$ with respect to $q_2^A$,
which is
$$
\frac{\partial^2 L}{\partial q_2^B\partial q_2^A} = \begin{cases}
\displaystyle \frac{\alpha}{2(q_1^i)^2\sqrt{g^3}} \left[ \left((q_1^iq_2^i)^2 - 2(q_1^i)^2(q_2^i)^2 \right)q_1^Bq_1^A \right.   & 
 \\
\displaystyle \qquad\qquad\qquad\quad \left. + (q_1^i)^2(q_1^iq_2^i)(q_2^Bq_1^A-q_1^Bq_2^A) - (q_1^i)^2(q_2^i)^2q_2^Bq_2^A \right]
& \mbox{ if } B \neq A  
\\[10pt]
\displaystyle \frac{\alpha}{\sqrt{g^3}}\left[ g - 
(q_2^i)^2q_1^Aq_1^A + 2(q_1^iq_2^i)q_1^Aq_2^A - (q_1^i)^2q_2^Aq_2^A \right] & \mbox{ if } B = A \ ,
\end{cases}
$$
and a long calculation shows that
$\ds \det\left( \frac{\partial^2 L}{\partial q_2^B\partial q_2^A} \right) = 0$.

Consider the bundles
$\W = J^3\pi \times_{J^{1}\pi} \Tan^*(J^{1}\pi)$ and
$\W_r = J^{3}\pi \times_{J^{1}\pi} J^{1}\pi^*$, with natural
coordinates $(t,q_0^i,q_1^i,q_2^i,q_3^i,p,p_i^0,p_i^1)$ and
$(t,q_0^i,q_1^i,q_2^i,q_3^i,p_i^0,p_i^1)$, respectively.
Now, if $\Theta_1 \in \df^{1}(\Tan^*(J^{1}\pi))$ and
$\Omega_{1} \in \df^2(\Tan^*(J^{1}\pi))$ are the canonical forms
of the cotangent bundle of $J^{1}\pi$, we define
$$
\Theta = \rho_2^*\Theta_1=p_i^0 \d q_0^i + p_i^1 \d q_1^i + p\d t\in \df^{1}(\W)\quad ; \quad
\Omega = \rho_2^*\Omega_1=dq_0^i \wedge dp_i^0 + dq_1^i \wedge dp_i^1 - \d p \wedge \d t
\in \df^{2}(\W) \ .
$$
The coupling $1$-form $\hat{\C} \in \df^{1}(\W)$ has the local expression
$\hat{\C} = \hat{C}\cdot \rho_\R^*\eta = (p + p_i^0q_1^i + p_i^1q_2^i)\d t$,
and from this we can introduce the Hamiltonian submanifold
$\W_o\stackrel{j_o}{\hookrightarrow} \W$,
which is locally defined by the constraint function $\hat{C} - \hat{L} = 0$,
whose coordinate expression is
$$
\hat{C} - \hat{L} = p + p_i^0q_1^i + p_i^1q_2^i - \frac{\alpha}{(q_1^i)^2} \sqrt{g} - V(t,q_0^i) \, .
$$
This allows us to construct the Hamiltonian $\mu_\W$-section $\hat{h} \in \Gamma(\mu_\W)$,
which is specified by giving the local Hamiltonian function $\hat{H}$,
whose local expression is
$$
\hat{H}(t,q_0^i,q_1^i,q_2^i,q_3^i,p_i^0,p_i^1) = p_i^0q_1^i + p_i^1q_2^i - \frac{\alpha}{(q_1^i)^2} \sqrt{g} - V(t,q_0^i) \, ,
$$
that is, we have $\hat{h}(t,q_0^i,q_1^i,q_2^i,q_3^i,p_i^0,p_i^1) = (t,q_0^i,q_1^i,q_2^i,q_3^i,-\hat{H},p_i^0,p_i^1)$.
Using this Hamiltonian section, we define the forms $\Theta_o = j_o^*\Theta \in \df^{1}(\W_o)$
and $\Omega_o = j_o^*\Omega \in \df^{2}(\W_o)$, with local expressions
\begin{align*}
&\Theta_o = p_i^0 \d q_0^i + p_i^1 \d q_1^i
+ \left( \frac{\alpha}{(q_1^i)^2} \sqrt{g} + V(t,q_0^i) - p_i^0q_1^i - p_i^1q_2^i \right) \\
&\Omega_o = \d q_0^i \wedge \d p_i^0 + \d q_1^i \wedge \d p_i^1 + \left( q_1^A\d p_A^0 + q_2^A\d p_A^1 - \derpar{V}{q_0^i}\d q_0^i \right. \\
&\qquad \qquad + \left[ p^0_A + \frac{\alpha}{((q_1^i)^2)^2\sqrt{g}}\left[ \left((q_1^i)^2(q_2^i)^2 - 2(q_1^iq_2^i)^2\right)q_1^A + (q_1^iq_2^i)(q_1^i)^2q_2^A \right] \right]\d q_1^A \\
&\qquad \qquad + \left. \left[ p_A^1 - \frac{\alpha}{(q_1^i)^2\sqrt{g}}\left( (q^i_1)^2 q_2^A - (q_1^iq_2^i)q_1^A \right) \right]\d q_2^A \right) \wedge \d t \, .
\end{align*}

In order to state the Lagrangian-Hamiltonian problem for sections
for this second-order system, let $Y \in \vf(\W_o)$ be a generic vector
field locally given by
$$
Y = f\derpar{}{t} + f_0^A \derpar{}{q_0^A} + f_1^A \derpar{}{q_1^A} + F_2^A \derpar{}{q_2^A} + 
F_3^A \derpar{}{q_3^A} + G_A^0 \derpar{}{p_A^0} + G_A^1 \derpar{}{p_A^1} \ .
$$
Now, if $\psi_o(t) = (t,q_0^i(t),q_1^i(t),q_2^i(t),q_3^i(t),p_i^0(t),p_i^1(t))$
is a holonomic section of the projection $\rho_\R^o$, equation
\eqref{eqn:DynEquationSections} leads to the following $5n$ equations
(the redundant equation \eqref{eqn:LocalCoordRedundantEqSections} is omitted):
\begin{align}
\label{eqn:Example2_Holonomy}
&\dot{q}_0^A = q_1^A \quad , \quad \dot{q}_1^A = q_2^A \\
\label{eqn:Example2_DiffEquations}
&\dot{p}_A^0 = \derpar{V}{q_0^A} \quad , \quad \dot{p}_A^1 = - p^0_A - 
\frac{\alpha}{((q_1^i)^2)^2\sqrt{g}}\left[ \left((q_1^i)^2(q_2^i)^2 - 2(q_1^iq_2^i)^2\right)q_1^A + 
(q_1^iq_2^i)(q_1^i)^2q_2^A \right] \\
\label{eqn:Example2_LastMomentumCoord1}
&p_A^1 = \frac{\alpha}{(q_1^i)^2\sqrt{g}}\left( (q^i_1)^2 q_2^A - (q_1^iq_2^i)q_1^A \right)
\end{align}
Equations \eqref{eqn:Example2_Holonomy} give the condition of holonomy
of type $2$ for the section $\psi_o$, which are also redundant since the
holonomy of $\psi_o$ is already assumed. Equations \eqref{eqn:Example2_LastMomentumCoord1}
are an algebraic condition, from which we conclude that the section $\psi_o$
must lie in a submanifold $\W_1$ that can be identified with the graph of
the extended Legendre-Ostrogradsky map, $\widetilde{\Leg}$.
The expression in natural coordinates
of this map $\widetilde{\Leg} \colon J^{3}\pi \to \Tan^*(J^{1}\pi)$
is obtained from Corollary \ref{corol:W1GraphExtendedFLSections}
and is
\begin{align}
\nonumber
&\Leg^*p^0_A = \frac{\alpha}{(q_1^i)^2\sqrt{g^3}} \left[ \left( (q_2^i)^2g + (q_1^i)^2(q_2^i)^2(q_1^iq_3^i) -
(q_1^i)^2(q_1^iq_2^i)(q_ 2^iq_3^i) \right)q_1^A\right] \\
\nonumber
&\qquad\qquad\quad + \frac{\alpha}{(q_1^i)^2\sqrt{g^3}} \left[ \left( ((q_1^i)^2)^2(q_2^iq_3^i) - (q_1^i)^2(q_1^iq_2^i)(q_1^iq_3^i) - 
(q_1^iq_2^i)g \right)q_2^A - (q_1^i)^2gq_3^A\right] \\ 
\label{eqn:Example2_ExtendedLegendreMap}
&\Leg^*p^1_A = \frac{\alpha}{(q^i_1)^2\sqrt{g}} \left[ (q^i_1)^2q_2^A - (q^i_1q^i_2)q^A_1 \right] \\
\nonumber
&\Leg^*p = \frac{\alpha}{(q_1^i)^2} \sqrt{g} + V(t,q_0^i)
- \left(\frac{\alpha}{(q_1^i)^2\sqrt{g^3}} \left[ \left( (q_2^i)^2g + (q_1^i)^2(q_2^i)^2(q_1^iq_3^i) -
(q_1^i)^2(q_1^iq_2^i)(q_ 2^iq_3^i) \right)q_1^A\right] \right. \\
\nonumber
&\qquad\qquad\quad \left .+ \frac{\alpha}{(q_1^i)^2\sqrt{g^3}} \left[ \left( ((q_1^i)^2)^2(q_2^iq_3^i) - (q_1^i)^2(q_1^iq_2^i)(q_1^iq_3^i) - (q_1^iq_2^i)g \right)q_2^A - (q_1^i)^2gq_3^A\right] \right)q_1^A \\
\nonumber
&\qquad \qquad \quad - \frac{\alpha}{(q^i_1)^2\sqrt{g}} \left[ (q^i_1)^2q_2^A - (q^i_1q^i_2)q^A_1 \right]q_2^A
\end{align}

Hence, the section $\psi_o \in \Gamma(\rho_\R^o)$ is holonomic and
lies in the submanifold $\W_1 \hookrightarrow \W_o$ defined by the
constraint functions given above, and its last components satisfy the
$2n$ differential equations
$$
\dot{p}_A^0 = \derpar{V}{q_0^A} \quad , \quad \dot{p}_A^1 = - p^0_A - 
\frac{\alpha}{((q_1^i)^2)^2\sqrt{g}}\left[ \left((q_1^i)^2(q_2^i)^2 - 2(q_1^iq_2^i)^2\right)q_1^A + 
(q_1^iq_2^i)(q_1^i)^2q_2^A \right] \, .
$$

Now we state the Lagrangian-Hamiltonian problem for vector fields, that is,
we wish to find a vector field $X_o \in \vf(\W_o)$ solution
to equations \eqref{eqn:DynEquationVectorFields}. If the vector field $X_o$
is locally given by
$$
X_o = f\derpar{}{t} + f_0^A \derpar{}{q_0^A} + f_1^A \derpar{}{q_1^A} + F_2^A \derpar{}{q_2^A} + 
F_3^A \derpar{}{q_3^A} + G_A^0 \derpar{}{p_A^0} + G_A^1 \derpar{}{p_A^1} \ .
$$
The equations \eqref{eqn:DynEquationVectorFields} lead
to the following $5n+1$ equations (the redundant equation
\eqref{eqn:LocalCoordRedundantEqVectorFields} is omitted):
\begin{align}
\label{eqn:Example2_Semispray1}
&f_0^A = f\cdot q_1^A \quad ; \quad f_1^A = f\cdot q_2^A \\
\label{eqn:Example2_DynamicalEquations}
&G_A^0 = \derpar{V}{q_0^A} \quad ; \quad G_A^1 = - p^0_A - 
\frac{\alpha}{((q_1^i)^2)^2\sqrt{g}}\left[ \left((q_1^i)^2(q_2^i)^2 - 2(q_1^iq_2^i)^2\right)q_1^A + 
(q_1^iq_2^i)(q_1^i)^2q_2^A \right] \\
\label{eqn:Example2_FixingGauge}
&f = 1 \\
\label{eqn:Example2_LastMomentumCoord2}
&f\cdot \left( p_A^1 - \frac{\alpha}{(q_1^i)^2\sqrt{g}}\left( (q^i_1)^2 q_2^A - (q_1^iq_2^i)q_1^A \right) \right) = 0
\end{align}
From equations \eqref{eqn:Example2_Semispray1} we obtain the condition
of semispray of type $2$ for the vector field $X_o$. In addition, equations
\eqref{eqn:Example2_LastMomentumCoord2} are algebraic relations
between the coordinates in $\W_o$ which give, in coordinates, the result
stated in Propositions \ref{prop:ExistSolDynEq} and \ref{prop:W1GraphFLVectorFields},
that is, the vector field $X_o$ is defined along a submanifold
$\W_1$ which we identify with the graph of the extended Legendre-Ostrogradsky.
Thus, using \eqref{eqn:Example2_Semispray1}, \eqref{eqn:Example2_DynamicalEquations}
and \eqref{eqn:Example2_FixingGauge}, the vector field $X_o$ is given locally by
\begin{equation}
\label{eqn:Example2_VectorFieldBeforeTangency}
X_o = \derpar{}{t} + q_1^A \derpar{}{q_0^A} + q_2^A \derpar{}{q_1^A} + F_2^A \derpar{}{q_2^A} + 
F_3^A \derpar{}{q_3^A} + \derpar{V}{q_0^A} \derpar{}{p_A^0} + G_A^1 \derpar{}{p_A^1} \ ,
\end{equation}
where the functions $G_A^1$ are determined by \eqref{eqn:Example2_DynamicalEquations}.
Since we wish to recover the solutions in the Lagrangian formalism from the vector
field $X_o$, we must require it to be a semispray of type $1$. This condition
reduces the set of vector fields $X_o \in \vf(\W_o)$ given by \eqref{eqn:Example2_VectorFieldBeforeTangency}
to the following ones
\begin{equation}
\label{eqn:Example2_SemisprayType1dBeforeTangency}
X_o = \derpar{}{t} + q_1^A \derpar{}{q_0^A} + q_2^A \derpar{}{q_1^A} + q_3^A \derpar{}{q_2^A} + 
F_3^A \derpar{}{q_3^A} + \derpar{V}{q_0^A} \derpar{}{p_A^0} + G_A^1 \derpar{}{p_A^1} \ .
\end{equation}
Notice that the functions $F_3^A$ are not determinated until
the tangency of the vector field $X_o$ on $\W_1$ is required.
From the expression in local coordinates \eqref{eqn:Example2_ExtendedLegendreMap}
of the map $\widetilde{\Leg}$, we obtain the primary constraints defining the
closed submanifold $\widetilde{\P} = {\rm Im}(\widetilde{\Leg}) \hookrightarrow \Tan^*(J^{1}\pi)$,
which are
\begin{equation}
\label{eqn:Example2_Constraints0}
\phi^{(0)}_1 \equiv p^1_iq_1^i = 0 \quad ; \quad
\phi^{(0)}_2 \equiv (p_i^1)^2 - \frac{\alpha^2}{(q^i_1)^2} = 0 \ ;
\end{equation}
Let $\Leg_o \colon J^{3}\pi \to \widetilde{\P}$. Then, the submanifold
$\W_1 = {\rm graph}{\Leg_o} = {\rm graph}{\widetilde{\Leg}}$ is defined by
$$
\W_1 = \left\{ w \in \W_o \colon \xi(w) = \xi_0^A(w) = \xi_1^A(w) = \phi_1^{(0)}(w) = \phi_2^{(1)}(w) = 0 \right\}
$$
where $\xi_r^A = p_A^r - \widetilde{\Leg}^*p_A^r$, $\xi = p - \widetilde{\Leg}^*p$.

Next, we compute the tangency condition for the vector field $X_o \in \vf(\W_o)$,
given locally by \eqref{eqn:Example2_SemisprayType1dBeforeTangency} on
the submanifold $\W_1 \hookrightarrow \W_o \hookrightarrow \W$,
by checking if the following identities hold
\begin{align}
\label{eqn:Example2_LagEquations1}
\restric{\Lie(X)\xi^A_0}{\W_1} = 0 \quad &; \quad \restric{\Lie(X)\xi^A_1}{\W_1} = 0 \\
\label{eqn:Example2_LieDerConstraints0}
\restric{\Lie(X)\phi^{(0)}_1}{\W_1} = 0  \quad &; \quad \restric{\Lie(X)\phi^{(0)}_2}{\W_1} = 0 \ .
\end{align}
As we have seen in Section \ref{subsubsection:DynEqVectFields}, equations
\eqref{eqn:Example2_LagEquations1} give us the Lagrangian equations for
the vector field $X_o$. However, equations \eqref{eqn:Example2_LieDerConstraints0}
do not hold, since
$$
\Lie(X_o)\phi^{(0)}_1 = \Lie(X_o)(p^1_iq_1^i) = - p_i^0q_1^i \quad , \quad
\Lie(X_o)\phi^{(0)}_2 = \Lie(X_o)((p^1_i)^2 - \alpha^2 / (q_1^i)^2) = - 2p_i^0p_i^1 \ ,
$$
and hence we obtain two first-generation secondary constraints
\begin{equation}
\label{eqn:Example2_Constraints1}
\phi^{(1)}_1 \equiv p_i^0 q_1^i = 0 \quad , \quad \phi^{(1)}_2 \equiv p_i^0p_i^1 = 0 
\end{equation}
that define a new submanifold $\W_2 \hookrightarrow \W_1$.
Now, by checking the tangency of the vector field $X_o$ to this
new submanifold, we obtain
$$
\Lie(X_o)\phi^{(1)}_1 = \Lie(X_o)(p_i^0q_1^i) = 0 \quad , \quad \Lie(X_o)\phi^{(1)}_2 = \Lie(X_o)(p_i^0p_i^1) = -(p^0_i)^2 \ ,
$$
and a second-generation secondary constraint appears,
\begin{equation}
\label{eqn:Example2_Constraints2}
\phi^{(2)} \equiv (p^0_i)^2 = 0 \ ,
\end{equation}
which defines a new submanifold $\W_3 \hookrightarrow \W_2$.
Finally, the tangency of the vector field $X_o$ on this submanifold
gives no new constraints, since
$$
\Lie(X_o)\phi^{(2)} = \Lie(X_o)((p^0_i)^2) = 0 \ .
$$
So we have two primary constraints \eqref{eqn:Example2_Constraints0},
two first-generation secondary constraints \eqref{eqn:Example2_Constraints1},
and a single second-generation secondary constraint \eqref{eqn:Example2_Constraints2}.
Notice that these five constraints only depend on $q_1^A$, $p^0_A$ and $p^1_A$, 
and so they are $\hat{\rho}_2^o$-projectable.

Notice that we still have to check \eqref{eqn:Example2_LagEquations1}.
As we have seen in Section \ref{subsubsection:DynEqVectFields}, we obtain the following equations
\begin{align}
\label{eqn:Example2_EulerLagrangeInitial}
&\left(F_3^B - d_T\left(q_3^B\right)\right)\derpars{\hat{L}}{q_2^B}{q_2^A} +  \derpar{\hat{L}}{q_0^A} - d_T\left(\derpar{\hat{L}}{q_1^A}\right) +  d_T^2\left(\derpar{\hat{L}}{q_2^A}\right) +  \left(F_2^B - q_3^B\right)d_T\left(\derpars{\hat{L}}{q_2^B}{q_2^A}\right) = 0 \\
\label{eqn:Example2_SemisprayType1Condition}
&\left(F_2^B - q_3^B\right)\derpars{\hat{L}}{q_2^B}{q_2^A} = 0 
\end{align}
Since we have already required the vector field $X_o$ to be a semispray of type $1$ in $\W_o$,
equations \eqref{eqn:Example2_SemisprayType1Condition} are satisfied identically
and equations \eqref{eqn:Example2_EulerLagrangeInitial} become
\begin{equation}
\label{eqn:Example2_EulerLagrangeFinal}
\left(F_3^B - d_T\left(q_3^B\right)\right) \derpars{\hat{L}}{q_2^B}{q_2^A} +
\derpar{\hat{L}}{q_0^A} - d_T\left(\derpar{\hat{L}}{q_1^A}\right) + d_T^2\left(\derpar{\hat{L}}{q_2^A}\right) = 0 \ .
\end{equation}
A long calculation shows that this equation is compatible if, and only if,
$\derpar{V}{q_0^A} = 0$, for
$1 \leqslant A \leqslant n$. That is, we have $n$ first-generation secondary
constraints arising from the tangency
condition of $X_o$ along $\W_1$, thus defining a new submanifold
$\W_4 \hookrightarrow \W_3$ with constraint functions
$$
\phi^{(1)}_{3,A} \equiv \derpar{V}{q_0^A} = 0 \quad \mbox{for } 1 \leqslant A \leqslant n \, .
$$
Observe that, since $V$ is a function that depends only on $t$ and $q_0^A$,
these new constraints also depend only on the coordinates $t$ and $q_0^A$,
and thus they are $\hat{\rho}_2^o$-projectable.
From a physical viewpoint, these constraints
mean that the dynamics of the particle can take place on every level set of
the potential with respect to the position coordinates.

Finally, we recover the Lagrangian and Hamiltonian dynamics from
the unified formalism. For the Lagrangian solutions,
using Proposition \ref{prop:LagrangianSections}, we know that
from the holonomic section $\psi_o \in \Gamma(\rho_\R^o)$ we can recover
a holonomic section $\psi_\Lag = \rho_1^o \circ \psi_o \in \Gamma(\bar{\pi}^3)$
solution to equation \eqref{eqn:LagrangianDynEqSections}. In particular, if
$\psi_o(t) = (t,q_0^i(t),q_1^i(t),q_2^i(t),q_3^i(t),p_i^0(t),p_i^1(t))$, then
$\psi_\Lag(t) = (t,q_0^i(t),q_1^i(t),q_2^i(t),q_3^i(t))$ is a holonomic section
solution to equations \eqref{eqn:Example2_DiffEquations}. Now, bearing in mind the
local expression \eqref{eqn:Example2_ExtendedLegendreMap} of the
extendend Legendre-Ostrogradsky map, equations \eqref{eqn:Example2_DiffEquations}
give the last $n$ equations of the holonomy condition for $\psi_\Lag$, which
are identically satisfied since the holonomy condition has been already required,
and the classical higher-order Euler-Lagrange equations
$$
\restric{\derpar{L}{q_0^A}}{\psi_\Lag} - \restric{\frac{\d}{\d t}\derpar{L}{q_1^A}}{\psi_\Lag}
+ \restric{\frac{\d^2}{\d t^2}\derpar{L}{q_2^A}}{\psi_\Lag} = 0 \, .
$$

For the Lagrangian vector field, from Lemma \ref{lemma:CorrepondenceXoXL} and
Theorem \ref{thm:Unified-LagrangianSolutions}, we can recover from the semispray
of type $1$ $X_o \in \vf(\W_o)$ a semispray of type $1$, $X_\Lag \in \vf(J^3\pi)$,
solution to equations \eqref{eqn:LagrangianDynEqVectorFieldsSupportMf}
(with $M_f = \rho_1^o(\W_4)$), and it is locally given by
$$
X_o = \derpar{}{t} + q_1^A \derpar{}{q_0^A} + q_2^A \derpar{}{q_1^A} + q_3^A \derpar{}{q_2^A} + 
F_3^A \derpar{}{q_3^A} \ ,
$$
where $F_3^A$ are the solutions of equations \eqref{eqn:Example2_EulerLagrangeFinal}.

One can check that, if the semispray condition is not required at the beginning
and we perform all this procedure with the vector field given by \eqref{eqn:Example2_VectorFieldBeforeTangency},
the final result is the same. This means that, in this case,
the semispray condition does not give any additional constraint.

Now, since $\Lag$ is an almost-regular Lagrangian density, for the Hamiltonian
dynamics we must use the results stated in Section \ref{subsection:HamiltonianSingularCase}
and recover the Hamiltonian solutions passing through the Lagrangian formalism.
For the Hamiltonian sections, by Proposition \ref{prop:HamiltonianSectionsSingular},
from a section $\psi_o \in \Gamma(\rho_\R^o)$ solution to equation \eqref{eqn:DynEquationSections},
we can recover a section $\psi_h = \Leg \circ \rho_1^o \circ \psi_o \in \Gamma(\bar{\tau}_o)$
solution to the equation \eqref{eqn:HamiltonianSingularDynEqSections}.

For the Hamiltonian vector fields, we know that there are semisprays of type $1$
$X_\Lag \in \vf(J^{3}\pi)$, solutions to equations \eqref{eqn:LagrangianDynEqVectorFieldsSupportMf},
which are $\Leg_o$-projectable on $P_4 = \hat{\rho}_2^o(\W_4)$, tangent to $P_4$ and solutions to the Hamilton
equation.

\section{Conclusions and outlook}
\label{section:outlook}

The objective of this work is to develop a complete and
detailed geometric framework 
for describing the Lagrangian and Hamiltonian formalisms of
higher-order non-autonomous mechanical systems,
and to give some applications of it.

Our approach to the problem consists in extending the
Lagrangian-Hamiltonian unified formalism of Skinner and Rusk
to this case, starting from the generalization of this formalism
previously made for first-order non-autonomous dynamical systems
\cite{art:Barbero_Echeverria_Martin_Munoz_Roman08} and
higher-order autonomous mechanical systems
\cite{art:Prieto_Roman11}.
This enables us to derive
the Lagrangian and the Hamiltonian formalisms for these kinds of systems (in a natural way).
We pay special attention to describing the equations 
of motion in several equivalent ways, using sections and vector fields
in the bundles that constitute the phase spaces
of these systems, and showing how the equivalence between
the Lagrangian and the Hamiltonian formalisms is stated through the
Legendre-Ostrogradsky map, which is also obtained in a natural
way from the unified formalism.
Our analysis is performed both for regular and singular systems.

As applications of our formalism, we study
two physical examples:
a regular system describing the shape of a deformed elastic cylindrical
beam with fixed ends, and a singular system 
describing a relativistic particle
subjected to a generic potential depending on time and positions.

The background geometrical tools that we use are
higher-order jet bundles, in general, rather than the simpler
and more usual trivial bundles
(this particular case is also analyzed in the work),
since our aim is for this geometric framework to serve as
a guideline towards a geometric model for higher-order field theories,
which is free of the ambiguities present in their standard 
geometrical descriptions
(concerning the definition of the Poincar\'e-Cartan forms
and the Legendre transformation).
Some advances on this subject have been already obtained
\cite{art:Campos_DeLeon_Martin_Vankerschaver09,art:Vitagliano10},
and we trust that our future work will contribute to completing them.


\appendix
\section{A particular situation: trivial bundles}
\label{section:particular_situation}

Assume that the bundle $E \stackrel{\pi}{\longrightarrow}\R$
is trivial; that is, $E = \R \times Q$,
where $Q$ is a $n$-dimensional manifold.
In this case, we have that $J^{k}\pi \cong \R \times \Tan^{k}Q$,
where $\Tan^{k}Q$ is the $k$th order tangent bundle of $Q$
(see \cite{book:DeLeon_Rodrigues85} for details). The natural coordinates
in this case are defined in the same way as in the general case,
and are denoted by $(t,q_0^A,q_1^A,\ldots,q_k^A)$.
In this case, the bundles involved in the construction are
$$
J^{2k-1}\pi \cong \R \times \Tan^{2k-1} Q \quad , \quad
\Tan^*(J^{k-1}\pi) \cong \R \times \R^* \times \Tan^*(\Tan^{k-1}Q) \quad , \quad
J^{k-1}\pi^* \cong \R \times \Tan^*(\Tan^{k-1}Q)
$$
Thus, the higher-order restricted jet-momentum bundle is
\begin{align*}
\W_r = J^{2k-1}\pi \times_{J^{k-1}\pi} J^{k-1}\pi^* &\cong
(\R \times \Tan^{2k-1} Q) \times_{\R \times \Tan^{k-1}Q} (\R \times \Tan^*(\Tan^{k-1}Q)) \\
& \cong \R \times \Tan^{2k-1}Q \times_{\Tan^{k-1}Q} \Tan^*(\Tan^{k-1}Q) \cong \R \times \W_a \, ,
\end{align*}
where $\W_a = \Tan^{2k-1}Q \times_{\Tan^{k-1}Q} \Tan^*(\Tan^{k-1}Q)$
denotes the unified phase space in the autonomous formalism.
Natural coordinates in this bundle are the same as in the
non-autonomous case, that is,
$(t,q_0^A,q_1^A,\ldots,q_{2k-1}^A,p_A^0,p_A^1,\ldots,p_A^{k-1})$.

\textbf{Remark}: As we will see, in this particular situation the
extended jet-momentum bundle is not needed.
Thus, we denote $\W_r$ simply by $\W$ in this section.
The differential forms $\Theta_r$ and $\Omega_r$ (or, equivalently,
$\Theta_o$ and $\Omega_o$) are also denoted by $\Theta$ and $\Omega$,
respectively.

We have the following diagram
$$
\xymatrix{
\W \ar[rr]^-{\rho_{\W_a}} \ar[dd]_-{\rho_\R} & \ & \W_a \ar[dl]_{\pr_1} \ar[dr]^{\pr_2} & \ \\
\ & \Tan^{2k-1}Q \ar[dr]_{\rho^{2k-1}_{k-1}} & \ & \Tan^*(\Tan^{k-1}Q) \ar[dl]^{\pi_{\Tan^{k-1}Q}} \\
\R & \ & \Tan^{k-1}Q \ar[d]^{\beta^{k-1}} & \ \\
\ & \ & Q & \ \\
}
$$
where all the maps are the natural projections
(see \cite{art:Prieto_Roman11} for details). In coordinates, we have
\begin{align*}
&\rho_\R(t,q_i^A,q_j^A,p_A^i) = t \quad , \quad \rho_{\W_a}(t,q_i^A,q_j^A,p_A^i) = (q_i^A,q_j^A,p_A^i) \\
&\pr_1(q_i^A,q_j^A,p_A^i) = (q_i^A,q_j^A) \quad , \quad \pr_2(q_i^A,q_j^A,p_A^i) = (q_i^A,p_A^i)
\end{align*}

Now we see how to construct the canonical structures in $\W$,
described previously, from the canonical structures in $\W_a$.
Let $\theta_a \in \df^{1}(\W_a)$ and $\Omega_a \in\df^{2}(\W_a)$ be
the canonical forms on $\W_a$ defined as
$$
\theta_a = \pr_2^*\theta_{k-1} \quad , \quad
\Omega_a = \pr_2^*\omega_{k-1} = -\d\theta_a \, ,
$$
where $\theta_{k-1}$ and $\omega_{k-1}$ are the canonical
$1$ and $2$ forms on the cotangent bundle $\Tan^*(\Tan^{k-1}Q)$.

As stated before, the dynamics of the system is described by a
Lagrangian density $\Lag \in \df^{1}(\R \times \Tan^kQ)$, with
associated Lagrangian function $L \in \Cinfty(\R \times \Tan^{k}Q)$.
Then, if $C \in \Cinfty(\W_a)$ is the coupling function in the
autonomous formalism \cite{art:Prieto_Roman11}, we can construct a
globally defined Hamiltonian function
in the following way:
$$
H = \rho_{\W_a}^*C - L \, .
$$
Then, the forms $\Theta \in \df^{1}(\W)$ and $\Omega \in \df^{2}(\W)$
can be constructed as follows
$$
\Theta = \rho_{\W_a}^*\theta_a - H\rho_\R^*\eta \quad , \quad
\Omega = -\d\Theta = \rho_{\W_a}^*\Omega_a + \d H \wedge \rho_\R^*\eta \ .
$$

In local coordinates, bearing in mind the local expressions of $\theta_{k-1}$,
$\omega_{k-1}$ and $\C$:
$$
\theta_{k-1} = p_A^i\d q_i^A \quad , \quad
\omega_{k-1} = \d q_i^A \wedge \d p_A^i \quad , \quad
\C = p_A^iq_{i+1}^A \, ,
$$
we have that the local expression for the forms $\Theta$ and $\Omega$ are
$$
\Theta = p_A^i\d q_i^A - (p_A^iq_{i+1}^A - L)\d t \quad , \quad
\Omega = \d q_i^A \wedge \d p_A^i + \d(p_A^iq_{i+1}^A - L) \wedge \d t \, ;
$$
that is, we obtain the local expressions given in \eqref{eqn:LocalCoordFormsWo}.

The dynamical equations for sections and vector fields are now stated
as in Section \ref{section:SkinnerRuskform}, and the local expressions are
the same. There is only one difference: in Proposition \ref{prop:ExistSolDynEq}
a connection in $\W$ is needed in order to split the presymplectic form $\Omega$
into the sum of a $2$-form with
the wedge product of two $1$-forms (which are the differential of the
local Hamiltonian function, and the volume form in $\R$).
In this case, we do not need to use such a connection, since the bundles
are trivial and this splitting is natural.


\section*{Acknowledgments}

We acknowledge the financial support of the  {\sl
Ministerio de Ciencia e Innovaci\'on} (Spain), projects
MTM2008-00689, MTM2010-12116-E, MTM2011-22585,
and AGAUR, project 2009 SGR:1338.
One of us (PDPM) wants to thank the UPC for a Ph.D grant.
Special thanks to Prof. X. Gr\`acia for drawing our attention on
the example studied in Section \ref{subsection:CylindricalBeam},
and to Elisa Guzm\'an for her helpful comments.
We also thank Mr. Jeff Palmer for his assistance in preparing the English
version of the manuscript.


{\small
\bibliography{Bibliografia}

\providecommand{\bysame}{\leavevmode\hbox to3em{\hrulefill}\thinspace}
\providecommand{\MR}{\relax\ifhmode\unskip\space\fi MR }
\providecommand{\MRhref}[2]{%
  \href{http://www.ams.org/mathscinet-getitem?mr=#1}{#2}
}
\providecommand{\href}[2]{#2}
\begin{thebibliography}{10}

\bibitem{art:Aldaya_Azcarraga78_2}
V.~{Aldaya} and J.A. {de Azcarraga}, ``Variational {Principles} on r-th order
  jets of fibre bundles in {Field} {Theory}'', \textit{J. Math. Phys.}
  \textbf{19}(9)  (1978) 1869.

\bibitem{art:Aldaya_Azcarraga80}
V.~{Aldaya} and J.A. {de Azcarraga}, ``Higher order {Hamiltonian} formalism in
  {Field} {Theory}'', \textit{J. Phys. A} \textbf{13}(8)  (1980) 2545--2551.

\bibitem{art:Banerjee_Mukherjee_Paul10}
R.~{Banerjee}, P.~{Mukherjee}, and B.~{Paul}, ``Gauge symmetry and {W}-algebra
  in higher derivative systems'', \textit{JHEP} \textbf{08} (2011) 085.

\bibitem{art:Barbero_Echeverria_Martin_Munoz_Roman07}
M.~{Barbero Li\~{n}\'{a}n}, A.~{Echeverr\'{\i}a-Enr\'{\i}quez}, D.~{Mart\'{\i}n
  de Diego}, M.C. {Mu\~{n}oz-Lecanda}, and N.~{Rom\'an Roy}, ``Skinner-rusk
  unified formalism for optimal control systems and applications'', \textit{J.
  Math. Phys. A: Math. Theor.} \textbf{40}(40)  (2007) 12071--12093.

\bibitem{art:Barbero_Echeverria_Martin_Munoz_Roman08}
M.~{Barbero Li\~{n}\'{a}n}, A.~{Echeverr\'{\i}a-Enr\'{\i}quez}, D.~{Mart\'{\i}n
  de Diego}, M.C. {Mu\~{n}oz-Lecanda}, and N.~{Rom\'an-Roy}, ``Unified
  formalism for non-autonomous mechanical systems'', \textit{J. Math. Phys.}
  \textbf{49}(6)  (2008) 062902.

\bibitem{art:Barcelos_Natividade91_2}
J.~{Barcelos-Neto} and C.P. {Natividade}, ``Hamiltonian {Path} {Integral}
  {Formalism} with {Higher} {Derivatives}'', \textit{Zeits. Phys. C: Parts.
  Fields} \textbf{51}(2)  (1991) 313--319.

\bibitem{art:Batlle_Gomis_Pons_Roman88}
C.~{Batlle}, J.~{Gomis}, J.M. {Pons}, and N.~{Rom\'{a}n-Roy}, ``Lagrangian and
  {Hamiltonian} constraints for second-order singular {Lagrangians}'',
  \textit{J. Phys. A: Math. Gen.} \textbf{21}(12)  (1988) 2693--2703.

\bibitem{art:Belvedere_Amaral_Lemos_Carvalhaes00}
L.V. {Belvedere}, R.L.P.L. {Amaral}, N.A. {Lemos}, and C.G. {Carvalhaes},
  ``Higher-{Derivative} 2 {Dimensional} {Massive} {Fermion} {Theories}'',
  \textit{Int. J. Mod. Phys A} \textbf{15}(15)  (2000) 2237--2254.

\bibitem{book:Benson06}
D.J. {Benson}, \emph{Music: {A} {Mathematical} {Offering}}, , Cambridge
  University Press,  2006.

\bibitem{art:Campos_DeLeon_Martin_Vankerschaver09}
C.M. {Campos}, M.~{de Le\'on}, D.~{Mart\'{\i}n de Diego}, and K.~Vankerschaver,
  ``Unambigous formalism for higher-order lagrangian field theories'',
  \textit{J. Phys A: Math Theor.} \textbf{42} (2009) 475207.

\bibitem{proc:Cantrijn_Crampin_Sarlet86}
F.~{Cantrijn}, M.~{Crampin}, and W.~{Sarlet}, ``Higher-order differential
  equations and higher-order {Lagrangian} mechanics'', \textit{Math. Proc.
  Cambridge Philos. Soc.} \textbf{99}(3)  (1986) 565--587.

\bibitem{art:Carinena_Lopez92}
J.F. {Cari\~{n}ena} and C.~{L\'{o}pez}, ``The time-evolution operator for
  higher-order singular {Lagrangians}'', \textit{Internat. J. Modern Phys. A}
  \textbf{7}(11)  (1992) 2447--2468.

\bibitem{art:Chinea_deLeon_Marrero94}
D.~{Chinea}, M.~{de Le\'{o}n}, and J.C. {Marrero}, ``The constraint algorithm
  for time-dependent {Lagrangians}'', \textit{J. Math. Phys.} \textbf{35}(7)
  (1994) 3410--3447.

\bibitem{art:Cortes_Martinez_Cantrijn02}
J.~{Cort\'{e}s}, S.~{Mart\'{\i}nez}, and F.~{Cantrijn}, ``Skinner-{Rusk}
  approach to time-dependent mechanics'', \textit{Physics Letters A}
  \textbf{300} (2002) 250--258.

\bibitem{art:Crasmareanu00}
M.~{Cr\^{a}\c{s}m\v{a}reanu}, ``Noether theorem for time-dependent higher order
  {Lagrangians}'', \textit{Proc. Rom. Acad. Ser. A Math. Phys. Tech. Sci. Inf.
  Sci.} \textbf{1}(2)  (2000) 77--79.

\bibitem{unpub:Crampin_Saunders11}
M.~{Crampin} and D.J. {Saunders}, ``Homogeneity and projective equivalence of
  differential equation fields'', {arXiv}:1109.3640 [math.DG],  2011.

\bibitem{art:DeLeon_Marin_Marrero96}
M.~{de Le\'on}, J.~{Mar\'{\i}n Solano}, and J.C. {Marrero}, ``The constraint
  algorithm in the jet formalism'', \textit{Diff. Geom. Appl.} \textbf{6}(3)
  (1996) 275--300.

\bibitem{art:deLeon_Marin_Marrero_Munoz_Roman02}
M.~{de Le\'on}, J.~{Mar\'{\i}n-Solano}, J.C. {Marrero}, M.C.
  {Mu\~{n}oz-Lecanda}, and N.~{Rom\'{a}n-Roy}, ``Singular {Lagrangian} systems
  on jet bundles'', \textit{Fortschr. Phys.} \textbf{50} (2002) 105--169.

\bibitem{proc:DeLeon_Marrero92}
M.~{de Le\'on} and J.C. {Marrero}, ``Degenerate time-dependent {Lagrangians} of
  second order: the fourth order differential equation problem'', Proc. Conf.
  Opava, Silesian Univ. Opava\textit{Math. Publ.} \textbf{1} (1992) 497--508.

\bibitem{art:deLeon_Martin94}
M.~{de Le\'{o}n} and D.~{Mart\'{\i}n de Diego}, ``Classification of symmetries
  for higher-order {Lagrangian} systems {II}: the non-autonomous case'',
  \textit{Extracta Mathematicae} \textbf{9}(2)  (1994) 111--114.

\bibitem{art:DeLeon_Martin_Santamaria04}
M.~{de Le\'on}, D.~{Mart\'{\i}n de Diego}, and A.~{Santamar\'{\i}a Merino},
  ``Geometric numerical integration of nonholonomic systems and optimal control
  problems'', \textit{Eur. J. Control} \textbf{10}(5)  (2004) 515--524.

\bibitem{book:DeLeon_Rodrigues85}
M.~{de Le\'on} and P.R. {Rodrigues}, \emph{Generalized classical mechanics and
  field theory}, North-Holland Math. Studies, vol. 112, Elsevier Science
  Publishers B.V., Amsterdam 1985.

\bibitem{proc:deLeon_Rodrigues87}
M.~{de Le\'{o}n} and P.R. {Rodrigues}, ``Higher-order almost tangent geometry
  and non-autonomous {Lagrangian} dynamics'', Proc. Winter School on Geometry
  and Physics (Srn\'{\i},1987),\textit{Rend. Circ. Mat. Palermo} \textbf{2}(16)
   (1987) 157--171.

\bibitem{art:Echeverria_Lopez_Marin_Munoz_Roman04}
A.~{Echeverr\'{\i}a-Enr\'{\i}quez}, C.~{L\'{o}pez}, J.~{Mar\'{\i}n-Solano},
  M.C. {Mu\~{n}oz-Lecanda}, and N.~{Rom\'{a}n-Roy}, ``Lagrangian-{Hamiltonian}
  unified formalism for field theory'', \textit{J. Math. Phys.} \textbf{45}(1)
  (2004) 360--380.

\bibitem{book:Elsgoltz83}
L.~{Elsgoltz}, \emph{Differential equations and the calculus of variations},
  3rd ed., Mir Publishers,  1983.

\bibitem{proc:Garcia_Munoz83}
P.L. {Garc\'{\i}a} and J.~{Mu\~{n}oz}, ``On the geometrical structure of higher
  order variational calculus'', Procs. IUTAM-ISIMM Symposium on Modern
  Developments in Analytical Mechanics, Vol. I (Torino, 1982). \,\textit{Atti.
  Accad. Sci. Torino Cl. Sci. Fis. Math. Natur.} \textbf{117} (1983) suppl. 1,
  127--147.

\bibitem{art:Govaerts92}
J.~{Govaerts}, ``A {Quantum} {Anomaly} for {Rigid} {Particles}'', \textit{Phys.
  Lett. B} \textbf{293}(3-4)  (1992) 327--334.

\bibitem{art:Gracia_Pons_Roman91}
X.~{Gr\`{a}cia}, J.M. {Pons}, and N.~{Rom\'{a}n-Roy}, ``Higher-order
  {Lagrangian} systems: {Geometric} structures, dynamics and constraints'',
  \textit{J. Math. Phys.} \textbf{32}(10)  (1991) 2744--2763.

\bibitem{art:Gracia_Pons_Roman92}
X.~{Gr\`{a}cia}, J.M. {Pons}, and N.~{Rom\'{a}n-Roy}, ``Higher-order conditions
  for singular {Lagrangian} systems'', \textit{J. Phys. A: Math. Gen.}
  \textbf{25} (1992) 1981--2004.

\bibitem{art:Krupkova96}
O.~{Krupkova}, ``Symmetries and first integrals of time-dependent higher-order
  constrained systems'', \textit{J. Geom. Phys.} \textbf{18}(1)  (1996) 38--58.

\bibitem{art:Krupkova00}
O.~{Krupkova}, ``Higher-order mechanical systems with constraints'', \textit{J.
  Math. Phys.} \textbf{41}(8)  (2000) 5304.5324.

\bibitem{art:Kuznetsov_Plyushchay94}
Y.A. {Kuznetsov} and M.S. {Plyushchay}, ``(2+1)-dimensional {Models} of
  {Relativistic}-{Particles} with {Curvature} and {Torsion}'', \textit{J. Math.
  Phys.} \textbf{35}(6)  (1994) 2772--2784.

\bibitem{unpub:Mukherjee_Paul11}
P.~{Mukherjee} and B.~{Paul}, ``Gauge invariances of higher derivative
  {Maxwell}-{Chern}-{Simons} field theories - a new {Hamiltonian} approach'',
  {arXiv}:1111.0153v1 [hep-th],  2011.

\bibitem{art:Nesterenko89}
V.V. {Nesterenko}, ``Singular {Lagrangians} with higher-order derivatives'',
  \textit{J. Phys. A: Math. Gen.} \textbf{22}(10)  (1989) 1673--1687.

\bibitem{art:Pisarski86}
R.D. {Pisarski}, ``Field theory of paths with a curvature-dependent term'',
  \textit{Phys. Rev. D} \textbf{34}(2)  (1986) 670--673.

\bibitem{art:Plyushchay88}
M.S. {Plyushchay}, ``Canonical quantization and mass spectrum of relativistic
  particle: analog of relativistic string with rigidity'', \textit{Mod. Phys.
  Lett. A} \textbf{3}(13)  (1988) 1299--1308.

\bibitem{art:Plyushchay91}
M.S. {Plyushchay}, ``The model of relativistic particle with torsion'',
  \textit{Nucl. Phys. B} \textbf{362} (1991) 54--72.

\bibitem{art:Popescu_Popescu07}
P.~{Popescu} and M.~{Popescu}, ``Affine {Hamiltonians} in {Higher} {Order}
  {Geometry}'', \textit{Int. J. Theor. Phys.} \textbf{46} (2007) 2531--2549.

\bibitem{art:Prieto_Roman11}
P.D. {Prieto-Mart\'{\i}nez} and N.~{Rom\'{a}n-Roy},
  ``{Lagrangian}-{Hamiltonian} unified formalism for autonomous higher-order
  dynamical systems'', \textit{J. Phys. A: Math. Teor.} \textbf{44}(38)  (2011)
  385203.

\bibitem{art:Ramos_Roca95}
E.~{Ramos} and J.~{Roca}, ``W-symmetry and the {Rigid} {Particle}'',
  \textit{Nuc. Phys. B} \textbf{436}(3)  (1995) 529--541.

\bibitem{art:Roman09}
N.~{Rom\'{a}n-Roy}, ``Multisymplectic {Lagrangian} and {Hamiltonian} formalisms
  of classical field theories'', \textit{Symmetry Integrability Geom. Methods
  Appl. (SIGMA)} \textbf{5} (2009) Paper 100, 25 pp.

\bibitem{art:Saunders87}
D.J. {Saunders}, ``An alternative approach to the {Cartan} form in {Lagrangian}
  field theories'', \textit{J. Phys. A: Math. Gen.} \textbf{20} (1987)
  339--349.

\bibitem{book:Saunders89}
D.J. {Saunders}, \emph{The geometry of jet bundles}, London Mathematical
  Society, Lecture notes series, vol. 142, Cambridge University Press,
  Cambridge, New York 1989.

\bibitem{art:Saunders_Crampin90}
D.J. {Saunders} and M.~{Crampin}, ``On the {Legendre} map in higher-order field
  theories'', \textit{J. Phys. A: Math. Gen.} \textbf{23} (1990) 3169--3182.

\bibitem{art:Skinner_Rusk83}
R.~{Skinner} and R.~{Rusk}, ``Generalized {Hamiltonian} dynamics. {I}.
  {Formulation} on ${T^*Q} \oplus {TQ}$'', \textit{J. Math. Phys.}
  \textbf{24}(11)  (1983) 2589--2594.

\bibitem{art:Vitagliano10}
L.~{Vitagliano}, ``The {Lagrangian}-{Hamiltonian} {Formalism} for {Higher}
  {Order} {Field} {Theories}'', \textit{J. Geom. Phys.} \textbf{60} (2010)
  857--873.

\bibitem{art:Zloshchastiev00}
K.G. {Zloshchastiev}, ``Field-to-{Particle} {Transition} based on the
  zero-brane approach to quantization of multiscalar field theories and its
  applications for {Jackiw}-{Teitelboim} gravity'', \textit{Phys. Rev. D}
  \textbf{61}(12)  (2000) 5017.

\end{thebibliography}
\bibliographystyle{AMS_Mod}
}

\end{document}